\documentclass[12pt]{article}
\pdfoutput=1
\usepackage[margin=1in]{geometry}
\usepackage[utf8]{inputenc}
\usepackage{natbib}
\usepackage{amsmath, xfrac}
\usepackage{amsfonts}
\usepackage{amssymb}
\usepackage{amsthm}
\usepackage[hidelinks]{hyperref}
\usepackage[capitalize]{cleveref}
\usepackage{bbm}
\usepackage{enumerate}
\usepackage{mathtools}
\usepackage[dvipsnames]{xcolor}

\usepackage{tikz}
\usepackage{subfig}
\usepackage{graphicx}
\usepackage{setspace}
\newtheorem{theorem}{Theorem}
\newtheorem{lemma}{Lemma}[section]

\newtheorem{claim}[lemma]{Claim}

\newtheorem{corollary}{Corollary}

\newtheorem{proposition}{Proposition}
\newtheorem{definition}{Definition}

\crefformat{scoringruleprogram}{#2program~(#1)#3}
\Crefformat{scoringruleprogram}{#2Program~(#1)#3}

\newcommand{\state}{\theta}
\newcommand{\statespace}{\Theta}
\newcommand{\scorebound}{B}

\newcommand{\themean}{\mu}
\newcommand{\mean}[1]{\themean_{#1}}

\newcommand{\statedist}{D}

\newcommand{\report}{r}
\newcommand{\reportspace}{R}

\newcommand{\marginalReport}{F}

\newcommand{\alloc}{x}
\newcommand{\pay}{p}

\newcommand{\priorMeanState}{\mean{\statedist}}
\newcommand{\Pmeani}{\mean{\statedist_i}}

\newcommand{\priorMeanMean}{\mean{\marginalReport}}

\newcommand{\priorMean}{\priorMeanState}

\newcommand{\utility}{u}
\newcommand{\util}{u}

\newcommand{\reals}{\mathbb{R}}
\newcommand{\expect}[2]{{\mathbf{E}}_{#1}\left[#2\right]}

\newcommand{\prob}[2]{{\mathbf{Pr}}_{#1}\left[#2\right]}
\newcommand{\indicate}[1]{\mathbf{1}\!\left[#1\right]}

\newcommand{\numasm}{n}

\newcommand{\mos}{{\rm max-over-separate }}

\newcommand{\argmax}{\arg\max}
\newcommand{\score}{S}

\newcommand{\objval}{c}
\newcommand{\distset}{\mathcal{F}_{\objval}}
\newcommand{\quadra}{S_q}
\newcommand{\OPT}{{\rm OPT}}

\newcommand{\de}{\mathrm{d}}
\newcommand{\objfunc}{{\rm Obj}}

\newcommand{\distoverposterior}{F}
\newcommand{\posterior}{G}
\newcommand{\utilquad}{\util_q}

\newcommand{\conv}{\mathrm{conv}}

\newcommand{\centerpoint}{C}
\newcommand{\relint}{\mathrm{relint}}
\newcommand{\sg}{\xi}

\newcommand{\norm}[1]{\left\lVert#1\right\rVert}
\newcommand{\abs}[1]{\left|#1\right|}
\newcommand{\boundaryReport}{\partial \reportspace}
\newcommand{\interiorReport}{\relint(\reportspace)}
\newcommand{\subgradient}{\nabla}
\newcommand{\reportFull}{\posterior}
\newcommand{\car}{{\rm choose-and-report }}
\newcommand{\plane}{P}
\newcommand{\sgset}{\mathcal{G}}

\newcommand{\statej}{\state^{(j)}}

\newcommand{\OPTdenom}[1][\priorMean]{\max(#1,1-#1)}

\newcommand{\maxzero}[1]{\max(#1,0)}

\newcommand{\sdscore}{{\hat{\score}}}

\newcommand{\sdutil}{{\hat{\util}}}
\newcommand{\sdutili}[1][i]{\sdutil_{#1}}
\newcommand{\sdkappa}{{\hat{\kappa}}}

\newcommand{\sdsg}{{\hat{\sg}}}

\newcommand{\extMarginalReport}{\tilde{\marginalReport}}
\newcommand{\extreportspace}{\widetilde{\reportspace}}
\newcommand{\extstatespace}{\widetilde{\statespace}}
\newcommand{\extutil}{\tilde{\util}}

\newcommand{\extreputil}{\hat{\util}}

\newcommand{\extsymstatespace}{\extstatespace'}
\newcommand{\extconvstatespace}{\extstatespace''}

\newcommand{\OPTUB}{\widetilde{\OPT}}
\newcommand{\ubc}{\tilde{\objval}}
\newcommand{\distsetUB}{\mathcal{F}_{\ubc}}

\newcommand{\inftynorm}[1]{\lVert#1\rVert_{\infty}}

\newcommand{\prior}{\statedist}
\newcommand{\dd}{\,{\rm d}}

\newcommand{\rbr}[1]{\left(#1\right)}

\setcitestyle{authoryear,open={(},close={)}}
\title{Optimization of Scoring Rules\thanks{This work was supported by NSF CCF-1733860. 
Liren Shan was supported by NSF CCF-1955351 and CCF-1934931. 
Yingkai Li also thanks NSF SES-1947021 for financial support. 
A two-page abstract of this paper has appeared in the 23rd ACM Conference on Economics and Computation (EC'22). The order of the authors are certified random. The records are available in \url{https://www.aeaweb.org/journals/policies/random-author-order/search?RandomAuthorsSearch\%5Bsearch\%5D=4FJdnUr4sE80}.
The authors thank Tan Gan, Yingni Guo, Ryota Iijima, Nicolas Lambert, Jonathan A. Libgober, Xiaoyun Qiu, Asher Wolinsky, Boli Xu, Kai Hao Yang and the audience at the 33th Stony Brook International Conference on Game Theory, Yale Micro Theory Lunch and Seminars in Economic Theory for helpful suggestions. 
}}
\author{Yingkai Li\thanks{Cowles Foundation for Research in Economics and Department of Computer Science, Yale University, New Haven, CT, United States, 06511.
Email: \texttt{yingkai.li@yale.edu}.}
\and Jason D. Hartline\thanks{Computer Science Department, Northwestern University, Evanston, IL, United States, 60208. 
Email: \texttt{hartline@northwestern.edu}.}
\and Liren Shan\thanks{Toyota Technological Institute at Chicago, Chicago, IL, United States, 60637.
Email: \texttt{lirenshan@ttic.edu}.}
\and Yifan Wu\thanks{Computer Science Department, Northwestern University, Evanston, IL, United States, 60208.
Email: \texttt{yifan.wu@u.northwestern.edu}.}}
\date{}

\begin{document}

\maketitle
\begin{abstract}
We characterize the optimal reward functions (scoring rules) that incentivize an agent to acquire information and report it truthfully to the principal. The optimal scoring rules let the agent make a simple binary bet in single-dimensional problems, and choose the dimension with the most surprising signal to be scored on in symmetric multi-dimensional problems. This scoring rule format remains approximately optimal for asymmetric distributions. These results demonstrate the importance of linking incentives to obtain high-quality information in multi-dimensional problems. In contrast, standard scoring rules like the quadratic scoring rule, or averages of single-dimensional scoring rules can be far from optimal.
\end{abstract}



\section{Introduction}
\label{sec:intro}
Information elicitation plays a pivotal role across many fields. In decision-making contexts, it mitigates uncertainty and enhances decision quality by enabling a more informative assessment of potential outcomes and associated risks. 
In education, it incentivizes active knowledge acquisition. For instance, well-designed rewarding mechanisms for assignments and tests can stimulate students to actively seek out information. 
In machine learning, it ensures the availability of high-quality training data, which is critical to the quality of the training outcome. 
Effective reward functions (a.k.a.\ loss functions) can help extract valuable information from machine learning algorithms and guide algorithms to converge towards informative points. 
In medicine, it helps the evaluation of diverse diagnostic tests and treatment options, leading to better-informed healthcare decisions and improved patient outcomes. 
Overall, information elicitation incentivize agents to achieve better learning outcomes
and empowers decision-makers to make more informed choices
in their respective domains.
%This paper explores the multifaceted applications and implications of information elicitation, underscoring its essential role in optimizing processes and outcomes across a wide array of disciplines.

We study the optimization of scoring rules with the application of information elicitation in a canonical principal-agent setting, where an agent holds a subjective belief about an unknown state that can be verified later. For instance, an investor may seek advice (subjective beliefs) from a forecaster regarding the future stock market (verifable state) to make informed investment decisions. The principal’s goal is to design a state-contingent reward function (a.k.a., scoring rule) to elicit the forecaster’s subjective belief. By revelation principal, it is without loss of generality to focus on proper scoring rules where the forecaster is incentivized to truthfully report his belief to the principal. As pointed out by \citet{GR-07} and \citet{FK-19}, any decision problem of the agent can be viewed as a scoring rule, and hence the characterization of optimal scoring rules can be viewed as quantifying the maximum value of the agent's information source. Therefore, the optimization of scoring rule problems introduced in our paper transcend the application of information elicitation, and can also be used as a versatile tool for investigating the importance and value of information in a wide range of decision environments.

% We study a principal-agent model for information elicitation. In this model, 
% a forecaster holds a subjective belief about an unknown state that can be verified later. 
% The principal's goal is to design a state-contingent reward function (a.k.a., scoring rule) to elicit the forecaster's subjective belief. 
% For instance, an investor may seek advice (subjective beliefs) from a forecaster regarding the future stock market (unknown state) to make informed investment decisions. 
% The investor can then compensate the forecaster based on the actual stock market outcomes (state-contingent reward).
% The scoring rules that incentive the forecaster to truthfully report his belief are referred to as \emph{proper scoring rules}, 
% and the set of proper scoring rules has been characterized by \citet*{Mcc-56} and \citet{sav-71}, with more recent contributions from \citet{lam-11}.
% By revelation principle, it is without loss of generality to focus on proper scoring rules.

A novel feature in our paper is to introduce moral hazard into the model of information elicitation. 
Specifically, given a proper scoring rule, the forecaster can privately decide whether to exert costly effort to improve the quality of information before disclosing his belief. 
The quality of information is crucial for the principal in various applications, including decision making, education and etc.
Consequently, the principal is faced with a dual design challenge: incentivizing the forecaster to acquire high-quality information and subsequently encouraging truthful reporting of that information to the principal.

Not all proper scoring rules are equally effective in incentivizing effort. 
For example, consider a stylized application to exam grading. 
In the exam, there are $n$ different questions, and the final score must be between 0 and 1. 
Each question is a binary choice with uniformly drawn correct answers of ``True'' or ``False''. 
The student can choose whether to study the course materials before the exam or not. 
If the student chooses to study, 
he will receive an independent signal for each question, indicating the correct answer with a probability of~$\epsilon>0$. 
Alternatively, if the student chooses not to study, he will not receive any informative signals.
% he will independently receive a signal for each question, which equals the correct answer with a probability of~$\epsilon$. 
% Otherwise, he receive no informative signal.
The principal aims to incentivize students to study the course materials, as it leads to better learning outcomes. 
In this scenario, a naive approach is to score each question separately and then sum the scores. 
However, the maximum expected score difference between studying and not studying is at most~$\epsilon$, 
which upper bounds the incentive for the student to exert effort.
In contrast, as we will demonstrate in \cref{sec:peer grading}, by only scoring the question that the student feels most confident about,\footnote{More formally, the principal assigns a scoring rule to each question individually. The student then selects the question with the highest expected score based on his belief, and the principal only scores the student based on the chosen question.
} 
the principal can improve the expected score difference to $\frac{e-1}{2e}\approx 0.3$ when the number of questions $n\geq\frac{1}{\epsilon}$.
Therefore, by optimally designing the scoring rules, the principal can significantly enhance the student's incentive for effort, especially when the number of questions is large ($n$ is large) and all the questions are difficult ($\epsilon$ is small).
In particular, this example highlights the importance of linking decisions across different questions to incentivize effort \citep[c.f.,][]{JS-07};\footnote{Our result indicates the importance of linking decisions in an environment with moral hazard, where \citet{JS-07} showed its importance in an environment with pure adverse selection.} 
later we will demonstrate that this insight applies in a broad range of environments.

This paper provides a framework for the principal to optimize proper scoring rules for incentivizing effort.
We focus
on a simple binary effort model where the forecaster does or
does not exert effort,
and with this effort the forecaster obtains a
refined posterior distribution updated from the prior distribution on the
unknown state (e.g., by obtaining a signal that is correlated with the state). 
At the moment of making the effort decision, 
the forecaster is aware of 
% We adopt the objective that takes the perspective of the
% forecaster at the point of the effort decision with knowledge of 
both the distributions of posteriors that is obtained by exerting effort
and the prior when not exerting effort. 
We want a scoring rule that maximizes the difference in
expected scores for the posterior distribution and prior distribution.
While we do not explicitly model costs of effort, if the scoring rule maximizes the incentive for effort,
then it is adopted by forecasters with largest possible cost of effort. 
% The higher the expected score difference, 
% the higher the incentives to exert effort given the forecaster's cost of effort.\footnote{Since this expected difference is sufficient to quantify the forecaster's incentive for effort, 
% we will not model the cost explicitly in our paper.}
In this optimization program, 
we impose the ex post constraint that the score is in a
bounded range, i.e., without loss, between zero and one.\footnote{This constraint makes sense, e.g., for education applications. 
Subsequent studies of \citet{PW-22} show there is no qualitative difference in the solution with a bound in expected score.
}  
Notice that
this program would be meaningless without a bounded constraint on the score 
since otherwise the score could be scaled arbitrarily.

% The optimization of scoring rule problems introduced in our paper can also be interpreted as identifying the decision problem that maximizes the value of information, 
% given that the ex post reward given any action and any state in the decision problem is bounded. 
% The connection between scoring rules and the value of information has been observed in \citet{GR-07} and \citet{FK-19}. 
% Essentially, via the revelation principle, any decision problem can be transformed into a proper scoring rule by first eliciting the agent's belief and then optimizing the decisions for the agent based on the reported information. 
% Therefore, our characterizations of the optimal scoring rules can directly be applied to understand the maximum possible value of any given information structure in a decision environment with bounded rewards.

% The revelation principle implies that 
% it is without loss of generality to focus on optimizing scoring rules for eliciting the full distribution. 
\paragraph{Main Results.}

For problems with large state space (e.g., continuous state space), reporting full distributions requires unbounded communication between the principal and the forecaster, 
which is too costly to be implemented in practice
\citep[e.g.,][]{arrow1974limits,mookherjee2014mechanism}. 
Instead, the principal may only wish to elicit statistics of the posteriors in those scenarios. 
In this paper, we will focus on scoring rules that elicit the expectation of a
multi-dimensional state, 
where the forecaster is simultaneously reporting
the marginal expectations of the state in all dimensions.
The elicitation of marginal expectations is particularly valuable in numerous applications, such as 
% point forecasts \citep{gneiting2011making}, 
% or 
policy analysis aimed at understanding the average effects of different variables on various economic outcomes.

% In order to incentive the forecaster to truthfully report the posterior mean, 
% in betting mechanisms, the forecaster is restricted to select scores that are linear in the states. 
% Specifically, given parameter $c$ and normalization function $\kappa$ mapping from states to scores,
% the forecaster chooses a score function that is a linear function shifted by $\kappa$
% to maximizes his expected score at belief $\posterior$ 
% subjective to the constraint that the expected score at the prior is at most $c$. 
% Thus, here the forecaster is betting on linear score functions instead of selecting a subset of states. 
% We show that this betting mechanism is optimal for eliciting the mean. Next we show how to simplify and interpret it in various cases. 

We show that the optimal scoring rule for eliciting the mean can be implemented in the form of a betting mechanism. 
Specifically, the forecaster can make any bets that are linear in the realized state, 
subject to the constraint that the expected score given prior belief does not exceed a pre-determined threshold. 
For example, if the state space is $[0,1]$, 
the principal can specify a threshold $c$ such that
the forecaster can make bets with parameters $a,b$, and receives
a score of $a\state+b$ when the realized state is $\state\in[0,1]$, 
under the constraint that the expected score of $a\state+b$ given the prior belief over $\state$ is at most $c$. 
% A key difference between this mechanism and the one for the full distribution is that
% % lies in the fact that in betting mechanisms for eliciting the mean, 
% the forecaster is limited to making bets that are linear in the realized state.\footnote{For example, if the state space is $[0,1]$, 
% the forecaster can make bets with parameters $a,b$, and receives
% a score of $a\state+b$ when the realized state is $\state\in[0,1]$.
% } 
These betting mechanisms have simple interpretations in both single-dimensional space with arbitrary distributions 
and multi-dimensional space with symmetric distributions.
Moreover, we show that the implementation of the optimal betting mechanisms only requires that the principal knows the prior mean, not the distribution over posterior means, 
in these two ideal models. 
% and those scores are adjusted by an additional state dependent score~$\kappa$ in order to satisfy the boundedness constraints.
% Moreover, the chosen bets are constrained such that the expected score given the bets at the prior belief is at most $c$. 
% The choice of optimal parameters $c$ and~$\kappa$ in general may have complex dependence on the distribution over posterior means. 
% In this paper, we show how to identify the optimal parameters $c$ and $\kappa$ and provide simple economic interpretations of the optimal betting mechanism when the state space is single-dimensional 
% or the distribution is symmetric.\footnote{In those cases, it is without loss to set $\kappa=0$.} 
% Moreover, we show that a simple choice of $c=1/2$ and $\kappa=0$ achieves at least half of the optimal objective value for all possible distributions. 

% This is different from the betting mechanism for eliciting the full distribution where the agent selects subsets of states to get scored $0$ or $1$.
% We show that this betting mechanism is a $2$-approximation for eliciting the mean. A more complicated variant of this betting mechanism is optimal. 

For eliciting the mean in single-dimensional space, the optimal betting mechanism consists of two optional bets, 
and which bet is preferred depends on which side of the prior mean the posterior mean lies,
and thus the optimal mechanism induces a V-shaped interim utility with its lower tip at the expectation of the prior belief.
% Therefore, to implement this betting mechanism, the forecaster is asked to simply bet on which side the posterior mean is located compared to the prior mean, 
% and receive a linear score accordingly. 
% Therefore, states that are further away from the prior mean leads to a higher score difference between when the bet is correct. 

% as we expect for single-dimensional
% mechanism design problems for an agent with linear utility
% \citet{mye-81}, the optimal scoring rule is a step function, 
% which induces a V-shaped interim utility with its lower tip at the expectation
% of the prior belief.
% Thus, for single-dimensional space, the forecaster is betting on which side is the posterior located compared to the prior, 
% and receive a linear score accordingly. 

% To implement this scoring rule, it is sufficient for the principal to know the prior mean instead of the details on the distribution over posteriors.
% We also demonstrate a first result for prior-independent analysis of
% scoring rules. 
% Among scoring rules for reporting the expectation, the
% quadratic scoring rule is within a constant factor of optimal.

For multi-dimensional forecasting of the mean with symmetric distributions, 
our result highlights the importance of linking incentives in the design of multi-dimensional scoring rules, even when true distributions over states are independent \citep[c.f.,][]{JS-07}.
Specifically, the optimal scoring rule in rectangular state spaces
can be interpreted as a \emph{max-over-separate scoring rule}: 
the forecaster is scored on the dimension for which the forecaster's posterior in the optimal
single-dimensional scoring rule gives the highest expected utility.
This dimension is also the one with the most surprising posterior belief compared to the prior.
Equivalently, it can be implemented by letting the forecaster choose which dimension to report and be scored on.
(Recall, in our binary effort model, the agent either learns posteriors in all dimensions or none of them.)
% (after exerting effort to learn the posterior mean of all dimensions). 

For multi-dimensional forecasting without
a symmetry assumption, the optimal scoring rule may not exhibit simple representations.
Thus we focus on analyzing simple scoring rules and verifying whether those scoring rules are approximately optimal. 
% but we identify a V-shaped scoring rule that gives
% an 8-approximation. 
% In contrast, the ad hoc approach of scoring each dimension separately 
% may have an multiplicative loss in incentives for effort
% that is linear in the size of the dimensions. 
The method of approximation \citep{Hartline-12} allows economic conclusions to be drawn from simple mechanisms that are near optimal when optimal mechanisms are complex. 

Our main treatment of multi-dimensional scoring rules shows that \mos scoring rules that are optimal in symmetric environments -- where the forecaster is scored by the single dimensional scoring rule corresponding to the dimension for which the posterior update is most surprising -- are within a constant factor of optimal in asymmetric environments.
Moreover, to select an approximately optimal \mos scoring rule, it is sufficient for the principal to know the prior mean instead of the details of the distribution over posteriors.
This is consistent with the idea of ``ideal approximations" in \citet{Hartline-12}, 
where the goal is to prove that the optimal solution in ideal models is approximately optimal in general environments. 
In contrast, the common-in-practice separate scoring rule -- where the score is the sum of separate scoring rules on each dimension -- can be a linear factor from optimal in some cases. 
This finding reinforces the intuition that linking decisions is crucial for incentivizing effort in multi-dimensional problems.
% This further confirms our intuition of linking decisions from symmetric environments, 
% i.e., decisions must be linked in order to incentivize effort well. 

% To select an approximately optimal \mos scoring rule, it is sufficient for the principal to know the prior mean instead of the details on the distribution over posteriors.
% This result shows the importance of linking decisions in the design of multi-dimensional scoring rules, even when true distributions over states are independent \citep[c.f.,][]{JS-07}.
% \footnote{Our result indicates the importance of linking decisions in an environment with moral hazard, where \citet{JS-07} showed its importance in an environment with pure adverse selection.}

Our framework and results extend to settings with finite number of states where the principal wants to elicit the full distribution (instead of just the marginal means).
We show that the main insights we derived for eliciting the mean extend to eliciting the full distribution.
For example, optimal scoring rule for eliciting the full distribution can also be indirectly implemented as a betting mechanism, 
with the difference that the forecaster can make arbitrary bets for any state without the linearity constraint.\footnote{In fact, we show that the betting mechanism for eliciting the full distribution can be implemented as asking the agent to select a set of good states, a set of bad states and possibly a boundary state, where the good states are scored~$1$,  bad states are scored~$0$, and the boundary state if it exists is scored in the interior of $(0,1)$. }
Moreover, we show that by viewing each state as a separate dimension, it is also crucial to linking decisions across different states to incentivize effort
when eliciting the full distribution.
% That is, the principal fixes a parameter $c$, and for any posterior belief $\posterior$, 
% the forecaster chooses a score bounded in $[0,1]$ for each state
% that maximizes his expected score at belief $\posterior$,\footnote{We can show that the forecaster's optimal choice is to select score in $\{0,1\}$ for all but at most one state. }
% subject to the constraint that the expected score at the prior is at most~$c$. 
% Essentially, the forecaster is taking bets on the set of states that will occur with highest possible probability given posterior belief~$\posterior$, 
% subject to the constraint that the probability the prior allocates to that set is at most~$c$. 
% %We show that the betting mechanism is optimal for eliciting the full distribution. 
% Specifically, the betting mechanism for eliciting the full distribution can be implemented as asking the agent to select a set of good states, a set of bad states and a boundary state, where the good states are scored~$1$,  bad states are scored~$0$, and the boundary state is scored in the interior of $(0,1)$. 
% % A follow-up work \citep{PW-22} to our paper shows that this result is robust with the expected budget constraint. 
% We show that 

Finally, we show that 
the optimal incentives for the forecaster can be unboundedly smaller
for eliciting the mean compared to eliciting the full distribution. 
Therefore, even in applications where the only valuable information for the principal is the marginal means and communication is costly, 
the principal may still wish to elicit information beyond the mean (e.g., full information) and suffer from the additional communication cost 
in order to better incentivize the forecaster to exert effort.

\paragraph{Related Work.}
%%%%% characterization

Characterizations of scoring rules for eliciting the mean and for
eliciting a finite-state distribution play a prominent role in our
analysis.  Previous works show, in various contexts, that scoring
rules are proper if and only if their induced utility functions are
convex.  \citet*{Mcc-56} and \citet{sav-71} characterized proper scoring rules for
eliciting the full distribution on a finite set of
states. \citet*{Osb-85} characterized continuously differentiable
scoring rules that elicit multiple statistics of a probability
distribution. \citet*{lam-11} characterized the statistics that admit
proper scoring rules and characterized the
uniformly-Lipschitz-continuous scoring rules for the mean of a
single-dimensional state. \citet*{AF-12} characterized the proper
scoring rules for the marginal means of multi-dimensional random
states in the interior of the report space.  We augment this
characterization by showing that the induced utility function
converges to a limit on the boundary of the report space.  This
augmentation enables us to write the mathematical program that
optimizes over the whole report space.

%%%%% optimal information elicitation 

Most of the prior work looking at incentives of eliciting information
considers a fundamentally different model from ours.  This prior work
typically focuses on the incentives of the forecaster to exert effort
to obtain a signal (a.k.a., a data point), but then assumes that this
data point is reported directly (and cannot itself be misreported).
In this space, \citet*{CDP-15} consider the learning problem where the
principal aims to acquire data to train a classifier to minimize
squared error less the cost of eliciting the data points from
individual agents.  The mechanism for soliciting the data from the
agents trades off cost (in incentivizing effort) for accuracy of each
individual point.  \citet*{CILSZ-18} and \citet*{CZ-19} consider the
estimation of the mean of a population data.  Their objective is to
minimize the variance of the resulting estimator subject to a budget
constraint on the cost of procuring the data (from incentivizing
effort).

A few papers have considered incentivizing effort under a proper
scoring rule for a single-dimensional state.  \citet*{Osb-89} considers
incentivizing the forecaster to reduce variance under constraints that
result in the optimal scoring rule being quadratic.  \citet*{Z-11}
considers a slightly different single-dimensional model and derives that the optimal
scoring rule has V-shaped utility; our work begins with such a result
for our model.
% (Recall,
% single-dimensional mechanism design problems for agents with linear
% utility generally result in optimal mechanisms that induce V-shaped
% utility functions, \citealp{mye-81}.)
\citet*{NNW-20} consider a
forecaster with access to costly samples of a Bernoulli distribution
and characterizes optimal scoring rules in the limit as the sample
cost approaches zero.  Detailed discussion of these results is
deferred to \Cref{app:related-work}.  Our main contrasting result is
the approximate optimality of the V-shaped scoring rule for binary
effort and forecasts over multi-dimensional state spaces.

Our paper also relates to the literature on evaluating forecasters \citep{deb2018evaluating,Das-23}. The main difference is that the agent in their models cannot exert costly effort to acquire additional information. The information of the agent is exogenous in \citet{Das-23} and is controlled by the principal in \citet{deb2018evaluating}. In both papers, they consider the optimization of scoring rules for optimally separating good types from bad types and show that the optimal solution for a single-dimensional binary state is similarly V-shaped or V-shaped with ironing.

% Our paper also relates to the literature on evaluating forecasters \citep{deb2018evaluating,Das-23}. 
% The main difference is that the agent in their models cannot exert costly effort to acquire additional information. 
% The information of the agent is exogenous in \citet{Das-23}
% and is controlled by the principal in \citet{deb2018evaluating}. 
% Moreover, in both papers, the principal awards the agent by making an approval decision and the principal's utility is also affected by the decision. 
% In contrast, in our model, the principal has no value for the scores and seeks to incentivize effort subject to the budget constraints on scores. 

There are several papers on optimizing scoring rules following the
model proposed in our paper.  \citet{HLSW-21a} extend the framework
to the setting where the agent's effort is multi-dimensional (e.g.,
corresponding to independent tasks) and the agent can independently
exert effort in each dimension.  The main result of this extension is
that our intuition for the benefits of linking incentives across different dimensions generalizes. 
The authors propose a generalization of
the V-shaped scoring rule that is approximately optimal, which
requires the agent to predict $k$ dimensions correctly instead of one 
(where $k$ is a constant depending on the primitives).
% %
% \citet{HLSW-21b}
% extend the framework to the setting where the agent's effort is
% continuous (but single-dimensional) and the cost of the agent's effort
% is private to the agent.  In this case the principal benefits from
% offering several scoring rules (and agents with different costs choose
% different ones), each offered scoring rule is V-shaped.  The model
% also allows for the principal to have negative utility for payments to
% the agent.
% %
\citet{PW-22} design optimal scoring rules to minimize the expected payments to the agent 
under the constraint that certain social choice functions are implemented. 
This can be viewed as a dual problem to our objective of incentivizing effort subject to the budget constraints, 
and they provide polynomial time algorithms for computing the optimal scoring rules in their settings. 
\citet{CY-21} consider our objective of maximizing the incentives
of binary effort in a max-min design framework.  For example, they
show that the quadratic scoring rule is max-min optimal over a large
family of distributional settings.
\citet{Kong-21} generalizes the framework from single-agent scoring
rules to multi-agent peer prediction, i.e., without ground truth.  In
peer prediction, the designer needs to cross reference the reports of
different agents to verify the informativeness of the report.

% Scoring rules are also widely studied in the literature on peer
% prediction where ground truth is unknown and agent reports must be
% compared to each other. 
% For example, \citet*{FW-17} considers the optimization
% goal of incentive for effort in single-task peer prediction. 
% Their model is quite different from ours and the results are incomparable.

Our paper is relevant to the literature on information acquisition in the absence of a principal-agent relationship,
which provides decision theoretic foundations for various information cost functions \citep{sims2003implications, caplin2022rationally, pomatto2023cost, bloedel2021cost}
and information values \citep{blackwell1953equivalent,FK-19}.
Specifically, we focus on the design of scoring rules that maximize the agent's information value, thereby providing stronger incentives for the agent to exert effort given his exogenous information cost.

With broad strokes, our work connects the studies of optimal
mechanisms and optimal scoring rules.  A few points of connection
are especially pertinent.  Characterizations of incentives in scoring
rules and multi-dimensional mechanisms are similar.  The
multi-dimensional characterization for mechanism design is given by
\citet*{roc-85}, 
and is similar to the analogous results for scoring rules by \citet{Mcc-56}. 
One of our main results shows that a good scoring
rule for a multi-dimensional state is the \mos scoring rule, while
averaging over separate scoring rules is far from optimal.  This
result parallels the main contribution of \citet*{JS-07}, that linking
independent decisions improves incentives in mechanism design.  This
result also connects the study of simple scoring rules to the study of simple mechanisms like
the bundling-or-selling-separately mechanism of \citet*{BILW-14}.
Finally, the polynomial time algorithms we give in online appendix for computing optimal
scoring rules (in the cases where we do not provide simple analytic
characterizations) are based on a similar result of \citet*{BCKW-15} for computing revenue optimal pricing of randomized allocations.

\section{Application to Exam Grading}
\label{sec:peer grading}
In this section, we consider a stylized application to exam grading. 
This stylized example illustrates that optimal scoring rules are better than standard scoring rules
when the signals are not very informative, 
or the probabilities of acquiring informative signals are low. 

In this application, there are $n$ questions in the exam with binary outcomes $\{0,1\}$. 
For simplicity, we assume the prior is $\frac{1}{2}$ for all questions. 
Each student will submit a prediction $\mu_i\in [0,1]$ for each question $i \in [n]$ in the exam.\footnote{Note that for binary states, reporting the mean is equivalent to reporting the full distribution.} 
A natural interpretation is that all $n$ questions are True/False questions, 
and the students report their belief for each question being True.
These predictions are then graded by a scoring rule with ground
truth provided by the course staff.
Naturally such a scoring rule needs to satisfy the boundedness
constraint, 
and we normalize the bound of the score to $[0,1]$. 
We assume that the effort of the student is binary in this setting, 
i.e., either she can choose not to study before the exam and submit her prior belief for all questions, 
or she can study before the exam to have informative signals for all the questions in the exam.
The goal of the principal is to design the grading scheme to incentivize the student to study before the exam. 

To simplify the exposition, we consider a specific information structure with parameters $\epsilon,\delta\in[0,\frac{1}{2}]$ in this section.
Specifically, if the student chooses to study before the exam, 
independently for each question $i$,
she will receive an informative signal for question $i$ that leads her belief to either $\mu_i=\frac{1}{2}-\delta$ or $\mu_i=\frac{1}{2}+\delta$ 
with probability $\epsilon$ each. 
With the remaining probability $1-2\epsilon$, the posterior belief remains the same as the prior, i.e., $\mu_i = \frac{1}{2}$.
In this example, parameter~$\epsilon$ measures how likely the agent can get an informative signal for each question (with~$\frac{1}{2}$ being the most likely)
and parameter~$\delta$ measures the informativeness of the signal (with~$\frac{1}{2}$ being the most informative). 

\paragraph{Single Question}
We first consider the case when there is only a single question, 
i.e., $n=1$. 
We omit the subscript $i$ in notations in this single question application. 
One of the commonly used scoring rules in practice is quadratic scoring rule.
% \footnote{Another commonly used scoring rule is log scoring rule. However, log scoring rule is not bounded and hence cannot be applied in this application. } 
Specifically, when the student reports belief $\mu$ and $\state\in\{0,1\}$ is the correct answer of the question, 
the final score of the student is: 
\begin{align*}
S_q(\mu,\state) = 1-(\mu-\state)^2.
\end{align*}
To compute the incentives for the student to study before the exam given quadratic scoring rule, 
the expected score for not studying and reporting the prior $\frac{1}{2}$ is 
$1-(\frac{1}{2})^2=\frac{3}{4}$, 
while the expected score for studying before the exam is 
\begin{align*}
(1-2\epsilon)\cdot \frac{3}{4}
+ 2\epsilon\cdot \rbr{\rbr{\sfrac{1}{2}-\delta}\cdot\rbr{1-\rbr{\delta+\sfrac{1}{2}}^2} + \rbr{\sfrac{1}{2}+\delta}\cdot\rbr{1-\rbr{\delta-\sfrac{1}{2}}^2}}
= \frac{3}{4} + 2\epsilon\delta^2.
\end{align*}
Therefore, the expected score difference is $2\epsilon\delta^2$.
This implies that the student is more likely to study if she is more likely to acquire an informative signal (when $\epsilon$ is larger), and the signal is more informative (when $\delta$ is larger). 

However, we will show that quadratic scoring rule can be far from the optimal for providing incentives for the student to study before the exam 
when the signals are not very informative. 
Specifically, consider the scoring rule where if the student's report $\mu\geq\frac{1}{2}$, 
she receives score $1$ when the correct answer is $1$ and $0$ otherwise. 
Moreover, if the student's report $\mu<\frac{1}{2}$, 
she receives score $1$ when the correct answer is $0$ and $0$ otherwise. 
Formally, when the report is $\mu$ and the correct answer is $\state$, 
\begin{align*}
\hat{S}(\mu,\state) = 
\begin{cases}
1 & \mu < \frac{1}{2} \ \&\  \state = 0\\
1 & \mu \geq \frac{1}{2} \ \&\  \state = 1\\
0 & \text{otherwise}.
\end{cases}
\end{align*}
% \begin{align*}
% \hat{S}(\mu,\state) = 
% \begin{cases}
% \indicate{\state = 0} & \mu < \frac{1}{2}\\
% \indicate{\state = 1} & \mu \geq \frac{1}{2}.
% \end{cases}
% \end{align*}
% when the report is $\mu$ and the correct answer is $\state$.
% Here $\indicate{\cdot}$ is the indicator function.
% Intuitively, 
It is easy to verify that the student has incentives to truthfully report his belief in the exam given this scoring rule. 
Moreover, 
the expected score for not studying and reporting the prior is $\frac{1}{2}$, 
while the expected score for studying before the exam is 
\begin{align*}
(1-2\epsilon)\cdot \frac{1}{2}
+ 2\epsilon\cdot \rbr{\frac{1}{2}+\delta}
= \frac{1}{2} + 2\epsilon\delta.
\end{align*}
Therefore, the expected score difference given scoring rule $\hat{S}$ is $2\epsilon\delta$.
Note that the multiplicative ratio between expected score difference given scoring rule $\hat{S}$ and quadratic scoring rule $S_q$
is $\frac{1}{\delta}$. 
When the signals are not very informative, 
i.e., when $\delta$ is small, 
the multiplicative ratio is large 
and the incentives provided by the quadratic scoring rule are poor compared to the optimal.
% \footnote{Although the quadratic scoring rule can be far from optimal for specific distributions over posteriors, in the online appendix, we show that it is robustly optimal if the principal is uncertain about the distribution over posteriors.} 
Therefore, implementing optimal scoring rules are necessary in practice when signals are not very informative. 

An interesting observation for the scoring rule $\hat{S}$ we constructed in this example
is that it can be implemented by only requiring the student to make a binary report. 
That is, in the application of exam grading, it is sufficient to let the student answer True or False for the question, 
and the answer True indicates that 
the student believes that the correct answer is more likely to be True than the prior. 
The final score of the student is $1$ if and only if her report matches the correct answer. 
In \cref{sec:single assignment}, we will show that simple scoring rules with binary reports are optimal for all single-question settings
even when there are multiple states with arbitrary prior distribution and arbitrary distribution over posterior beliefs.

\paragraph{Multiple Questions}
We consider the case of multiple questions where $n\geq 2$. A natural approach for multiple questions in the exam is to grade the questions separately. 
Moreover, since all questions are i.i.d., it is without loss to grade them equally. 
That is, each question has a score budget of $\frac{1}{n}$, and the final score of the student is the sum of scores from all $n$ questions.
We will show that grading questions separately will have poor incentives for the student to study if signals are unlikely to be informative, i.e., when~$\epsilon$ is small. 

To grade the questions separately, it is optimal to adopt the optimal single-question scoring rule $\hat{\score}$ separately for each question and then take the sum. 
Based on our discussion in the single-question example, 
the optimal scoring rule that grades the questions separately is 
\begin{align*}
\hat{S}_n(\mu,\state) = \sum_{i\in[n]} \frac{1}{n}\cdot\hat{S}(\mu_i,\state_i).
\end{align*}
It is easy to compute that the expected score difference between studying or not is $2\epsilon\delta$ given scoring rule $\hat{S}_n$. 

Now consider another scoring rule $\bar{S}_n$
where instead of grading the questions separately and taking the sum, 
the instructor lets the student choose which question she would prefer to be graded on in the exam. 
Specifically, when the student chooses question $i^*$ in the exam, 
her score is 
\begin{align*}
\bar{S}_n(\mu,\state) = \hat{S}(\mu_{i^*},\state_{i^*}).
\end{align*}
Given scoring rule $\bar{S}_n$, the student always chooses the question $i^*$ that maximizes her expected score based on her posterior belief $\mu$. 
In this example, the student will choose the question for which she receives an informative signal, i.e., such that her posterior belief for that question is updated to either $\frac{1}{2}-\delta$ or $\frac{1}{2}+\delta$. 
If she does not receive any informative signals, she will randomly guess one.
This scoring rule can also be implemented by having the student report her posteriors for all questions. The instructor, acting on the student's behalf, will then choose a question $i^*$ based on the reports to maximize the expected score (with randomness according to the reports over the realized state).

It is easy to show that the expected score difference given scoring rule $\bar{S}_n$ is 
$(1-(1-2\epsilon)^n)\cdot\delta$. 
Note that this score difference is strictly larger than $2\epsilon\delta$ 
for any $n\geq 2$ and $\epsilon,\delta > 0$. 
That is, scoring rule $\bar{S}_n$ outperforms scoring separately for any number of questions in the exam. 
Moreover, the multiplicative gap between $\bar{S}_n$ and $\hat{S}_n$
is large when the probability of an informative signal $\epsilon$ is small and $n$ is large. 
In particular, when $n\geq \frac{1}{2\epsilon}$, 
the multiplicative gap is 
\begin{align*}
\frac{1}{2\epsilon} \cdot (1-(1-2\epsilon)^n) 
\geq \frac{1-\sfrac{1}{e}}{2\epsilon}
\approx \frac{0.3}{\epsilon}.
\end{align*}
Therefore, compared to scoring rule $\bar{\score}_n$, the incentive for the student to study before the exam is poor given scoring rules that grade the questions separately 
when there are many questions in the exam 
and the probability of receiving an informative signal for each questions is small, 
i.e., the questions in the exam are hard. 
Therefore, implementing optimal scoring rules are necessary in practice when the probabilities of acquiring informative signals are low. 

The scoring rule $\bar{S}_n$ proposed in this example exhibits a notable feature wherein the student is graded only on a single question. This question corresponds to the one on which she feels most confident based on her beliefs,
and is also the question with the most surprising posterior. 
In \cref{sec:multi}, we show that this scoring rule is optimal in symmetric environments and approximately optimal in asymmetric environments even when there are multiple states and multiple signals for each question.

\section{Preliminaries}\label{sec:model}

This paper considers the problem of optimizing scoring rules. 
There is an unknown state $\state \in \statespace$ where 
$\statespace \subseteq \reals^{\numasm}$ is a compact set in $\numasm$-dimensional Euclidean space. 
The agent has a private belief $\posterior\in \Delta(\statespace)$ about the unknown state. 
Let $\mean{\posterior}\in\reals^n$ be the marginal means of belief $\posterior$. 
The principal can commit to a scoring rule $\score:\reportspace\times\statespace \to\reals$, 
which is a mapping from the agent's reports and
the realized state to a score for the agent, 
to elicit information from the agent regarding his subjective belief. 
Here $\reportspace$ is the report space of the agent. 
In this paper, we mainly focus on the case when the report space is the marginal means of all dimensions, i.e., $\reportspace= \conv(\statespace) \subseteq \reals^n$ 
where $\conv(\statespace)$ is the convex hull of the state space. 
In \cref{sec:full}, we characterize the optimal scoring rules for eliciting the full distribution 
where the report space $\reportspace = \Delta(\statespace)$ is the set of all possible posterior distributions. 

A scoring rule $\score$ is proper for eliciting the mean if the agent has incentives to truthful report the marginal means of his belief
given scoring rule $\score$.
\begin{definition}[Proper]\label{def:proper_mean}
A scoring rule $\score(\report, \theta)$ is \emph{proper for eliciting the mean}\footnote{Our notion of proper scoring rule is weakly proper rather than strictly proper.  Most of the literature on scoring rules does not have an objective and, to obtain non-trivial results, requires scoring rules to be strictly proper.  When optimizing scoring rules there is no meaningful difference between strictly proper and proper as the strictness can be arbitrarily small and therefore provide insignificant additional benefit.  Note that any weakly proper scoring rule can also be made strictly proper by taking an arbitrarily small convex combination with a strictly proper scoring rule.} if 
for any distribution~$\posterior$
and report $\report \in \reportspace$, 
we have 
\begin{align*}
\expect{\state\sim\posterior}{\score(\mean{\posterior}, \theta)}
\geq \expect{\state\sim\posterior}{\score(\report, \theta)}.
\end{align*}
\end{definition}

In addition to the proper constraint, we also require the scoring rule to be bounded. 

\begin{definition}[Boundedness]\label{def:bounded}
A scoring rule $\score(\report, \theta)$ is bounded by $\scorebound$
in space $\reportspace \times \statespace$
if %there exists a constant $$ such that 
$\score(\report, \theta) \in [0, \scorebound]$ for any 
report $\report \in \reportspace$ and state $\state \in \statespace$. 
\end{definition}

\paragraph{Binary Effort Model}
We consider a canonical binary effort model for the optimization of scoring rules. 
In this model, there is a prior distribution $\statedist \in \Delta(\statespace)$
over the unknown state $\state \in \statespace$. 
The prior distribution $\statedist$ is publicly known by both the agent and the principal. 
In addition, if the agent chooses to exert effort, 
the agent can privately observe an additional 
signal about the true state, which induces a posterior~$\posterior$.
We denote $\distoverposterior$ as the distribution over posteriors.

The goal for the principal is to design a bounded proper scoring rule that maximizes the difference in expected score between agents who exert effort and those who do not. 
% This objective is justified because a larger difference in expected score corresponds to a higher maximum cost of the agent for exerting effort in this model.
Formally, given the maximum score of~$\scorebound$, the state space~$\statespace$ 
and the report space $\reportspace = \conv(\statespace)$,
the optimization program for maximizing the difference in expected score is\footnote{In the online appendix, we provide similar characterizations for a similar model where the ex post bounded score constraint is replaced with the bounded in expectation constraint.}
\begin{align}
\max_\score \qquad 
&\expect{\posterior \sim \distoverposterior, \state \sim \posterior}{\score(\mean{\posterior}, \state) - \score(\priorMean, \state)}\label[scoringruleprogram]{eq:mainprogram}\\
\text{s.t.}\qquad &\score \text{ is a proper scoring rule for eliciting the mean},\nonumber\\
&\score \text{ is bounded by $\scorebound$
in space $\reportspace \times \statespace$}.\nonumber
\end{align}

The above program aims to optimize the incentive for the agent to
exert effort.  Consider the situation where the agent has a private
stochastic cost for exerting effort and obtaining a signal of the true state. 
% If the
% agent chooses to pay the cost, she sees the realized signal, forms a
% posterior about the true state, and optimizes according to the
% posterior. 
The agent will only choose to pay the cost if her expected
gain from obtaining the signal, i.e., the objective value in
\Cref{eq:mainprogram}, is higher than her cost.  By designing the
optimal scoring rule for \Cref{eq:mainprogram}, we also maximize the
probability that the agent chooses to pay the cost.  This paper will
not formally model such costs.
Next we simplify the program for eliciting the mean 
using known characterizations for proper scoring rules.

\subsection{Proper Scoring Rules for Eliciting the Mean}
% \subsection{Eliciting the Mean with Canonical Scoring Rules}
\label{subsec:formal_program}\label{subsec:proper}

There is a canonical approach for constructing proper scoring rules for eliciting the mean.
In this section we focus on simplifying \Cref{eq:mainprogram} by restricting attention to canonical proper
scoring rules,
and then in \cref{apx:necessity_canonical} we show that this restriction is
without loss for the program under a mild technical condition. 
The following definition and
proposition are straightforward from first-order conditions and can be found, e.g., in \citet{AF-12}. 
We provide a geometric proof of Lemma~\ref{prop:canonical-proper} in Appendix~\ref{apx:prelim} for completeness.%We give the proof below because it is
%illustrative of the geometry of scoring rules for eliciting the mean.

\begin{definition}
  \label{def:canonical}
  A canonical scoring rule $\score$ for eliciting the mean is defined by convex
  utility function $\util : \reportspace \to \reals$ on report space
  $\reportspace$, subgradient $\sg : \reportspace \to
  \reals^{\numasm}$ of $\util$, and normalization function $\kappa : \statespace \to
  \reals$ on state space $\statespace$ as
  \begin{align}
    \label{eq:scoring-rule-construction}
    \score(\report, \state)= \util(\report) + \sg(\report) \cdot
    (\state-\report)+\kappa(\state).
  \end{align}
\end{definition}

\begin{lemma}[\citealp{AF-12}]
  \label{prop:canonical-proper}
  Canonical scoring rules are proper for eliciting the mean.
\end{lemma}

In any canonical scoring rule $\score$, normalization function $\kappa(\state)$ shifts the score based on the state 
and does not depend on the report, 
which does not affect the agent's incentive. 
By ignoring normalization function $\kappa$ for a moment, 
an interpretation of the canonical scoring rule for eliciting the mean 
is that the agent is essentially 
betting on being scored by hyperplanes. 
% taking bets on hyperplanes for scores. 
Specifically, the agent takes the bet on hyperplane
$\util(\report) + \sg(\report) \cdot (\state-\report)$
by making report~$\report$, 
and the score of the realized state is evaluated on this hyperplane.
The convexity of utility function $\util$ ensures that the expected score of the bet is maximized by reporting the posterior mean truthfully. 

\cref{prop:canonical-proper} implies that the scoring rule is essentially determined by the convex utility function $\util$. 
We say a scoring rule $\score$ is induced by utility function $\util$ if 
there exist subgradient $\sg(\report)$ and normalization function $\kappa(\state)$ such that 
\Cref{eq:scoring-rule-construction} holds. 

% \begin{proof}
% Canonical scoring rules have the following simple interpretation.  By
% making a report $\report$, the agent selects the supporting hyperplane
% of $\util$ at $\report$ on which to evaluate the state.  This
% supporting hyperplane has gradient $\sg(\report)$ and contains point
% $(\report,\util(\report))$.  The agent's utility is equal to the value
% of the realized state $\state$ on this hyperplane (plus constant
% $\kappa(\state)$ which is independent of the agent's report).  With
% utility given by a random point on a hyperplane, the expected utility
% is equal to its mean on the hyperplane.  When the agent's true
% posterior belief is that the state has mean $\report$, the agent's
% expected utility is $\util(\report)$ (plus a constant equal to the
% expected value of $\kappa(\cdot)$ under the agent's posterior belief;
% summarized below as \Cref{lem:canonical-expectation}).  Misreporting
% $\report'$ with belief $\report$ gives a utility equal to the value of
% $\report$ on the supporting hyperplane with gradient $\sg(\report')$
% at $\report'$.  By convexity of $\util$, a report of $\report$ gives
% the higher utility of $\util(\report)$.
% \end{proof}

% Therefore, we will focus on canonical scoring rules for eliciting the mean in this paper, and in \cref{apx:necessity_canonical} we show that this restriction is without loss under a mild technical condition.
Given the format of canonical scoring rules, the following lemma allows the objective and the boundedness
constraint of \Cref{eq:mainprogram} to be simplified. 
In particular, this lemma
justifies referring to $\util$ as the agent's utility function
and allows us to reformulate the optimization problem in terms of the utility function $\util$ and its subgradient $\sg$ instead of the scoring rule~$\score$.
% and its
% proof was observed in the proof of \Cref{prop:canonical-proper}
% and the proof of the second lemma 
The idea of \cref{lem:canonical-bound} is illustrated in \Cref{f:single}(a) for single-dimensional states, 
and the proof of the general result is deferred to \cref{apx:prelim}.

\begin{lemma}
  \label{lem:canonical-expectation}\label{lem:canonical-bound}
  For any canonical scoring rule for the mean $\score$ (defined by
  $\util$, $\sg$, and $\kappa$), the expected utility from belief
  $\posterior$ and truthfully report of $\mean{\posterior}$ is
  \begin{align}
    \label{eq:canonical-expected-utility}
    \expect{\state \sim \posterior}{\score(\mean{\posterior}, \state)}
    = \util(\mean{\posterior}) + \expect{\state \sim
      \posterior}{\kappa(\state)}.
  \end{align}
Moreover, for any utility function $\util$ and
subgradient $\sg$, 
there exists a normalization function $\kappa$ such that canonical scoring rule $\score$ (defined by $\util$, $\sg$ and $\kappa$) 
satisfies the score bound $\scorebound$
if and only if for any report $\report\in\reportspace$
and state $\state\in\statespace$,
  \begin{align}
    \label{eq:canonical-bound}
    \util(\state)-\util(\report) - \sg(\report)\cdot(\state - \report) &\leq \scorebound. 
  \end{align}
\end{lemma}

We now derive the simplified program for canonical scoring rules.  The
following notation is sufficient to describe this simplified program
and is adopted throughout the paper.  
Recall that $\distoverposterior$ is the distribution over posteriors if the agent exerts effort. 
Let~$\marginalReport_{\mu}$ be the distribution over posterior means
where $\marginalReport_{\mu}(z) \triangleq
\distoverposterior(\{\posterior: \mean{\posterior} \in z\})$ for any $z\subseteq\reportspace$. 
By slightly abusing notation, we drop the subscript in~$\marginalReport_{\mu}$ and denote both distributions by $\marginalReport$.
% Note that the distribution over posterior beliefs induces a distribution over posterior means,
% and thus by slightly abusing notation, we denote both distributions by
% $\marginalReport$.  Specifically, $\marginalReport_{\mu}(z) =
% \int_{\posterior: \mean{\posterior} \in z}
% \distoverposterior(\posterior) \, \mathrm{d}\posterior$ for any $z\subseteq\reportspace$,
% i.e., the density at posterior mean $\report$ is equal to the
% cumulative density of posteriors $\posterior$ with mean
% $\mean{\posterior} = \report$. 
Note that the prior mean $\priorMeanState$ equals the expected posterior mean~$\priorMeanMean$, 
i.e., $\priorMeanState = \expect{\state \sim \statedist}{\state} 
= \expect{\report \sim \marginalReport}{\report}
= \priorMeanMean$.

By \Cref{lem:canonical-expectation}, the objective function
in \Cref{eq:mainprogram} for canonical
scoring rules can be simplified as
$\objfunc(\util,\marginalReport) = \int_{\reportspace} \left[ \util(\report)-\util(\priorMeanMean)\right]\dd\marginalReport(\report)$.
% \begin{align*}
% \expect{\posterior \sim \distoverposterior, \state \sim \posterior}{\score(\mean{\posterior}, \state) - \score(\priorMean, \state)}
% &=\int_{\Delta(\statespace)}\left[ \util(\mean{\posterior})-\util(\priorMean)\right]
% \dd\distoverposterior(\posterior)\\
% &=\int_{\reportspace} \left[ \util(\report)-\util(\priorMeanMean)\right]\dd\marginalReport(\report).
% \end{align*}
% Note that the simplified objective function does not depend on
% subgradient~$\sg$ or state function~$\kappa$, the latter of which is
% cancelled in the score difference. 
% We denote this objective function as $\objfunc(\util,\marginalReport) = \int_{\reportspace} \left[ \util(\report)-\util(\priorMeanMean)\right]\dd\marginalReport(\report)$. 
Combining it with \Cref{lem:canonical-bound}, and 
shifting the utility function by a constant
such that
$\util(\priorMeanMean)=0$, we get the following optimization program
for optimizing over canonical scoring rules with report space $\reportspace = \conv(\statespace)$.  
\begin{align}
\OPT(\marginalReport, \scorebound, \statespace) 
= \max_\util \qquad &\int_{\reportspace} \util(\report)\dd\marginalReport(\report) \label[scoringruleprogram]{eq:program}\\
\text{s.t.} \qquad 
&\util\text{ is a convex function, and } \util(\priorMeanMean)=0,\nonumber\\ 
&\sg(\report) \in \subgradient \util(\report), \quad \forall \report \in \reportspace, \nonumber \\
&\util(\state)-\util(\report)
-\sg(\report) \cdot (\state-\report) \leq \scorebound, 
\quad \forall \report \in \reportspace, \state\in \statespace.\nonumber
\end{align}

% \begin{lemma}\label{lem:linear in max score}
% For any distribution $\marginalReport$ and state space $\statespace$, 
% the optimal objective $\OPT(\marginalReport, \scorebound, \statespace)$ 
% is a linear function of the maximum score $\scorebound$. 
% \end{lemma}
% \begin{proof}
% For any distribution $\marginalReport$, any state space $\statespace$, 
% and any valuation $\scorebound$ and $\scorebound'$, 
% let $\util$ be the optimal solution for 
% Program \eqref{eq:program} with maximum score $\scorebound$.
% Let $\tilde{\util}(\report) = \frac{\scorebound'}{\scorebound} \cdot \util(\report)$. 
% It is easy to verify that $\tilde{\util}$ is a feasible solution for
% Program \eqref{eq:program} with maximum score $\scorebound'$.
% Therefore, we have 
% $\OPT(\marginalReport, \scorebound', \statespace) 
% \geq \frac{\scorebound'}{\scorebound} \cdot \OPT(\marginalReport, \scorebound, \statespace)$. 
% Similarly, we can also show that 
% $\OPT(\marginalReport, \scorebound, \statespace) 
% \geq \frac{\scorebound}{\scorebound'} \cdot \OPT(\marginalReport, \scorebound', \statespace)$.
% Combining the inequalities, 
% we know that the optimal objective is a linear function of the maximum score. 
% \end{proof}

Note that for any distribution $\marginalReport$ and state space
$\statespace$, the objective $\OPT(\marginalReport,
\scorebound, \statespace)$ is a linear function of the maximum score
$\scorebound$.  In most of the paper, we normalize $\scorebound = 1$
and mainly consider the state space $\statespace = [0,1]^\numasm$.  To
simplify the notation, we let $\OPT(\marginalReport) =
\OPT(\marginalReport, 1, [0,1]^\numasm)$.  We will write
$\OPT(\marginalReport, \scorebound, \statespace)$ explicitly in
\Cref{subsec:max over separate} when we discuss general state spaces with bound
$\scorebound \neq 1$.

\section{Single-dimensional Scoring Rules}\label{sec:single assignment}

In this section, we focus on the special case of single-dimensional
state spaces.  We characterize the optimal single-dimensional scoring
rules for eliciting the mean and show that the optimal scoring rules
are simple and only depend on the prior mean of the distribution.
% In the online appendix, we
% compare the quadratic scoring rule to the optimal scoring rule and
% show that the quadratic scoring rule, though it can be far from
% optimal for specific distributions over posteriors, it is
% robustly optimal.

We normalize the state space $\statespace$ so that its
convex hull, i.e., the report space $\reportspace$, is $[0,1]$ and the
boundedness constraint is given by $\scorebound = 1$.
First note that for single dimensional
scoring rules, the boundedness constraint of \Cref{eq:program}
can be further simplified.
Intuitively, as illustrated in \Cref{f:single}(a), the following lemma suggests that 
since the utility function is convex, 
the boundedness constraints only bind with both reports and states in the boundary. 
Specifically, for boundary state $\state = 1$, 
the maximum difference in score is attained between boundary reports $\report=0$ and $\report = 1$, 
with the difference being 
\begin{align*}
\score(1,1) - \score(1,0) = \util(1)-\util(0)-\sg(0) \leq 1
\end{align*}
due to the boundedness constraint. 
The equality holds due to the format of canonical scoring rules in \cref{def:canonical}.
Similarly, for boundary state $\state = 0$,
the maximum difference in score is $\util(0)-\util(1)-\sg(1) \leq 1$.
We show that these two inequalities are sufficient to capture all boundedness constraints. 

% \begin{figure}[t]
% \centering
% \input{figures/bounded_single}
% \caption{Bounded constraint for proper scoring rule for single assignment.}
% \label{fig:bounded single}
% \end{figure}
% \begin{figure}[t]
% \centering
% \input{figures/opt_single}
% \caption{Characterization of optimal scoring rule for single assignment.
% In the above figure, for any convex function $\util$ (dotted line) that induces a bounded scoring rule, 
% there exists another convex function $\tilde{\util}$ (solid line) which also induces a bounded scoring rule
% and weakly improves the objective. 
% }
% \label{fig:opt single}
% \end{figure}

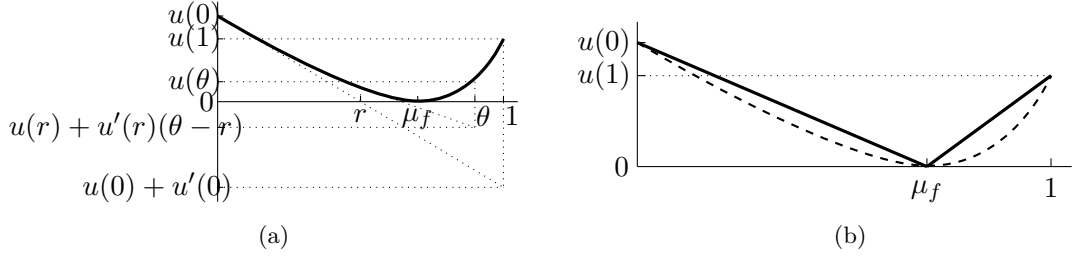
\begin{figure}[t]
\centering
\subfloat[]{
\begin{tikzpicture}[scale = 0.48]

% \draw [white] (0, 0) -- (15, 0);
\draw (0,0) -- (10.5, 0);
\draw (0, -3.5) -- (0, 3.5);

\draw [very thick] plot [smooth, tension=0.8] coordinates {
(0, 3) (6.8,0.03) (10, 2.2)
};

\draw [dotted] (0, 3) -- (10, -3);

\draw (-0.38, 0) node {$0$};
\draw (10.3, -0.5) node {$1$};
\draw (7, -0.6) node {$\priorMean$};
\draw (7, 0) -- (7, 0.2);
\draw (10, 0) -- (10, 0.2);

\draw (0, 3) -- (0.2, 3);
\draw (0, 2.2) -- (0.2, 2.2);

\draw (-1, 3) node {$\util(0)$};
\draw (-1, 2.2) node {$\util(1)$};
\draw (-2.3, -3) node {$\util(0)+\util'(0)$};
\draw (0, -3) -- (0.2, -3);

\draw [dotted] (10, 2.2) -- (10, -3);
\draw [dotted] (0, 2.2) -- (10, 2.2);
\draw [dotted] (0, -3) -- (10, -3);

\draw (5, -0.5) node {$\report$};
\draw (5, 0) -- (5, 0.2);
\draw [dotted] (5, 0) -- (5, 0.5);

\draw (9.3, -0.5) node {$\state$};
\draw (9, 0) -- (9, 0.2);
\draw [dotted] (9, 0.7) -- (9, -0.9);

\draw [dotted] (5, 0.5) -- (9, -0.9);

\draw (-3.5, -0.9) node {$\util(\report) + \util'(\report)(\state - \report)$};
\draw (0, -0.9) -- (0.2, -0.9);
\draw [dotted] (0, -0.9) -- (9, -0.9);

\draw (-1, 0.7) node {$\util(\theta)$};
\draw (0, 0.7) -- (0.2, 0.7);
\draw [dotted] (0, 0.7) -- (9, 0.7);

\end{tikzpicture}
\label{fig:bounded single}
}
\subfloat[]{
\begin{tikzpicture}[scale = 0.45]

\draw [white] (0, 0) -- (11.5, 0);
\draw (0,0) -- (10.5, 0);
\draw (0, 0) -- (0, 3.5);

\draw [dashed, thick] plot [smooth, tension=0.8] coordinates {
(0, 3) (6.8,0.03) (10, 2.2)
};

\draw [very thick] plot (0, 3) -- (7, 0);
\draw [very thick] plot (7, 0) -- (10, 2.2);

\draw (-0.38, 0) node {$0$};
\draw (10, -0.5) node {$1$};
\draw (7, -0.6) node {$\priorMean$};
\draw (7, 0) -- (7, 0.2);
\draw (10, 0) -- (10, 0.2);

\draw (0, 3) -- (0.2, 3);
\draw (0, 2.2) -- (0.2, 2.2);

\draw (-0.9, 3) node {$\util(0)$};
\draw (-0.9, 2.2) node {$\util(1)$};
% \draw (-1.9, -3) node {$\util(0)+\util'(0)$};
% \draw (0, -3) -- (0.2, -3);

% \draw [dotted] (10, 2.2) -- (10, -3);
\draw [dotted] (0, 2.2) -- (10, 2.2);
% \draw [dotted] (0, -3) -- (10, -3);

% \draw (5, -0.5) node {$\report$};
% \draw (5, 0) -- (5, 0.2);
% \draw [dotted] (5, 0) -- (5, 0.5);

% \draw (9.3, -0.5) node {$\state$};
% \draw (9, 0) -- (9, 0.2);
% \draw [dotted] (9, 0.7) -- (9, -0.9);

% \draw [dotted] (5, 0.5) -- (9, -0.9);

% \draw (-2.8, -0.9) node {$\util(\report) + \util'(\report)(\state - \report)$};
% \draw (0, -0.9) -- (0.2, -0.9);
% \draw [dotted] (0, -0.9) -- (9, -0.9);

% \draw (-0.8, 0.7) node {$\util(\theta)$};
% \draw (0, 0.7) -- (0.2, 0.7);
% \draw [dotted] (0, 0.7) -- (9, 0.7);

\end{tikzpicture}
\label{fig:opt single}
}
\caption{\label{f:single}
The figure on the left hand side illustrates the bounded constraint for proper scoring rule for single dimensional states.
The figure on the right hand side characterizes the optimal scoring rule (solid line) for single dimensional states.
In this figure, for any convex function $\util$ (dotted line) that induces a bounded scoring rule, 
there exists another convex function $\tilde{\util}$ (solid line) which also induces a bounded scoring rule
and weakly improves the objective. }
\end{figure}

\begin{lemma}\label{lem:bounded for single}
  For state space $\statespace$ with convex hull $[0,1]$ and any convex utility function~$\util$, there exists a
  proper scoring rule induced by function $\util$
  which is bounded by $\scorebound = 1$ if and only if there exists a set of
  subgradients $\sg(\report) \in \subgradient \util(\report)$ such
  that
$$
\util(1)-\util(0)-\sg(0) \leq 1 \text{ and } 
\util(0)-\util(1)+\sg(1) \leq 1.
$$
\end{lemma}
The proof of \cref{lem:bounded for single} is provided in \cref{sub:bound_single}. 
With \Cref{lem:bounded for single}, 
\Cref{eq:program} can be simplified as
\begin{align}\label[scoringruleprogram]{eq:1d_program_u}
\max_u \quad&
\int_0^1 \util(\report) \dd
\marginalReport(\report) \\
\text{s.t.} \quad& 
\util(\report) \text{ is convex and $\util(\priorMean) = 0$, } \nonumber\\
&\sg(\report) \in \subgradient \util(\report), \forall \report \in [0,1], \nonumber \\
& \util(1)-\util(0)-\sg(0) \leq 1, \nonumber\\
& \util(0)-\util(1)+\sg(1) \leq 1. \nonumber
\end{align}
The main result of this section is the following characterization of
the optimal solutions to \Cref{eq:1d_program_u}.

\begin{definition}\label{def:v-shaped scoring rule}
A function $\util$ is \emph{V-shaped} at $\themean$ if there exists
parameters $a\leq b$ such that 
\begin{align*}
\util(\report) = 
\begin{cases}
a\cdot(\report - \themean) & \report < \themean\\
b\cdot(\report - \themean) & \report \geq \themean.
\end{cases}
\end{align*}
% $\util(\report) = a\cdot(\report -
% \themean)$ for $\report \leq \themean$ and $\util(\report) = b\cdot(\report
% - \themean)$ for $\report \geq \themean$.
\end{definition}

\begin{theorem}\label{thm:1d_opt}
For prior $\prior$ with prior mean $\priorMean$, 
any function $\util$ that is V-shaped at $\priorMean$ with parameters $a\leq b$
such that 
\begin{align*}
(b-a)\cdot\max\{\priorMean,1-\priorMean\}=1
\end{align*}
is optimal for \Cref{eq:1d_program_u} 
for any distribution $\marginalReport$ over the posterior means.
Moreover, the optimal objective value is 
\begin{align*}
\frac{1}{\OPTdenom} \int_{\priorMean}^{1} (\report - \priorMean) 
\dd \marginalReport(\report).
\end{align*}
\end{theorem}
The optimality of V-shaped utility functions is illustrated in \Cref{f:single}(b). 
Intuitively, by increasing the convex function $\util$ on both left side and right side of the prior mean 
while keeping the utilities at both the prior mean and the boundary points unchanged, 
the objective value weakly increases in \Cref{eq:1d_program_u}. 
The convexity of function $\util$ ensures that the maximum utility difference is attained when function $\util$ is linear on both left side and right side of the prior mean. 
Moreover, this operation relaxes the boundedness constraint at the boundary, 
and hence is feasible for \Cref{eq:1d_program_u} according to \cref{lem:bounded for single}. 
The proof is deferred to \cref{apx:thm:1d_opt}. 

% to its linear upper bound without violating the convexity constraint, 
% the utility evaluated at prior mean remains the same, 
% while the utility evaluated at all other posterior mean weakly increases. 
% Therefore, the objective value weakly increases in \Cref{eq:1d_program_u}. 
% Moreover, this operation relaxes the boundedness constraint at the boundary, 
% and hence is feasible for \Cref{eq:1d_program_u} according to \cref{lem:bounded for single}. 
% The proof is deferred to \cref{apx:thm:1d_opt}. 

Note that there is a kink in all optimal utility functions at prior mean $\priorMean$ for \Cref{eq:1d_program_u}. 
This kink is crucial for the optimality.
Suppose instead consider an alternative utility function $\hat{\util}$ that is continuously differentiable at prior mean~$\priorMean$. 
Then there exists a sufficiently small interval around the prior mean such that the utility function~$\hat{\util}$
evaluated within that interval is similar to a straight line. 
This implies that utility function $\hat{\util}$ provides almost zero incentives for effort
for any distribution over posterior means concentrated in that small interval, 
which is far from the optimal in terms of multiplicative approximation. 

The utility function that is optimal for \Cref{eq:1d_program_u} is not unique given \cref{thm:1d_opt}. 
However, all the optimal utility functions illustrated in \cref{thm:1d_opt} are essentially equivalent through rotations. 
That is, for any V-shaped functions $\util$ and $\util'$ at $\priorMean$ with parameters $a\leq b$ and $a'\leq b'$ respectively
that are optimal for \Cref{eq:1d_program_u}, 
function $\util'$ can be obtained by~$\util$ by adding a linear function as rotation. 
That is, there exists $c_1,c_0$ such that 
\begin{align*}
\util'(\themean) = \util(\themean) + c_1\themean+c_0, \quad\forall \themean\in[0,1].
\end{align*}
The main reason why such rotation does not affect the optimality for \Cref{eq:1d_program_u} is because 
both the prior mean and the distribution over posterior means have the same expectation when evaluated on linear functions
due to Bayesian consistency. 

% Moreover, observe that there is a kink in all optimal utility functions at prior mean $\priorMean$ for \Cref{eq:1d_program_u}. 
% This kink is crucial for the optimality.
% Suppose instead consider an alternative utility function $\hat{\util}$ that is differentiable at prior mean~$\priorMean$. 
% Then there exists a sufficiently small interval around the prior mean such that the utility function~$\hat{\util}$
% evaluated within that interval is similar to a straight line. 
% This implies that utility function $\hat{\util}$ provides almost zero incentives for effort
% any distribution over posterior means concentrated in that small interval, 
% which is far from the optimal in terms of multiplicative approximation. 

Since there are multiple utility functions that are optimal for \Cref{eq:1d_program_u}, 
there are also multiple scoring rules that can implement the optimal solution for incentivizing effort. 
As discussed above for the multiplicity of utility functions, 
all optimal scoring rules are essentially the same albeit rotations. 
Next we present one scoring rule $\score^*$ that is induced by one of the optimal utilities functions in \cref{thm:1d_opt}. 
Let $z_0, z_1\in[0,1]$ be the constant such that 
$z_0 = \min\{\frac{\priorMean}{1-\priorMean}, 1\}$
and $z_1 = \min\{\frac{1-\priorMean}{\priorMean}, 1\}$. 
We have 
\begin{align*}
\score^*(\report,\state) = 
\begin{cases}
z_0\cdot (1-\state) & \report < \priorMean\\
z_1\cdot \state & \report \geq \priorMean. 
\end{cases}
\end{align*}

\begin{figure}[t]
\centering
\begin{tikzpicture}[scale = 0.6]

\draw (0,0) -- (6.5, 0);
\draw (0, 0) -- (0, 6.5);

\draw [thick, dashed] plot (0, 4) -- (6, 0);
\draw [thick] plot (0, 0) -- (6, 6);

\draw [dotted] plot (0, 6) -- (6, 6);
\draw [dotted] plot (6, 0) -- (6, 6);

\draw (-0.5, 6) node {$z_1$};
\draw (-0.5, 4) node {$z_0$};
\draw (2.4, -0.6) node {$\priorMean$};
\draw [dotted] plot (2.4, 0) -- (2.4, 2.5);

\draw (-0.38, 0) node {$0$};
\draw (6, -0.5) node {$1$};
\draw (6, 0) -- (6, 0.2);

\end{tikzpicture}
\caption{\label{f:score_expost}
This figure illustrates the scoring rule $\score^*$ as a function of the realized state~$\state$.
The dashed line is the score function when the report $\report < \priorMean$, 
and the solid line is the score function when the report $\report \geq \priorMean$.}
\end{figure}

In scoring rule $\score^*$, parameter $z_0$ measures how large the score the agent gets when the realized state $\state$ is closer to $0$ 
when his report is smaller than the prior mean, 
and parameter $z_1$ measures how large the score the agent gets when the realized state $\state$ is closer to $1$ 
when his report is larger than the prior mean. 
Moreover, to implement scoring rule $\score^*$, it is sufficient to require the agent to make a binary report, 
i.e., whether the posterior mean is smaller or larger than the prior mean. 
This scoring rule is illustrated in \Cref{f:score_expost}.

A simple interpretation of scoring rules $\score^*$, when the state is binary (i.e., $\statespace=\{0,1\}$) is that the agent receives a score of zero if the prediction does not match the state. 
That is, the agent gets a score of zero when reporting $\report < \priorMean$ when the state is~$1$, or reporting $\report \geq \priorMean$ when the state is~$0$.
Moreover, in scoring rule $\score^*$, parameter~$z_0$ is increasing in prior mean $\priorMean$ 
while~$z_1$ is decreasing in $\priorMean$. 
Since $\priorMean$ is also the probability of state $1$ in binary state model, 
the agent receives a higher score if the prediction matches the state for a less likely state. 
That is, the agent receives a higher score for predicting a more surprising outcome correctly. 
When the state is not binary, the same intuition extends 
and the agent receives a higher score if the realized state locates further in the direction of the report, 
i.e., the received score is higher when the realized state is larger if $\report \geq \priorMean$, 
or when the realized state is smaller if $\report < \priorMean$.
% This is illustrated in \cref{fig:}.
% fits the prediction better (e.g., when $\report < \priorMean$ and $\state$ is smaller). 

A crucial practical feature of the optimal scoring rule for eliciting the mean is that 
it can be implemented with only the knowledge about the prior mean $\priorMean$, not the details about the
distribution~$\marginalReport$ over posterior means.
Therefore, such scoring rule can be implemented robustly in various situations even if the principal has uncertainty over the technology of the agent for acquiring additional information 
as long as the principal has good estimates about the prior distribution. 
In the online appendix, we show that the incentive loss in the optimal scoring rule with an estimate of the prior mean  
is small when the estimation loss for prior mean is small.

\section{Multi-dimensional Scoring Rules}
\label{sec:multi}
% \section{Eliciting the Mean}
% \label{sec:elicit mean}

In this section, we consider the multi-dimensional scoring rule for eliciting the mean. 
That is, the state space is $\statespace \subseteq \reals^n$. 
We show that the optimal scoring rule takes the form of a betting mechanism. 
In the special cases of single-dimensional state space, 
the betting mechanism simplifies to the scoring rule with V-shaped utilities characterized in \cref{sec:single assignment}.
Moreover, we provide simplifications and interpretations of the optimal betting mechanisms in symmetric environments with multi-dimensional state space. 
In the general non-symmetric environments, 
as illustrated in \cref{sec:peer grading}, 
the standard approach in both theory and practice of scoring the agents separately in each dimension is not a good approximation to the optimal multi-dimensional scoring rule, 
and we show that a simple max-over-separate scoring rule is approximately optimal.

\subsection{Betting Mechanism for Eliciting the Mean}
\label{subsec:betting for mean}
As interpreted in \cref{subsec:formal_program}, 
the canonical scoring rule (\cref{def:canonical}) 
can be viewed as letting the agent take bets on hyperplanes for scores. 
Based on this interpretation, we introduce the betting mechanism (\cref{def:betting}), 
and show that betting mechanisms can be viewed as an indirect implementation of the optimal scoring rules for eliciting the mean. 

% In order to elicit the mean from the agent, the betting mechanism cannot allow the agent to choose an arbitrary score for any state. 
% As shown in \cref{subsec:proper}, in order for agents with different posteriors with the same mean to have the same incentives, 
% the score of the agent is a linear function of the state variables in proper scoring rules.
% Thus in the betting mechanism for eliciting the mean, we also restrict the action space of the agent to only choosing linear coefficients,
% and show that the betting mechanism for eliciting the mean
% can be converted to a proper scoring rule for eliciting the mean. 

\begin{definition}\label{def:betting}
A mechanism is a \emph{betting mechanism for eliciting the mean}
with parameter $c\in [0,1]$ and a normalization function $\kappa: \statespace \to \reals$, 
if when the prior is $\prior$, 
the agent chooses an $n$-dimensional coefficients $\sg\in \reals^n$ and a shift parameter $b\in\reals$ such that 
\begin{itemize}
    \item fixed expected score at prior mean: $\sg\cdot\priorMean+b\leq c$;
    \item bounded score in the space: $\sg\cdot\state+b+\kappa(\state)\in[0, 1], \forall \state\in\statespace$.
\end{itemize}
The agent receives score $\sg\cdot\state+b+\kappa(\state)$ when the realized state is $\state$.
\end{definition}
The betting mechanism can be viewed as the agent taking bets on hyperplanes with parameters $\sg$ and $b$. 
Note that given any choice of $\sg$ and $b$, 
the expected score of the agent given posterior $\posterior$
is $\sg\cdot\mean{\posterior} +b+\expect{\state\sim\posterior}{\kappa(\state)}$. 
It is easy to verify that the optimal choice of $\sg$ and $b$
only depends on the posterior mean, 
and hence the betting mechanism can be converted to a proper scoring rule for mean, 
where the principal elicits the mean from the agent 
and optimize the score for the agent. 
Therefore the following claim holds with proof omitted. 
\begin{claim}
The betting mechanism for eliciting the mean can be converted to a proper scoring rule for eliciting the mean. 
\end{claim}
\begin{theorem}
The optimal scoring rule can be implemented as a betting mechanism for eliciting the mean.
\end{theorem}

\begin{proof}
For any proper scoring rule for eliciting the mean with convex utility function~$\util$ and normalization function $\kappa$,
consider the betting mechanism with parameter $c=\util(\priorMean)$ and the same $\kappa$. 
Let~$\hat{\util}$ be the utility function induced by the betting mechanism. 
It is easy to verify that $\hat{\util}(\priorMean) = \util(\priorMean)$.
Moreover, for any posterior mean $\mean{\posterior}$, one feasible choice for the agent in the betting mechanism is to select the subgradient of the utility function $\util$, 
and obtain expected score at least $\util(\mean{\posterior})+\expect{\state\sim\posterior}{\kappa(\state)}$. 
Therefore, $\hat{\util}(\mean{\posterior})$ is weakly higher than $\util(\mean{\posterior})$ for any posterior mean $\mean{\posterior}$,
and the objective value of the betting mechanism is weakly higher. 
\end{proof}

In the special case of single-dimensional state, 
the bet the agent takes is binary:
either reports posterior mean that is smaller than the prior mean 
to bet on the hyperplane that maximizes the score for states that are closer to 0, 
or reports posterior mean that is larger than the prior mean 
to bet on the hyperplane that maximizes the score for states that are closer to 1. 
This is consistent with the characterization in \cref{thm:1d_opt}. 

In order to design the optimal betting mechanism for multi-dimensional state space, 
the principal need to compute the optimal choice of $\kappa$ and $c$ based on the distribution over posteriors. 
In general there is no simple characterization for the optimal choice of $\kappa$ or $c$. 
In the next proposition, we show that a simple choice of $\kappa(\cdot)=0$ and $c=\frac{1}{2}$ is approximately optimal. 
The proof of \cref{thm:betting mean approx} is provided in \cref{apx:thm:betting mean approx}
% Without the knowledge of optimal normalization, setting $\kappa(\cdot)=0$ can give a $2$-approximation.

\begin{proposition}\label{thm:betting mean approx}
The betting mechanism for eliciting the mean with $c=\frac{1}{2}$ and $\kappa(\cdot)=0$ obtains at least half of the optimal.
\end{proposition}

In \cref{subsec:optimal symmetric multi d}, we further show that in symmetric environments, 
the optimal betting mechanisms can be viewed as scoring rules with generalized V-shaped utility functions in multi-dimensional space.

\subsection{Optimal Scoring Rules for Symmetric Distributions}\label{subsec:optimal symmetric multi d}

This section characterizes the optimal multi-dimensional scoring rule
when the distribution over posteriors is symmetric about its center.
The result obtained in the single-dimensional setting extends to multi-dimensional state spaces by extending the definition of the V-shaped utility function to multi-dimensional environments,
i.e., \Cref{eq:program} is optimized by a symmetric V-shaped utility function (\cref{def:symmetric v-shaped scoring rule}).  
This characterization affords a simple interpretation for
rectangular report and state spaces.
Specifically, the optimal scoring
rule can be calculated by taking the maximum score over optimal
single-dimensional scoring rules for each dimension, i.e., it is a
\mos scoring rule.  As these single-dimensional scoring rules depend
only on the prior mean, so does the optimal multi-dimensional scoring
rule.  We first give the characterization and then give the
interpretation.

\begin{definition}\label{def:symmetric distribution}
A $n$-dimensional distribution $\marginalReport$ is \emph{center symmetric} if
there exists a center in the report space, 
i.e., $\centerpoint \in \reportspace$ 
such that for any $\report \in \reportspace$, 
$\marginalReport(\centerpoint - \report) = \marginalReport(\centerpoint + \report)$. 
\end{definition}

Note that for any center symmetric distribution $\marginalReport$ over
posterior means, the mean of the prior coincides with the center of
the space, i.e., $\priorMean = \centerpoint$.  The following
definition generalizes symmetric V-shaped functions to
multi-dimensional state and report spaces.
Let $\boundaryReport$ be the boundary of the report space $\reportspace$. 

\begin{definition}\label{def:symmetric v-shaped scoring rule}
A function $\util$ is \emph{symmetric V-shaped} in report
space $\reportspace = \conv(\statespace)$ with non-empty interior and center $\centerpoint$ if
\begin{itemize}
\item utility is zero at the center, i.e., $\util(\centerpoint) = 0$;

\item utility is $\sfrac 1 2$ on the boundary, i.e., 
$\util(\report) = \sfrac 1 2$ for $\report \in \boundaryReport$; and

\item all other points
linearly interpolate between the center and the boundary, i.e.,
$\util(\alpha \cdot \report + (1-\alpha) \cdot \centerpoint) =
\frac{\alpha}{2}$ for any $\alpha \in [0,1]$ and $\report \in
\boundaryReport$.
\end{itemize}
\end{definition}

V-shaped utility functions on convex and center symmetric spaces are
bounded and convex, i.e., they are feasible solutions to
\Cref{eq:program}. The proof of \Cref{lem:symm_nd_bounded} is deferred to \Cref{sec:appendix-proof-symm_nd_bounded}.

\begin{lemma}\label{lem:symm_nd_bounded}
  For any convex and center symmetric report and state space
  $\reportspace = \statespace$ with non-empty interior, the center
  symmetric utility function is convex and bounded for $\scorebound =
  1$.
\end{lemma}

We show that the expected utility function of the optimal betting mechanism corresponds to a symmetric V-shaped function in center symmetric environments.
The following theorem is proved by following a standard approach in
multi-dimensional mechanism design, e.g., \citet{arm-96} and
\citet{HH-15}.  The problem is relaxed onto single-dimensional paths,
solved optimally on paths, and it is proven that the solution on paths
combine to be a feasible solution on the whole space.  Note that in
relaxing the problem onto paths, constraints on pairs of reports that
are not on the same path are ignored. The full proof of \Cref{thm:opt_symm_nd} is deferred to \Cref{sec:appendix-proof-opt_symm_nd}.
Similar to the single dimensional V-shaped scoring rule, the implementation of multi-dimensional V-shaped scoring rule only requires the knowledge of the prior mean $\priorMean$ for the principal. 

\begin{theorem}\label{thm:opt_symm_nd}
For any center symmetric distribution $\marginalReport$ over posterior
means in convex report and state space $\reportspace = \statespace$,
the optimal solution for \Cref{eq:program} is symmetric V-shaped.
\end{theorem}

In the remainder of this section we give an interpretation of scoring
rules that correspond to V-shaped utility functions on rectangular
report and state spaces.  On such spaces, these optimal scoring rules
can be implemented as the maximum over separate scoring rules (for
each dimension). 
Intuitively, the \mos scoring rule rewards the agent only on the dimension that the agent will receive highest expected score according to his posterior belief.

% % Since this section focuses on scoring rules where the agent directly
% % reports the marginal means of the distribution, 
% Next we formally introduce the revelation version of \car scoring rules, which we refer to as
% \mos scoring rules.  It is easy to verify that any \mos
% scoring rule can be implemented as a \car scoring rule, 
The definition of \mos scoring rule is formally introduced in \cref{def:max over separate},
and it is easy to verify that a \mos
scoring rule is proper and bounded if is based on single dimensional
scoring rules that are proper and bounded.
In particular, in \mos scoring rules, we only consider the case where the normalization function $\kappa(\state)$ is set to be a constant for all $\state\in\statespace$. 

\begin{definition}[max-over-separate]\label{def:max over separate}
A scoring rule $\score$ is \mos if there exists single dimensional
scoring rules $(\sdscore_1, \dots, \sdscore_{\numasm})$ 
and constant $\sdkappa\in\reals$ 
such that
\begin{enumerate}
\item for any dimension $i$, $\sdscore_i(\report_i, \state_i)=
  \sdutil_i(\report_i) +\sdsg_i(\report_i) \cdot
  (\state_i-\report_i)+\sdkappa$
  where $\sdsg_i(\report_i)$ is a subgradient of convex function $\sdutil_i(\report_i)$;

\item \label{item:mos-score} the score is $\score(\report, \state) = \sdscore_i(\report_i,
  \state_i)$ where $i = \argmax_{j} \sdutil_i(\report_i)$.
\end{enumerate}
\end{definition}

Condition 1 in \cref{def:max over separate} ensures that each single-dimensional scoring rule chosen in \mos scoring rule is proper for eliciting the mean. 
The overall incentives of \mos is ensured by condition 2
and the fact that 
% equality of $\sdscore_j(\report_j, \report_j)$ (from
% condition~\ref{item:mos-score}) and 
$\expect{\state_j \sim \posterior_j}{\score_j(\report_j, \state_j)}
= \sdutil_i(\report_i) + \sdkappa$ 
for any marginal posterior distribution $\posterior_j$ on dimension $j$ with mean~$\report_j$.  
% Specifically, since the function $\sdkappa_j$ is a constant function of the state,
% all posteriors $\posterior_j$ with the same mean induce the same
% expected score.

We show that in rectangular report and state spaces, 
scoring rules with symmetric V-shaped utility functions, which are shown to
be optimal by \Cref{thm:opt_symm_nd}, can be interpreted as \mos scoring rules.
This can be shown by algebraically calculating the expected utility functions for \mos scoring rules, 
and hence the proof of the following proposition is deferred to \cref{apx:lem_vshape}.

\begin{proposition}\label{lem:v shape is max over separate}
Symmetric V-shaped function $\util$ in $\numasm$-dimensional rectangle
report and state space $\reportspace = \statespace =
\bigtimes_{i=1}^{\numasm}[a_i,b_i]$ with function $\kappa(\state) =
\sfrac{1}{2}$ can be implemented as \mos scoring rule with single
dimensional bounded proper scoring rules $\{\sdscore_i\}_{i=1}^n$ where
\begin{align*}
    \sdscore_i(\report_i,\state_i) &= 
    \begin{cases}
    -\frac{1}{b_i-a_i}(\state_i-\Pmeani) + \frac{1}{2} & \text{for $\report_i \leq \Pmeani$,}\\
    \frac{1}{b_i-a_i}(\state_i-\Pmeani) + \frac{1}{2} & \text{for $\report_i \geq \Pmeani$,}
    \end{cases}
\end{align*}
where $\Pmeani=\sfrac{(a_i+b_i)}{2}$ is the $i^{th}$ coordinate of the prior mean $\priorMean$.
\end{proposition}

% \cref{lem:v shape is max over separate} implies that 
% for any center symmetric distribution $\marginalReport$ over posterior
% means in rectangular report and state space $\reportspace
% = \statespace$, a \mos scoring rule is optimal.
\begin{corollary}\label{cor:opt_symm_nd}
For any center symmetric distribution $\marginalReport$ over posterior
means in rectangular report and state space $\reportspace
= \statespace$, a \mos scoring rule is optimal.
\end{corollary}

% \begin{definition}\label{def:max over symmetric}
% A scoring rule $\score$ is \mosym if it is \mos over single dimensional scoring rules for each assignment
% \begin{equation*}
% \sdscore_i(\report_i, \state_i)
% = \begin{cases}
% {\priorMean}_i -\state_i & \report_i < {\priorMean}_i,\\
% \state_i - {\priorMean}_i & \report_i \geq {\priorMean}_i.
% \end{cases}
% \end{equation*}
% \end{definition}

Finally, these \mos scoring rules have an indirect \car
implementation where the agent reports the dimension to be scored on
and the mean for that dimension.  This indirect implementation has a
practical advantage that 
when the communication between the principal and the agent is costly
since in $\numasm$-dimensional spaces, it requires
only reporting two rather than $\numasm$ numbers.\footnote{In the application of exam grading, it also implies that it is sufficient for the instructor to only grades one question instead of $n$ questions.}
Note that \car and \mos are essentially the same scoring rule, with different implementations.

\begin{definition}[choose-and-report]\label{def:choose and report}
A scoring rule $\score$ is \car if there exists single dimensional
scoring rules $(\sdscore_1, \dots, \sdscore_{\numasm})$ such that the
agent reports dimension~$i$ and mean value $\report_i$, and receives
score $\score((i,\report_i), \state) = \sdscore_i(\report_i, \state_i)$.
\end{definition}

An agent's optimal strategy in the \car scoring rule for proper
single-dimensional scoring rules
$(\sdscore_1,\ldots,\sdscore_{\numasm})$ is to choose the dimension
$i$ with the highest expected score according to the posterior, i.e., $i = \argmax_{j} \expect{\state_j \sim
\posterior_j}{\sdscore_j(\mean{\posterior_j}, \state_j)}$, 
and to report the mean of the posterior for that dimension, i.e.,
$\mean{\posterior_i}$. 

For the \car scoring rule that corresponds to the optimal V-shaped utility function of \Cref{eq:program}, 
the dimension $i$ that maximizes the expected utility is the dimension with posterior mean $\mean{\posterior_i}$ that is furthest to the prior mean $\mean{\prior_i}$, 
i.e., $\abs{\mean{\posterior_i} - \mean{\prior_i}}$ is maximized. 
Therefore, based on the interpretation of \car scoring rules, 
the agent is only scored on the dimension with the most surprising observation.

% As described above, the advantage of such an
% indirect scoring rule is that it only requires the agent to report two
% values to the principal. 
% \Cref{lem:proper and bounded for mos} illustrate a nice properties of \car scoring rules, 
% with proof deferred to \Cref{sec:appendix-proof-proper and bounded for mos}.
% \begin{lemma}\label{lem:proper and bounded for mos}
% The \car scoring rule $\score$ defined by proper and bounded
% single-dimensional scoring rules $(\sdscore_1, \dots, \sdscore_{\numasm})$
% is itself proper and bounded.
% \end{lemma}

% % The proof of \Cref{lem:proper and bounded for mos} is deferred to \Cref{sec:appendix-proof-proper and bounded for mos}.
%\subsection{Optimal Scoring Rules for Symmetric Distributions}

\subsection{Inapproximation by Separate Scoring Rules}
\label{subsec:separate}

In general asymmetric environments, 
one way to design the scoring rule for an $\numasm$-dimensional space
is to average independent scoring rules for the marginal distributions
of each dimension.  In this section we show that the worst-case
multiplicative approximation of scoring each dimension separately and
scoring optimally is $\Theta(\numasm)$.  
% Moreover, the upperbound
% $O(\numasm)$ holds for general correlated report distributions, while
% the lowerbound $\Omega(\numasm)$ holds for independent distributions.
The main idea of this large gap is already illustrated in \cref{sec:peer grading}
in the application of exam grading
when the probability of acquiring an informative signal is small for each dimension. 
Hence, the proof of \Cref{thm:score separately multiplicative} is deferred to \Cref{sec:appendix-proof-score separately multiplicative}.

% We instantiate this gap by the following two observations. 
% % Suppose there are $\numasm$ i.i.d.\ assignments. 
% % with posterior distribution $\marginalReport$. 
% First, the multiplicative approximation ratio can be 
% as large as $O(\numasm)$. %$\frac{\numasm}{2\ln \numasm}$. 
% Secondly, there exists a distribution over posteriors
% such that the objective value of the optimal scoring rule is 
% close to $\frac{1}{2}$ while 
% the objective value of scoring separately approaches 0.  
% % the additive approximation ratio can be 
% % as large as $\frac{1}{2} - \frac{1+\ln \numasm}{\numasm}$. 

\begin{definition}\label{def:separate score}
A scoring rule $\score$ is a \emph{separate scoring rule} 
if there exists single dimensional scoring rules $(\score_1, \dots, \score_{\numasm})$ such that 
$\score(\report, \state)
= \sum_i\score_i(\report_i, \state_i)$. 
\end{definition}

\begin{proposition}\label{thm:score separately multiplicative}
  In $\numasm$-dimensional rectangular report and state spaces, the worst-case
  approximation factor of scoring each dimension separately is
  $\Theta(\numasm)$.
\end{proposition}

\cref{thm:score separately multiplicative} highlights the importance of linking incentives across different dimensions to incentivize effort.
This concept of linking incentives has been previously recognized in \citet{JS-07}, where it was applied to allocating items to maximize welfare without transfers, which is an environment with pure adverse selection.
Our work complements this by extending the philosophy of linking incentives to an environment with moral hazard and endogenous information.
Moreover, it's worth noting that the main driver for the importance of linking incentives in our model is different from \citet{JS-07}. 
In our case, the key reason for linking incentives across dimensions is to prevent the situation where an uninformed agent who does not exert effort can make multiple uneducated guesses for different dimensions. This scenario makes it challenging to distinguish between such an uninformed agent and an agent who genuinely exerted effort, particularly when the signals arrive with only small probabilities in each dimension. The ability to separate those two types of agents is essential to ensure effective effort incentivization in our context.

\subsection{Approximately Optimal Scoring Rules for General Distributions}
\label{subsec:max over separate}

% The \mos scoring rule restricts agent's  bets to contracts linear in one dimension of the rectangular report space. 
In \Cref{subsec:optimal symmetric multi d}, we have shown that the \mos scoring rule is optimal for symmetric distributions. 
When the distribution is not symmetric, 
although the \mos scoring rule may not be optimal, 
we show that the optimal
\mos scoring rule always outperforms the separate scoring rule.
Moreover, there exists a \mos scoring rule that is an $8$-approximation to the optimal for any asymmetric and
possibly correlated distribution over a high dimensional rectangular space,
and the design of this approximately optimal scoring rule only requires the knowledge of the prior mean, not the distribution over posteriors. 
In the online appendix, we further relax this assumption and design approximately optimal scoring rules when the principal only have imprecise estimates of the prior mean.

To show the approximate optimality of \mos scoring rules, we symmetrize the distribution over posteriors, and construct a V-shaped scoring rule on the symmetrized distribution. 
This V-shaped scoring rule can be implemented as a \mos scoring rule on the original problem, which only requires the knowledge of prior mean. 

\begin{theorem}\label{thm:approx_max_over_separate}
For any distribution $\marginalReport$ over posterior means in
$\numasm$-dimensional rectangular report and state space $\reportspace
= \statespace = \bigtimes_{i=1}^{\numasm}[a_i,b_i]$, the utility
function $\util$ of optimal \mos scoring rule for \Cref{eq:program}
achieves at least $\sfrac{1}{8}$ of the optimal objective value,
i.e. $\objfunc(\utility, \marginalReport)\geq \sfrac{1}{8}\cdot
\OPT(\marginalReport, \scorebound, \statespace)$.
\end{theorem}

\paragraph{Interpretations}
Before delving into the proof of \cref{thm:approx_max_over_separate}, it is important to briefly discuss the interpretation of the approximation factor of 8. 
At first glance, it may appear that the factor of 8 is large, suggesting that the principal might incur significant losses by adopting the suboptimal solution. 
However, it is important to note that our analysis takes a worst-case approach, and the actual gap between the optimal and the proposed \mos scoring rule in real-life scenarios can be much smaller. 
For instance, as demonstrated in \cref{thm:opt_symm_nd}, 
the gap is only 1, i.e., \mos scoring rules are optimal, for symmetric distributions. 
Furthermore, we encourage readers not to interpret the exact approximation factor literally, 
but rather focus on the relative comparison among simple scoring rules based on their approximation factor. 
In this context, the primary comparison is between \mos scoring rules and separate scoring rules. 
The former maintains a constant approximation factor irrespective of the number of dimensions, whereas the latter exhibits linear degradation in the number of dimensions (\cref{thm:score separately multiplicative}).\footnote{In fact, in \cref{apx:thm:mos better}, we show that for any instance and any separate scoring rule, there always exists a \mos scoring rule that outperforms it. } 
This further highlights the crucial role of linking incentives across different dimensions in promoting effort in high-dimensional problems \citep[c.f.,][]{JS-07}.

% Before the proof of \cref{thm:approx_max_over_separate}, we would like to briefly discuss how to interpret the approximation factor of 8. 
% At first glance, the factor of 8 seems large and the principal potentially loses a lot by resorting to the suboptimal solution. 
% One justification is that our analysis is an worst-case approach, and the actual gap between the optimal and the proposed \mos scoring rule for real-life environments can be much smaller. 
% For example, as illustrated in \cref{thm:opt_symm_nd}, the gap is 1 for symmetric distributions. 
% More crucially, we encourage the readers to not take the exact approximation factor literally, and instead focus on the relative comparison among simple scoring rules according to their approximation factor. 
% Here the main comparison is between \mos scoring rules and separate scoring rules. 
% The approximation factor of the former is a constant regardless of the number of dimensions, while the approximation factor of the latter degrades linearly in the number dimensions. 
% This highlights the fact that the incentives across different dimensions is crucial for incentivizing effort in high-dimensional problems. 

\paragraph{Proof Sketch}
In the following discussion, we assume without loss of generality that
$\Pmeani \geq \sfrac{(a_i + b_i)}{2}$ for every dimension~$i$. 
To prove \cref{thm:approx_max_over_separate}, we will show that the \mos scoring rule with utility functions 
\begin{align*}
\util_i(\report_i) = \frac{1}{2}+\frac{1}{2(\Pmeani-a_i)} \abs{\report_i-\Pmeani}
\end{align*}
for each dimension $i$
is approximately optimal. 
In particular, for each dimension $i$, the chosen utility function is optimal for the single-dimensional scoring rule problem with prior $\Pmeani$. 

\begin{figure}[t]
    \centering
    \begin{tikzpicture}
[scale = 2,
pnt/.style={circle, draw=black!100, fill=black!100, thick, inner sep = 0.001 mm, minimum size = 0.5 mm}]
\shade [left color = black!7, right color = black!7] plot coordinates{
(0, 0) (1, 0) (1, 1/2) (4/3, 1/2) (4/3, 3/2) (1/3, 3/2) (1/3, 1) (0, 1) (0, 0)
};

\draw [->](0, 0) -- (5/3, 0);
\draw [->](0, 0) -- (0, 5/3);
\draw (0, 0) -- (1, 0);
\draw (0, 0) -- (0, 1);
\draw (1, 0) -- (1, 1);
\draw (0, 1) -- (1, 1);
\draw (1/3, 1/2) -- (1/3, 3/2);
\draw (1/3, 1/2) -- (4/3, 1/2);
\draw (1/3, 3/2) -- (4/3, 3/2);
\draw (4/3, 1/2) -- (4/3, 3/2);
\draw [dotted] (0, 0) -- (4/3, 3/2);
\draw [dotted] (1, 0) -- (1/3, 3/2);
\draw [dotted] (0, 1) -- (4/3, 1/2);

\draw [line width = 2pt, color=gray, dashed] plot coordinates{
(0, 0) (1, 0) (4/3, 1/2) (4/3, 3/2) (1/3, 3/2) (0, 1) (0, 0)
};

\draw [very thick] plot coordinates{
(0, 0) (4/3, 0) (4/3, 3/2) (0, 3/2) (0, 0)
};

\node at (2/3, 3/4) [pnt] {};
\draw (2/3-0.1, 3/4+0.1) node {$\priorMean$};
\draw (1, -0.1) node {$1$};
\draw (-0.1, 1) node {$1$};
\draw (-0.1, -0.1) node{$0$};

\end{tikzpicture}
    \caption{\label{fig:convHull} This figure depicts a
      two-dimensional state space.  The state space $\statespace =
      [0,1]^2$ and its point reflection around the prior mean
      $\priorMean$ are shaded in gray.  The extended report and state
      space are depicted by the region within the thick black
      rectangle.  
      %The convex hull of the shaded region, enclosed by a dashed gray perimeter, is the state space $\extconvstatespace$ in \Cref{sss:lem-obj-ratio}.
      }
\end{figure}
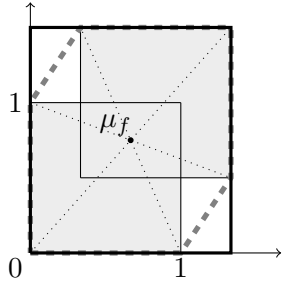

The proof of \cref{thm:approx_max_over_separate} introduces the following constructs:
\begin{itemize}
\item The {\em extended report and state space} are
  $\extreportspace = \extstatespace =
  \bigtimes_{i=1}^{\numasm} [a_i, 2\Pmeani-a_i]$.  These are
  rectangular and contain the original report and state spaces
  $\reportspace = \statespace$.  See \Cref{fig:convHull}.

\item The {\em symmetric extended distribution of $\marginalReport$}
  on the extended report space is $\extMarginalReport(\report) =
  \frac{1}{2}(\marginalReport(\report)+\marginalReport(2\priorMean -
  \report))$.  Note in this definition that the original distribution $\marginalReport$
  satisfies $\marginalReport(\report) = 0$ for any $\report \in
  \extreportspace \setminus \reportspace$.

 \end{itemize}
\Cref{thm:approx_max_over_separate} now follows by directly combining the following lemmas, with proofs provided in \cref{app:opt-separate,sss:lem-obj-ratio}.
Essentially, by symmetrizing the distribution $\marginalReport$ to~$\extMarginalReport$, 
we show that the optimal objective values are close between $\marginalReport$ and $\extMarginalReport$.
Moreover, by adopting the optimal scoring rule for symmetric distribution $\extMarginalReport$, 
which is the \mos scoring rule we described above, 
the loss in objective value is small when the true distribution is $\marginalReport$. 
Therefore, \mos scoring rule is approximately optimal for the original distribution~$\marginalReport$.

\begin{lemma}
  \label{lem:opt-separate}
  Evaluated on any distribution over posterior means
  $\marginalReport$, the optimal \mos scoring rule for the
  distribution $\marginalReport$ and the state space $\statespace$ is
  at least as good as the optimal scoring rule for the extended
  distribution $\extMarginalReport$ and the extended state space
  $\extstatespace$.
\end{lemma}

\begin{lemma}
  \label{lem:extopt-to-orig}
  The symmetric optimizer $\extutil$ for the symmetric extended
  distribution $\extMarginalReport$ and extended state
  space $\extstatespace$ attains the same objective
  value on the original distribution $\marginalReport$, i.e.,
  $\objfunc(\extutil,\marginalReport) =
  \OPT(\extMarginalReport,\scorebound, \extstatespace)$.
\end{lemma}

%% \begin{lemma}
%%   \label{lem:extopt-relax}
%%   The original and relaxed programs achieve the same optimum on
%%   the symmetric extended distribution $\extMarginalReport$ and the extended
%%   report and state space $\extreportspace = \extstatespace$, i.e.,
%%   $\OPT(\extMarginalReport,\scorebound, \extstatespace) =
%%   \OPTrelax(\extMarginalReport,\scorebound, \extstatespace)$.
%% \end{lemma}

\begin{lemma}
  \label{lem:symmetry-in-ext}
  On extended state space $\extstatespace$, the optimal value of
  \Cref{eq:program} for the symmetric extended
  distribution $\extMarginalReport$ is at least half that for the
  original distribution $\marginalReport$, i.e.,
  $\OPT(\extMarginalReport,\scorebound, \extstatespace) \geq
  \frac{1}{2} \OPT(\marginalReport,\scorebound, \extstatespace)$.
\end{lemma}

\begin{lemma}\label{lem:obj ratio}
  For any distribution over posterior means $\marginalReport$, the
  optimal value of \Cref{eq:program} on the extended state space
  $\extstatespace$ is at least a quarter of that of the original state
  space $\statespace$, i.e.,
  $\OPT(\marginalReport,\scorebound, \extstatespace) \geq \frac{1}{4}
  \OPT(\marginalReport,\scorebound, \statespace)$ or equivalently
  $\OPT(\marginalReport,4\scorebound, \extstatespace) \geq
  \OPT(\marginalReport,\scorebound, \statespace)$.
\end{lemma}

\section{Elicitation of Full Distribution}
\label{sec:full}
In this section, we consider the extension where the state space is finite, i.e., $\statespace=\{\state_1,\dots,\state_m\}$, 
and the principal can elicit the full posterior distribution from the agent instead of just the statistics such as the posterior mean. 
In this case, by slightly overloading the notations, 
we also use $\state$ to represent a $m$-dimensional unit vector, 
where the $i$th coordinate is $1$ if and only if $\state = \state_i$. 
The following lemma characterizes the set of proper scoring rules for eliciting the full distribution. 

\begin{definition}[Proper]\label{def:proper_full}
A scoring rule $\score(\report, \theta)$ is proper for eliciting full distribution if 
for any distribution~$\posterior$
and report $\report \in \Delta(\statespace)$, 
we have 
\begin{align*}
\expect{\state\sim\posterior}{\score(\posterior, \theta)}
\geq \expect{\state\sim\posterior}{\score(\report, \theta)}.
\end{align*}
\end{definition}

% \subsection{Proper Scoring Rules for Eliciting the Full Distribution}
% For eliciting the full distribution, we consider the case that the state space $\statespace$ is finite, i.e., $\statespace=\{\state_1,\dots,\state_m\}$. 
% In this case, by slightly overloading the notations, 
% we also use $\state$ to represent a $m$-dimensional unit vector, 
% where the $i$th coordinate is $1$ if and only if $\state = \state_i$. 

\begin{lemma}[\citealp{Mcc-56}]
\label{thm:characterization of report full dist}
For any finite state space $\statespace$ and corresponding report space $\reportspace$,
a scoring rule $\score$ is proper for eliciting the full distribution in space $\statespace$
if and only if there exists a convex function $\util: \reportspace \to \reals$ such that
for any report 
$\reportFull \in \reportspace$ 
and any state $\state \in \statespace$,
we have 
$$
 \score(\reportFull, \state)= \util(\reportFull) +\sg(\reportFull) \cdot (\state-\reportFull),
$$
where $\sg(\reportFull) \in \subgradient \util(\reportFull)$ is a subgradient of $\util$.\footnote{Note that in general the principal can also provide additional state-dependent rewards $\kappa(\state)$ without affecting the incentives. 
However, for eliciting the full distribution, for any $\util(\reportFull),\sg(\reportFull) \in \subgradient \util(\reportFull)$ and $\kappa(\state)$, 
there exists $\util^{\dagger}(\reportFull),\sg^{\dagger}(\reportFull) \in \subgradient \util^{\dagger}(\reportFull)$
such that $\util(\reportFull) +\sg(\reportFull) \cdot (\state-\reportFull) + \kappa(\state)= \util^{\dagger}(\reportFull) +\sg^{\dagger}(\reportFull) \cdot (\state-\reportFull)$.
Therefore, considering the representation without $\kappa$ is without loss for eliciting the full distribution.}
\end{lemma}
 
Applying the characterization in \cref{thm:characterization of report full dist}, for any $\posterior\in\reportspace$, we have
\begin{align*}
\expect{\state \sim \posterior}{\score(\mean{\posterior}, \state)}
= \expect{\state \sim \posterior}{\util(\reportFull) +\sg(\reportFull) \cdot (\state-\reportFull)}
= \util(\reportFull) +\sg(\reportFull) \cdot \expect{\state \sim \posterior}{(\state-\reportFull)}
= \util(\reportFull),
\end{align*}
and thus the optimization program can be simplified as follows. 
\begin{align}
\OPT(\marginalReport, \scorebound, \statespace) 
= \max_\util \qquad &\int_{\reportspace} \util(\posterior)\marginalReport(\posterior)
\ \textrm{d}\posterior - \util(\prior)\label{eq:program full}\\
\text{s.t.} \qquad 
&\util\text{ is a convex function,}\nonumber\\ 
&\sg(\posterior) \in \subgradient \util(\posterior), \quad \forall \posterior \in \reportspace, \nonumber \\
&\util(\posterior)
+\sg(\posterior) \cdot (\posterior-\report) \in [0, \scorebound], 
\quad \forall \posterior \in \reportspace, \state\in \statespace,\nonumber\\
&\reportspace = \Delta(\statespace)\nonumber.
\end{align}

\subsection{Optimal Betting Mechanism}\label{sec:report full}

In this section, we characterize the optimal scoring rule for eliciting the full distribution. 
We first formally define the \emph{betting mechanism}, which is a non-revelation mechanism to help explain the structure of the mechanism. 
The main difference compared to the betting mechanism for eliciting the mean is that instead of taking bets on hyperplanes, 
the betting mechanism for eliciting the full distribution now takes bets on the set of states. 
% Note that according to the revelation principle it can be converted to a revelation mechanism, or equivalently a proper scoring rule, 
% by asking the agent to truthfully report the posterior, and then optimize the choice of the agent in the non-revelation betting mechanism based on the report.
Again we focus on the case that the score bound is $\scorebound=1$.

% We characterize the optimal scoring rule for eliciting the full distribution and a simple scoring rule which is approximately optimal. In particular, we consider the following program:
% \begin{align}
% \max_\score \qquad 
% &\expect{G \sim f, \theta \sim G}{\score(G, \theta) - \score(D, \theta)}\label[scoringruleprogram]{eq:program_full}\\
% \text{s.t.}\qquad &\score \text{ is a proper scoring rule for eliciting the full distribution},\nonumber\\
% &\score \text{ is bounded by $1$
% in space $\reportspace \times \Theta$}.\nonumber
% \end{align}

% We define the betting scoring rule as follows. 
\begin{definition}
In the betting mechanism with parameter $c$, 
when the prior is $\prior$, 
the agent chooses the score $s(\theta) \in [0,1]$ for each state $\theta \in \Theta$ satisfying $\expect{\theta \sim \prior}{s(\theta)} \leq c$. 
The agent receives score $s(\theta)$ when the realized state is $\state$.
% A scoring rule for eliciting the full distribution is a betting scoring rule if there exists a fixed constant $c \in [0,1]$ such that the agent chooses the score $s(\theta) \in [0,1]$ for each state $\theta \in \Theta$ satisfying $\expect{\theta \sim D}{s(\theta)} \leq c$.
%if for any state $\theta$, there exist two disjoint subsets $B,B' \subset [n]$ such that $|B\cup B'|\geq n-1$,
% \begin{align*}
% S(G,\theta) = 
% \begin{cases}
% 1, \text{if $G \in B$},\\
% 0, \text{if $G \in B'$}\\
% s, \text{if $G \not\in B\cup B'$}.
% \end{cases}
% \end{align*}
\end{definition}

Note that the revelation version of the betting mechanism is a feasible solution of Program~(\ref{eq:mainprogram}) as the score is obviously bounded. 
Given this betting mechanism, the optimal choice of the agent is actually quite simple. 
Given any posterior belief $\posterior$, 
the problem of maximizing 
$\expect{\theta \sim \posterior}{s(\theta)}$. 
subject to the constraints that 
$\expect{\theta \sim \prior}{s(\theta)} \leq c$
and $s(\theta) \in [0,1]$ for any $\theta$
is exactly the fractional knapsack problem, which can be solved by greedily assigning score $s(\theta)$ to 1 according the ratio $\frac{\posterior(\state)}{\prior(\state)}$. 
The following lemma is folklore and hence we omit the proof of it.
\begin{lemma}\label{lem:frac knapsack}
For any posterior belief $\posterior$, the agent's optimal choice is to select $\state^*, s(\theta^*)\in[0,1]$ and $\statespace_1\in \statespace\backslash\{\state^*\}$
such that
\begin{enumerate}
    \item $s(\theta)=1$ and $\frac{\posterior(\state)}{\prior(\state)}\geq \frac{\posterior(\state^*)}{\prior(\state^*)}$ for any $\state\in\statespace_1$;
    \item $s(\theta)=0$ and $\frac{\posterior(\state)}{\prior(\state)}\leq \frac{\posterior(\state^*)}{\prior(\state^*)}$ for any $\state\not\in\statespace_1\cup \{\state^*\}$;
    \item $\expect{\theta \sim \prior}{s(\theta)} = c$.
\end{enumerate}
\end{lemma}

\cref{lem:frac knapsack} indicates that the optimal choice of the agent is essentially betting on a subset of the states (potentially a randomized subset since $s(\theta^*)\in[0,1]$)
that maximizes the probability of this subset given posterior belief $\posterior$ 
subject to the constraint that the prior places probability exactly $c$ on that subset. 
The proof of \cref{thm:betting_full} is provided in \cref{apx:report full}. 

% For each fixed bet of scores, the utility function is a hyperplane. Since the agent chooses the best bet at each posterior, the utility function of a betting scoring rule is the maximum of hyperplanes.  Thus, the utility function of a betting scoring rule is convex, which means the betting scoring rule is proper. Since scores $s(\theta) \in [0,1]$, the betting scoring rule is bounded by 1. 

\begin{theorem}\label{thm:betting_full}
The optimal scoring rule for eliciting the full distribution is the revelation version of the betting mechanism. 
\end{theorem}

The optimal choice of parameter $c$ in the betting mechanism depends on the prior and the distribution over posterior beliefs. 
In \cref{apx:report full}, we show that the objective value of the betting mechanism is concave in the choice of parameter $c$ and hence a simple choice of setting $c=\frac{1}{2}$ is approximately optimal
regardless of the distribution over posteriors. 

\paragraph{linking incentives}
Note that both separate scoring rules and \mos scoring rules for eliciting the mean have straightforward interpretations in the model of eliciting the full distribution. 
Essentially, by treating each state as a separate dimension, eliciting the full distribution is equivalent to eliciting the probability of each state. 
Moreover, each dimension can be viewed as a Bernoulli distribution where eliciting the mean is equivalent to eliciting the full distribution. 
Consequently, our results in \cref{sec:multi} can be naturally extended to show that, in order to incentivize effort, it is crucial to link incentives across different states when eliciting the full distribution.

\subsection{Comparison of Eliciting Mean and Full Distribution}
In many application of interests, the principal elicits distributional information from the agent for making better decisions, 
and the optimal choice of action often only depends on the posterior mean. 
In this section, we show that even when information beyond the posterior mean is not useful for subsequent optimizations, 
the principal may still wish to elicit information beyond the posterior mean in order to better incentivize the agent to exert costly effort. 
Specifically, in this section, we measure the multiplicative gap between the optimal proper
scoring rule for eliciting the full distribution and the optimal
proper scoring rule for eliciting the mean, 
and show that the gap can be unbounded, even when the
size of the state space is a constant. 
The proof of \Cref{thm:gap between full dist and mean} is deferred to \Cref{sec:appendix-proof-gap between full dist and mean}.

\begin{theorem}\label{thm:gap between full dist and mean}
For any $\epsilon \in (0, \frac{1}{2}]$, there exists a finite state space $\statespace\subseteq \reals$, 
and a distribution~$\distoverposterior$ over the posteriors 
such that the objective value of optimal scoring rule for eliciting full
distribution is at least $\sfrac{1}{4}$, while the objective value
of optimal scoring rule for eliciting mean is at most~$\epsilon$.
\end{theorem}

\cref{thm:gap between full dist and mean} implies that there exists settings where if the principal restricts to only eliciting the mean, the agent will not exert costly effort and 
only report the prior mean to the principal. 
However, if the principal elicits the full distribution, the agent can be incentivized to exert costly effort and report the updated posterior belief to the principal, 
which is beneficial for the principal's subsequent optimizations. 
Therefore, in general, the principal may face a tradeoff between minimizing the communication cost 
and incentivizing effort for high quality information. 
It is an interesting open question to understand the optimal tradeoff of the principal in various applications of information elicitation.

%\input{Paper/expected-score-bounded}
% \input{Paper/conclusion}

% \newpage
% \section{Acknowledgements}
% \subsection{Funding}
% This work was supported by NSF CCF-1733860. 
% Liren Shan was supported by NSF CCF-1955351 and CCF-1934931. 
% Yingkai Li also thanks NSF SES-1947021 for financial support. 

% \subsection{Conflict of Interest}
% Researchers from Northwestern University and Yale University.

\bibliography{ref.bib}
\appendix
\section{Detailed Discussion of Related Work.}
\label{app:related-work}

In this section we give a detailed discusson of some of the most
related works.

The early work of \citet*{Osb-89} is close to ours in that it assumes
that the agent has a prior and, with a continuous level effort, can
receive a signal from which the prior is updated to a posterior.  The
principal then aims at optimizing a quadratic loss function while
incentivizing the agent to both put in effort and truthfully report
the posterior.  \citet*{Osb-89} imposes additional constraints on the
scoring rule such that the restricted optimal scoring rule is
quadradic.  In our setting of binary effort, we impose no constraint
on the scoring rule except the ex post boundedness, and we find that
the optimal scoring rule for incentivizing effort is V-shape instead
of quadratic.

\citet*{Z-11} considers the optimization of scoring rules
in the binary state setting, and he shows that among all scoring rules
that induces a certain level of effort, the V-shaped scoring rule is
the one that minimizes the expected transfer from the principal to the
agent.  This objective is qualitatively different from ours, where we
consider the objective of maximizing the agent's expected surplus for
exerting effort, subject to the ex post boundedness constraint.  In
addition, the model in \citet*{Z-11} is restricted in the following
two ways: 1) it only considers single dimensional (i.e., binary
states) optimization problem; 2) agent's cost of effort is known to
the principal.  In our paper, we show that the V-shaped scoring rule
is optimal in the single dimensional problem even when the agent has
private cost of effort, and more importantly, in the multi-dimensional
problem, the V-shaped scoring rule is approximately optimal for
eliciting effort.

Contemporaneously with and independently from our
work, \citet*{NNW-20} consider the optimization of scoring rules for a
binary state setting with uniform prior.  The forecaster has access to
costly samples and solves the optimal stopping problem given the cost
and the scoring rule.  They show that all scoring rules can be ranked
by an incentivization index such that when the cost of the
forecaster's samples converges to zero, the scoring rule with higher
incentivization index induces lower prediction error given that the
forecaster optimizes his expected reward net the cost.  The authors
characterize the scoring rule that maximizes the incentivization
index.  The main difference between their paper and ours is: in their
model, different scoring rules only lead to prediction error with
lower order terms that vanishes to zero, and under equilibrium the
forecaster acquires almost perfect information about the state.  In
contrast, in our model, scoring rule plays a crucial rule for
incentivizing effort, and both the additive gap and the multiplicative
gap between the optimal scoring rule and heuristic scoring rules
(e.g., quadratic scoring rules) for providing incentives can be large.

\citet*{FW-17}
considers the same optimization goal of maximizing incentive in the
different single-task peer prediction setting. In the peer prediction
model, the designer does not have access to a sample of the ground
truth and must cross reference the reports from different agents to
elicit the truthful report. Thus, the truthful peer prediction
mechanism is unique up to positive affine transformations. Their
optimization program reduces to the optimization of the parameters for
affine transformations, which is significantly different from the
optimization of scoring rules.

\section{Missing Proofs in Preliminaries}
\label{app:continuity}

\subsection{Proofs in Section~\ref{subsec:formal_program}}
\label{apx:prelim}

\begin{proof}[Proof of Lemma~\ref{prop:canonical-proper}]
Canonical scoring rules have the following simple interpretation.  By
making a report $\report$, the agent selects the supporting hyperplane
of $\util$ at $\report$ on which to evaluate the state.  This
supporting hyperplane has gradient $\sg(\report)$ and contains point
$(\report,\util(\report))$.  The agent's utility is equal to the value
of the realized state $\state$ on this hyperplane (plus constant
$\kappa(\state)$ which is independent of the agent's report).  With
utility given by a random point on a hyperplane, the expected utility
is equal to its mean on the hyperplane.  When the agent's true
posterior belief is that the state has mean $\report$, the agent's
expected utility is $\util(\report)$ (plus a constant equal to the
expected value of $\kappa(\cdot)$ under the agent's posterior belief;
summarized below as \Cref{lem:canonical-expectation}).  Misreporting
$\report'$ with belief $\report$ gives a utility equal to the value of
$\report$ on the supporting hyperplane with gradient $\sg(\report')$
at $\report'$.  By convexity of $\util$, a report of $\report$ gives
the higher utility of $\util(\report)$.
\end{proof}

\begin{proof}[Proof of \cref{lem:canonical-bound}]
\Cref{eq:canonical-expected-utility} can be derived through simple algebraic calculation, and hence is omitted here. 

Moreover, similar to the proof of \Cref{prop:canonical-proper}, canonical
scoring rules (\Cref{def:canonical}) can be interpreted via
supporting hyperplanes of the utility function.  The first term on
the left-hand side of \eqref{eq:canonical-bound} upper bounds the
utility that an agent can obtain at state $\state$, specifically, it
is the utility from reporting state $\state$.  The remainder of the
left-hand side subtracts the utility that the agent obtains from
report $\report$ in state $\state$, i.e., it evaluates, at state
$\state$, the supporting hyperplane of $\util$ at report $\report$.
Thus, the boundedness constraint of canonical scoring rule $\score$ requires the difference between the
utility function and the value of any supporting hyperplane of the
utility function to be bounded at all states $\state
\in \statespace$. 
\Cref{f:single}(a) illustrates this bound.

% The subgradient in $\{\sg(\report) : \report \in\reportspace\}$ that
% maximizes the right-hand side of the inequality identifies the range
% of ex post score of the agent for this scoring rule. 
Given any utility $\util$ and subgradient $\sg$ that satisfies the inequality, 
To enforce that the score is within $[0,\scorebound]$, 
for any state $\state\in\statespace$,
it is sufficient to select $\kappa(\state)$ equals the negative of the minimum score for any report $\report\in\reportspace$
so that the score is 0 for the report with the worst score at
state~$\state$. 
Since the difference in score is at most $\scorebound$ at state $\state$, the maximum score is at most $\scorebound$. 
\end{proof}

\subsection{Necessity of Canonical Scoring Rules}
\label{apx:necessity_canonical}
Now we provides a partial converse to
\Cref{prop:canonical-proper} and shows that the restriction to
canonical scoring rules is without loss, i.e., \Cref{eq:mainprogram}
and \Cref{eq:program} are equivalent.  The converse will require a
weak technical restriction on the set of scoring rules
considered.\footnote{The literature on scoring rules for eliciting the mean, 
to the best of our knowledge, obtains converses to
\Cref{prop:canonical-proper} only with restrictions.  For example,
\citet{lam-11} assumes the scoring rules are continuously
differentiable in the agent's report.  The restriction we employ is
weaker than differentiability.}  With this restriction,
\citet{AF-12} provide a converse to \Cref{prop:canonical-proper} for
reports in the relative interior of the report space. 
We generalize their observation to the boundary of the report space when the scoring rule is bounded. 
Formally, we have the following result establishing 
that \Cref{eq:mainprogram}
and \Cref{eq:program} are equivalent.
\begin{definition}[\citealp{AF-12}]
A scoring rule $\score$ is $\mu$-differentiable if all directional
derivatives of $\expect{\state\sim
  \posterior}{\score(\mean{\posterior},\state)}$ exists for all
posteriors $\posterior$ with mean $\mean{\posterior}$ in the relative
interior of~$\reportspace$.
\end{definition}

\begin{theorem}\label{thm:equal program}
  For optimization of the incentive for exerting a binary effort via a
  bounded and $\mu$-differentiable scoring rule for the mean, it is
  without loss to consider canonical scoring rules, i.e.,
  \Cref{eq:mainprogram} and \Cref{eq:program} are equivalent.
\end{theorem}

In the subsequent
discussion, the boundary of the report space is denoted by
$\boundaryReport$ and the interior of the report space by
$\interiorReport = \reportspace \setminus \boundaryReport$.

\begin{lemma}[\citealp{AF-12}]\label{thm:multid_character}
  Any proper and $\mu$-differentiable scoring rule for eliciting the
  mean $\score$ coincides with a canonical scoring rule (defined by
  $\util$, $\sg$, and $\kappa$) at reports in the relative interior of the
  report space, i.e., it satisfies equation~\eqref{eq:scoring-rule-construction} for all $\report
  \in \interiorReport$.
\end{lemma}  

The main new results need to show that canonical scoring rules are
without loss for \Cref{eq:mainprogram} are extensions of
\Cref{thm:multid_character} to the boundary of the report space
$\boundaryReport$.  The form of scoring rules considered enters the
program in two places: the objective and the boundedness constraint.
The two lemmas below show that canonical scoring rules are without
loss in these two places in the program. 
% Both lemmas are proved in
% \Cref{app:continuity}.

\begin{lemma}\label{thm:interim utility on the boundary}
  Any $\mu$-differentiable, bounded, and proper scoring rule $\score$ for
  eliciting the mean is equal in expectation of truthful
  reports to a canonical scoring rule (defined by $\util$, $\sg$, and
  $\kappa$), i.e., it satisfies
  equation~\eqref{eq:canonical-expected-utility}.
\end{lemma}

\begin{lemma}\label{thm:bounded_simple}
  For any $\mu$-differentiable and proper scoring rule $\score$ for eliciting
  the mean that induces utility function $\util$ (via
  \Cref{thm:interim utility on the boundary}) and satisfies score
  bounded in $[0,\scorebound]$, there is a canonical scoring rule
  defined by $\util$ (and some $\sg$ and $\kappa$) that satisfies the
  same score bound, i.e., it satisfies
  equation~\eqref{eq:canonical-bound}.
\end{lemma}

\sloppy
Note that \Cref{thm:interim utility on the boundary} implies that the
utility function $\util$ corresponding to any $\mu$-differentiable
scoring rule $\score$ can be identified (via the equivalent cannonical
scoring rule); thus, the assumption of \Cref{thm:bounded_simple} is
well defined.  \Cref{thm:interim utility on the boundary} and
\Cref{thm:bounded_simple} combine to imply that \Cref{eq:mainprogram}
and \Cref{eq:program} are equivalent
and \cref{thm:equal program} holds.

Next, we will formally prove \Cref{thm:interim utility on the boundary} and \ref{thm:bounded_simple}. 
First we show that when the scoring rule is bounded, 
the corresponding functions $\util(\report), \sg(\report), \kappa(\state)$ in the characterization of \Cref{thm:multid_character}
are bounded in the interior as well.

\begin{lemma}\label{lem:bounded functions}
For any bounded scoring rule $\score$,
there exist convex function 
$\util: \reportspace \to \reals$ and function $\kappa : \statespace \to\reals$ such that
for any report 
$\report \in \interiorReport$ 
and any state $\state \in \statespace$,
$$
    \score(\report, \state)= \util(\report) +\sg(\report) \cdot (\state-\report)+\kappa(\state)
$$
where $\sg(\report) \in \partial \util(\report)$ is a subgradient of $\util$, 
% if scoring rule $\score$ is bounded,
and functions $\util(\report), \sg(\report), \kappa(\state)$
are bounded for any report $\report \in \interiorReport$ 
and any state $\state \in \statespace$. 
\end{lemma}
\begin{proof}
\sloppy
Since scoring rule $\score$ is bounded, 
let 
$\bar{\scorebound}_{\state} = \sup_{\report\in \interiorReport} \score(\report, \state)$
and 
$\underline{\scorebound}_{\state} = \inf_{\report\in \interiorReport} \score(\report, \state)$. 
Let $\hat{\report} \in \interiorReport$ be a report in the interior such that 
% $\util(\hat{\report})$ is finite.
both $\util(\hat{\report})$ and $\sg(\hat{\report})$ are finite.
Note that for any state $\state \in \statespace$, 
state $\state$ locate on the boundary of the report space, 
i.e., $\state \in \boundaryReport$,
and the report space is a linear combination of the state space.

For any report $\report \in \interiorReport$, 
by the convexity of function $\util$, we have 
\begin{align*}
\util(\report) 
\geq \util(\hat{\report}) - \sg(\hat{\report}) \cdot (\report-\hat{\report})
\end{align*}
% \begin{align*}
% &\util(\report) 
% \geq \util(\hat{\report}) - \abs{\util(\state) - \util(\hat{\report})} \cdot \frac{\norm{\report-\hat{\report}}}{\norm{\state-\hat{\report}}}
% \geq \util(\hat{\report}) - \abs{\util(\state) - \util(\hat{\report})} 
% \end{align*}
and hence $\util(\report)$ is bounded below. 

Next we show that $\util(\report)$ is bounded above for any report $\report \in \interiorReport$. 
We first show that fixing any state $\state$,
any report $\report$ which is a linear combination of $\state$ and $\hat{\report}$
has bounded utility $\util(\report)$. 
% Suppose by contradiction there exists report $\report$
% such that $\util(\report) > $. 
% Then there exist report $\report'$ which is a linear combination of $\report$ and $\hat{\report}$ 
% such that $\sg(\report') \geq $
If $\util(\report) \leq \util(\hat{\report})$, 
then naturally $\util(\report)$ is bounded above. 
Otherwise, note that
\begin{align*}
&\bar{\scorebound}_{\state} - \underline{\scorebound}_{\state}
\geq \score(\report, \state) - \score(\hat{\report}, \state)
= \util(\report) + \sg(\report) \cdot (\state-\report)
- \util(\hat{\report}) - \sg(\hat{\report}) \cdot (\state-\hat{\report})\\
&\geq (\util(\report) - \util(\hat{\report})) \cdot \frac{\norm{\state-\hat{\report}}}{\norm{\hat{\report}-\report}} + \util(\hat{\report})
- \util(\hat{\report}) - \sg(\hat{\report}) \cdot (\state-\hat{\report})
\geq \util(\report) - \util(\hat{\report})  - \sg(\hat{\report}) \cdot (\state-\hat{\report}),
\end{align*}
where the first inequality holds because the scoring rule is bounded. 
The second inequality holds because the convex function $\util$
projected on line $(\state, \hat{\report})$ is still a convex function. 
The last inequality holds because report $\report$ lies in between $\state$ and $\hat{\report}$. 
Therefore, we have that $\util(\report)$ is bounded above
for report $\report$ lies in between $\state$ and $\hat{\report}$. 
For any state $\state \in \statespace$, 
let $\hat{\util}(\state) = \lim_{k\to\infty} \util(\report^k)$
where $\{\report^k\}_{k=1}^\infty$ is a sequence of report on line $(\state, \hat{\report})$ that converges to $\state$. 
Since $\util(\report^k)$ are bounded for any $\report^k$, 
we have that $\hat{\util}(\state)$ is bounded as well. 
Since the report space is a subset of the convex hull of the state space, 
we have that for any report $\report \in \interiorReport$,
$\util(\report)$ is upper bounded by the convex combination 
of $\hat{\util}(\state)$, which is also bounded by above. 

For any state $\state \in \statespace$, 
we have 
\begin{align*}
\score(\hat{\report}, \state) = \util(\hat{\report}) + \sg(\hat{\report}) \cdot (\state-\hat{\report})+\kappa(\state),
\end{align*}
which implies $\kappa(\state)$ is bounded since all other terms are bounded. 

Finally, 
for any report $\report \in \interiorReport$ 
and any state $\state \in \statespace$, 
\begin{align*}
\score(\report, \state) = \util(\report) + \sg(\report) \cdot (\state-\report)+\kappa(\state),
\end{align*}
which implies $\sg(\report) \cdot (\state-\report)$ is bounded. 
Since the boundedness holds for all directions, 
the subgradient $\sg(\report)$ must also be bounded.
\end{proof}

% we have 
% \begin{align*}
% \score(\report, \report) - \score(\report, \state)
% \end{align*}

% First suppose by contradiction function $\util$ is unbounded, 
% that is, 
% for any constant $c > 0$, there exist $\report \in \interiorReport$ such that 
% either $\util(\report) > c$
% or $\util(\report) < -c$. 

% \begin{theorem}\label{thm:interim utility on the boundary}
% Given any bounded $\mu$-differentiable proper scoring rule $\score$,
% continuous and convex function 
% $\util: \reportspace \to \reals$ and function $\kappa : \statespace \to\reals$ such that
% for any report 
% % $\report$ in the interior of the report space, i.e., 
% $\report \in \interiorReport$ 
% and any state $\state \in \statespace$,
% $$
%     \score(\report, \state)= \util(\report) +\sg(\report) \cdot (\state-\report)+\kappa(\state)
%     % , \qquad \forall \report \in \interiorReport, \state \in \statespace,
% $$
% where $\sg(\report) \in \partial \util(\report)$ is a subgradient of $\util$, 
% % if scoring rule $\score$ is bounded, 
% then
% \begin{align*}
% \expect{\state \sim \posterior}{\score(\mean{\posterior}, \state)}
% = \util(\mean{\posterior}) + \expect{\state \sim \posterior}{\kappa(\state)}
% \end{align*}
% holds for any posterior $\posterior$ such that $\mean{\posterior} \in \reportspace$.
% \end{theorem}

\begin{lemma}\label{lem:convergence}
Given any state space $\statespace$ 
and report space $\reportspace$ with non-empty interior, 
for any distribution $\posterior \in \Delta(\statespace)$ with mean $\mean{\posterior}$,
there exists a sequence of posteriors $\{\posterior^k\}$
such that for any bounded function $\phi(\state)$ in space $\statespace$, 
we have $\{\expect{\state \sim \posterior^k}{\phi(\state)}\}$
converges to $\expect{\state \sim \posterior}{\phi(\state)}$.
\end{lemma}
\begin{proof}
Since space $\reportspace$ has a non-empty interior, 
let $\widetilde{\posterior}$ be a distribution with mean $\mean{\widetilde{\posterior}}$ in the interior of $\reportspace$. 
Let the sequence of posteriors 
$\posterior^k = (1-\sfrac{1}{k})\cdot\posterior + \sfrac{1}{k} \cdot \widetilde{\posterior}$. 
For any bounded function $\phi(\state)$ in space $\statespace$, 
we have 
\begin{equation*}
\lim_{k\to\infty} \expect{\state \sim \posterior^k}{\phi(\state)} 
= \lim_{k\to\infty} [(1-\sfrac{1}{k})\cdot\expect{\state \sim \posterior}{\phi(\state)} 
+ \sfrac{1}{k} \cdot \expect{\state \sim \widetilde{\posterior}}{\phi(\state)}]
\to \expect{\state \sim \posterior}{\phi(\state)}. \qedhere
\end{equation*}
\end{proof}

% \begin{numberedlemma}{\ref{thm:interim utility on the boundary}}
%   Any $\mu$-differentiable, bounded, and proper scoring rule $\score$
%   for eliciting the mean is equal in expectation of truthful reports
%   to a canonical scoring rule (defined by $\util$, $\sg$, and
%   $\kappa$), i.e., it satisfies
%   equation~\eqref{eq:canonical-expected-utility}.
% % A scoring rule $\score$ is $\mu$-differentiable, bounded
% % and proper for eliciting the mean 
% % if and only if there exists bounded, continuous and convex function 
% % $\util: \reportspace \to \reals$ with bounded subgradients 
% % and bounded function $\kappa : \statespace \to\reals$ such that
% % \begin{align*}
% % \expect{\state \sim \posterior}{\score(\mean{\posterior}, \state)}
% % = \util(\mean{\posterior}) + \expect{\state \sim \posterior}{\kappa(\state)}
% % \end{align*}
% % for any posterior $\posterior$ such that $\mean{\posterior} \in \reportspace$.
% \end{numberedlemma}
\begin{proof}[Proof of \cref{thm:interim utility on the boundary}]
% [Proof of \Cref{thm:interim utility on the boundary}]
By \Cref{thm:multid_character}, 
for $\mu$-differentiable proper scoring rule $\score$,
there exists convex function 
$\util: \reportspace \to \reals$ and function $\kappa : \statespace \to\reals$ such that
for any report 
% $\report$ in the interior of the report space, i.e., 
$\report \in \interiorReport$ 
and any state $\state \in \statespace$,
we have 
$$
    \score(\report, \state)= \util(\report) +\sg(\report) \cdot (\state-\report)+\kappa(\state)
    % , \qquad \forall \report \in \interiorReport, \state \in \statespace,
$$
where $\sg(\report) \in \subgradient \util(\report)$ is a subgradient of $\util$.
% there exists a convex function $\util$ such that the characterization holds in the interior. 
By \Cref{lem:bounded functions}, since the scoring rule is bounded, 
function $\util$ is convex and bounded 
and hence continuous in the interior.
Thus, we can well define
the value of $\util$ on the boundary as its limit from the interior, 
i.e., set $\util(\report) = \lim_{k \to \infty} \util(\report^k)$
for any $\report$ on the boundary of the report space $\reportspace$
and $\{\report^k\}_{k=1}^\infty$ as a sequence of interior reports converging to $\report$.
Thus we can replace the convex function $\util$ 
% in the characterization of \Cref{thm:multid_character} 
with continuous and convex function $\util$ 
for bounded scoring rules and the characterization still holds in the interior. 

For any bounded proper scoring rule, we have that 
$\util(\report)$ is bounded for any report 
$\report \in \interiorReport$ 
and $\kappa(\state)$ is bounded for any state $\state \in \statespace$.
Given any posterior $\posterior$ such that $\mean{\posterior} \in \boundaryReport$, 
let $\{\posterior^k\}$ be the sequence of posteriors constructed in \Cref{lem:convergence}. 
\begin{enumerate}
\item The identity function $\phi(\state) = \state$ is bounded. 
Therefore, the mean of the posteriors converges, 
i.e., 
$\lim_{k\to\infty}\mean{\posterior^k} = \mean{\posterior}$. And all means $\{\mean{\posterior^k}\}$ are in the interior of $\reportspace$.

\sloppy
\item Function $\kappa(\state)$ is bounded. 
Therefore, the expected value for function $\kappa$ converges. 
That is, 
$\lim_{k\to\infty} \expect{\state \sim \posterior^k}{\kappa(\state)}
= \expect{\state \sim \posterior}{\kappa(\state)}$.

\item The ex post score $\score(\report, \state)$ is bounded. 
Therefore, the expected score for reporting $\mean{\posterior}$ converges, 
i.e., 
$\lim_{k\to\infty} \expect{\state \sim \posterior^k}{\score(\mean{\posterior}, \state)}
= \expect{\state \sim \posterior}{\score(\mean{\posterior}, \state)}$. 
\end{enumerate}

Moreover, considering the sequence of expected score for reporting 
$\mean{\posterior^k}$ with distribution $\posterior$, 
we have 
% we have the expected score $\{\expect{\state \sim \posterior_i}{\score(\mean{\posterior_i}, \state)}\}$
% converges to $\expect{\state \sim \posterior}{\score(\mean{\posterior_i}, \state)}$. 
\begin{align*}
&\lim_{k\to\infty} \expect{\state \sim \posterior}{\score(\mean{\posterior^k}, \state)}
= \lim_{k\to\infty} [\util(\mean{\posterior^k}) 
+ \expect{\state \sim \posterior}{\sg(\mean{\posterior^k}) \cdot (\state-\mean{\posterior^k})}
+ \expect{\state \sim \posterior}{\kappa(\state)}] \\
&= \lim_{k\to\infty} [\util(\mean{\posterior^k}) 
+ \expect{\state \sim \posterior^k}{\kappa(\state)}]
= \lim_{k\to\infty} [\expect{\state \sim \posterior^k}{\score(\mean{\posterior^k}, \state)}
\end{align*}
where the second equality holds because 
$\lim_{k\to\infty} \expect{\state \sim \posterior^k}{\kappa(\state)}
= \expect{\state \sim \posterior}{\kappa(\state)}$
and 
$\lim_{k\to\infty} \mean{\posterior^k} = \mean{\posterior}$.
% \begin{align*}
% \lim_{i\to\infty} [\expect{\state \sim \posterior_i}{\score(\mean{\posterior_i}, \state)}
% - \expect{\state \sim \posterior}{\score(\mean{\posterior_i}, \state)}
% = \lim_{i\to\infty} [\util(\mean{\posterior_i}) + \expect{\state \sim \posterior_i}{\kappa(\state)}]
% = \lim_{i\to\infty} [\util(\mean{\posterior_i}) + \expect{\state \sim \posterior_i}{\kappa(\state)}]
% \end{align*}
Combining the equalities, we have 
\begin{align*}
&\expect{\state \sim \posterior}{\score(\mean{\posterior}, \state)}
= \lim_{k\to\infty} \expect{\state \sim \posterior^k}{\score(\mean{\posterior}, \state)} 
\leq \lim_{k\to\infty} \expect{\state \sim \posterior^k}{\score(\mean{\posterior^k}, \state)}\\
&= \lim_{k\to\infty} \expect{\state \sim \posterior^k}{\score(\mean{\posterior^k}, \state)}
= \lim_{k\to\infty} \expect{\state \sim \posterior}{\score(\mean{\posterior^k}, \state)}
\leq \expect{\state \sim \posterior}{\score(\mean{\posterior}, \state)}
\end{align*}
where the inequalities holds by the properness of the scoring rule. 
Therefore, 
all inequalities must be equalities, 
and hence 
\begin{align*}
&\expect{\state \sim \posterior}{\score(\mean{\posterior}, \state)}
= \lim_{k\to\infty} \expect{\state \sim \posterior^k}{\score(\mean{\posterior^k}, \state)}\\
&= \lim_{k\to\infty} \expect{\state \sim \posterior^k}{\util(\mean{\posterior^k}) + \kappa(\state)}
= \util(\mean{\posterior}) + \expect{\state \sim \posterior}{\kappa(\state)}. 
\end{align*}
where the last equality hold since function $\util$ is continuous. 

Finally, given any bounded, continuous and convex function 
$\util$ with bounded subgradients
and any bounded function $\kappa$, 
the corresponding canonical scoring rule is proper, bounded, 
and the expected score coincides.  
% let $$
%     \score(\report, \state)= \util(\report) +\sg(\report) \cdot (\state-\report)+\kappa(\state)
%     % , \qquad \forall \report \in \interiorReport, \state \in \statespace,
% $$
% where $\sg(\report) \in \subgradient \util(\report)$ is a subgradient of $\util$. 
% It is easy to verify that 
% the resulting scoring rule $\score$ is $\mu$-differentiable, proper, bounded, and 
% \begin{equation*}
% \expect{\state \sim \posterior}{\score(\mean{\posterior}, \state)}
% = \util(\mean{\posterior}) + \expect{\state \sim \posterior}{\kappa(\state)} \qedhere
% \end{equation*}
\end{proof}

% \begin{numberedlemma}{\ref{thm:bounded_simple}}
%   For any $\mu$-differentiable and proper scoring rule $\score$
%   inducing utility function $\util$ (via \Cref{thm:interim utility on
%     the boundary}) and satisfying score bounded in $[0,\scorebound]$,
%   there is a canonical scoring rule defined by $\util$ (and some $\sg$
%   and $\kappa$) that satisfies the same score bound, i.e., it satisfies
%   equation~\eqref{eq:canonical-bound}.
% % Function $\util$ induces a $\mu$-differentiable proper scoring rule
% % that is bounded by $\scorebound$ in space $\reportspace\times\statespace$
% % if and only if there exists a set of subgradients $\xi(\report) \in \subgradient \util(\report)$
% % % for the utility function $\util$ of its representative 
% % such that
% % $$
% % \util(\state)-\util(\report) - \sg(\report)(\state - \report)\leq \scorebound 
% % % \qquad \forall \report\in \reportspace, \state\in \statespace.
% % $$
% % for any report $\report\in \reportspace$ and state $\state\in \statespace$.
% \end{numberedlemma}
\begin{proof}[Proof of \cref{thm:bounded_simple}]
If a proper scoring rule 
$\score$ is induced by function $\util$ 
and bounded by $\scorebound$ in space $\statespace$, 
by \Cref{thm:multid_character},
there exists function $\kappa : \statespace \to\reals$ such that
for any report 
% $\report$ in the interior of the report space, i.e., 
$\report \in \interiorReport$ 
and any state $\state \in \statespace$,
$$
    \score(\report, \state)= \util(\report) +\sg(\report) \cdot (\state-\report)+\kappa(\state)
    % , \qquad \forall \report \in \interiorReport, \state \in \statespace,
$$
where $\sg(\report) \in \subgradient \util(\report)$ is a subgradient of $\util$. 
Moreover, 
% in this class which is bounded by $\scorebound$ in space $\statespace$,
% which implies 
the score $\score(\report,\state) \in [0,\scorebound]$ for any report and state $\report \in \reportspace,\state \in \statespace$. 
% Let $\sg(\report)$ be the set of subgradients chosen by this scoring rule $\score$ for any report $\report \in \interiorReport$. 
Thus, it holds that for any report and state $\report \in \interiorReport,\state \in \statespace$
$$
\score(\state,\state) - \score(\report,\state) = \util(\state) - \util(\report) - \sg(\report)(\state - \report) \leq \scorebound.
$$
For any report $\reportspace \in \boundaryReport$, 
there exists a sequence of reports $\report_i$ 
such that $\{\report_k\}$ converges to $\report$
and $\sg(\report) = \lim_{k\to\infty} \sg(\report_k)$ is a subgradient at report $\report$. 
Thus, it holds that for any report $\report \in \boundaryReport$
and state $\state \in \statespace$,
$$
\score(\state,\state) - \score(\report,\state) = \util(\state) - \util(\report) - \lim_{k\to\infty} \sg(\report_k)(\state - \report) \leq \scorebound.
$$
Therefore, the canonical scoring rule defined by $\util$ with the same function $\kappa$ is proper and bounded in $[0, B]$. 
\end{proof}

\section{Single-Dimensional Scoring Rules}
\label{apx:single}

\subsection{Proof of Lemma~\ref{lem:bounded for single}}
\label{sub:bound_single}
\begin{proof}[Proof of \cref{lem:bounded for single}]
By \Cref{lem:canonical-bound}, it is sufficient to only consider convex
function~$\util$ such that there exists a set of subgradients
$\sg(\report)$ satisfying constraints that for any $\report, \state\in [0,1]$, i.e., 
$$
\util(\state)-\util(\report)
-\sg(\report) \cdot (\state-\report) \leq 1.
$$
By convexity of utility $\util$ and the monotonicity of
subgradients $\sg$ on report space $\reportspace = [0,1]$, it is straightforward to observe that the left-hand side of the boundedness constraint is maximized at $\state \in \{0,1\}$ with $\report = 1-\state$ (see \Cref{f:single}(a)).  
\end{proof}

\subsection{Proof of Theorem~\ref{thm:1d_opt}}
\label{apx:thm:1d_opt}

\begin{proof}[Proof of Theorem~\ref{thm:1d_opt}]
Consider any feasible solution $\util(\report)$ of
\Cref{eq:1d_program_u}. We construct a V-shaped utility function
$\tilde{\util}(\report)$ as
\begin{align*}
    \tilde{\util}(\report) = 
    \begin{cases}
    -\frac{\util(0)}{\priorMean}(\report-\priorMean)& \text{for $\report \leq \priorMean$,}\\
    \frac{\util(1)}{1-\priorMean}(\report-\priorMean)& \text{for $\report \geq \priorMean$.}
    \end{cases}
\end{align*}
The construction of $\tilde{\util}$ is illustrated in \Cref{fig:opt
  single}.  It is easy to see that $\tilde{\util}$ is convex,
$\tilde{\util}(\priorMean) = 0$ and $\tilde{\util}(\report) \geq
\util(\report)$ for any $\report \in [0,1]$.  Therefore, the objective
value for function $\tilde{\util}$ is higher than objective value for
function $\util$.  Moreover, we have $\tilde{\util}(0) = u(0)$,
$\tilde{\util}(1) = u(1)$, $\tilde{\util}'(0) \geq \sg(0)$ and
$\tilde{\util}'(1) \leq \sg(1)$, which implies $\tilde{\util}$ is also
a feasible solution to \Cref{eq:1d_program_u}.  Thus, an optimal
solution is V-shaped.

Next we focus on finding the optimal V-shaped function $\tilde{\util}$
for \Cref{eq:1d_program_u}.  Let $a = -\sfrac{\util(0)}{\priorMean} =
\tilde{\util}'(0)$ and $b= \sfrac{\util(1)}{(1-\priorMean)} =
\tilde{\util}'(1)$.  Since function $\tilde{\util}$ satisfies the
constraints in \Cref{eq:1d_program_u}, we get
\begin{align*}
b(1-\priorMean) = \tilde{\util}(1) \leq 1+ \tilde{\util}(0) + \tilde{\util}'(0) = 1-a\cdot \priorMean+a,\\
b(1-\priorMean) = \tilde{\util}(1) \geq \tilde{\util}'(1) + \tilde{\util}(0) -1 = b-a\cdot \priorMean-1,
\end{align*}
which implies $b \leq a + \sfrac{1}{(1-\priorMean)}$ and $b \leq a +
\sfrac{1}{\priorMean}$. If $b < a + \sfrac{1}{\max\{\priorMean,
  1-\priorMean\}}$, then we can either increase $b$ or decrease $a$ to
get a better feasible V-shaped utility function. Suppose we fix
parameter $a$, the objective value is pointwise maximized for any
report $\report$ when $b = a + \sfrac{1}{\max\{\priorMean,
  1-\priorMean\}}$.

Next we fix the optimal choice for parameter $b$.  Note that the
objective value given any parameter~$a$ is
\begin{align}
  \notag
\int_0^1 \util(\report) 
\dd \marginalReport(\report) 
&= \int_{0}^{\priorMean} a(\report - \priorMean) 
\dd \marginalReport(\report) 
+ \int_{\priorMean}^{1} \left(a+\frac{1}{\OPTdenom}\right)(\report - \priorMean) 
\dd \marginalReport(\report) \\
\label{eq:sd-opt}
&= \frac{1}{\OPTdenom} \int_{\priorMean}^{1} (\report - \priorMean) 
\dd \marginalReport(\report),
% = \expect{\report - \priorMean \given \report \geq \priorMean},
\end{align}
which invariant of parameter $a$.  Therefore, any V-shaped utility
function with parameters satisfying $b=a+\sfrac{1}{\max\{\priorMean,
1-\priorMean\}}$ is optimal and obtains objective value given by \Cref{eq:sd-opt}.
\end{proof}

\subsection{Center Prior Mean}
An important special case for our subsequent analyses is when the mean
of the posteriors is in the center of the report space, i.e.,
$\priorMean = \sfrac 1 2$ for report space $[0,1]$. 
In this case, by \cref{thm:1d_opt},
an optimal utility function $\util$ is V-shaped at $\sfrac 1 2$ with
$\util(0) = \util(1) = \frac{1}{2}$.

\begin{corollary}
  \label{c:symmetric-single-dimensional}
For any distribution $\marginalReport$ over the posterior means with
expectation $\priorMean = \sfrac 1 2$, one of the optimal solution
of \Cref{eq:1d_program_u} is symmetric 
and V-shaped at $\sfrac 1 2$ with $\util(0) =
\util(1) = \sfrac{ 1} {2}$. 
\end{corollary}

% \subsection{Proof of \Cref{cor:upper bound opt}}\label{sec:appendix-proof-upper bound opt}
% \begin{proof}[Proof of \Cref{cor:upper bound opt}]
% In the characterization of the optimal performance of
% \Cref{thm:1d_opt}, i.e., $$\OPT(\marginalReport) = {\expect{\report
%     \sim \marginalReport}{\maxzero{\report - \priorMean}}} /
%      {\OPTdenom},$$ it is easy to see that the numerator is maximized
%      and the denominator is minimized in when the distribution of
%      posterior means $\marginalReport$ is uniform on the extreme
%      points $\{0,1\}$.  For this distribution, the numerator is
%      $\sfrac 1 4$ and the denominator is $\sfrac 1 2$.  Thus,
%      $\OPT(\marginalReport) = \sfrac 1 2$.
% \end{proof}

\section{Multi-dimensional Scoring Rules}
\label{app:section4}

\subsection{Proof of Proposition~\ref{thm:betting mean approx}}
\label{apx:thm:betting mean approx}
\begin{proof}[Proof of \cref{thm:betting mean approx}]
We first show that, 
for any proper scoring rule for eliciting the mean with utility function $\hat{\util}$, there exists a utility function $\util$ which has the same objective value and satisfies (1) $\util(\priorMean)=\frac{1}{2}$; 
and (2) $\util\in [\frac{1}{2}, \frac{3}{2}]$. 
This is equivalent to showing that there exists a utility function~$\util$ with the same objective value and satisfies (1) $\util(\priorMean)=0$; and (2) $\util\in [0, 1]$.

By convexity, there exists a hyperplane $(\sg(\priorMean), b_{\priorMean})$ that lower bounds the convex function $\hat{\util}$ and passes through $\hat{\util}(\priorMean)$, i.e.\ $\sg(\priorMean)\cdot\priorMean+b_{\priorMean}=\hat{\util}(\priorMean)$, 
and $\sg(\priorMean)\cdot\report+b_{\priorMean}\leq\hat{\util}(\report)$ for any report~$\report$. Consider the new function $\util(\report)=\hat{\util}(\report)-(\sg(\priorMean)\cdot\report+b_{\priorMean})$. 
By the linearity of the construction, the new utility function $\util$ has the same objective as the utility function $\hat{u}$. This utility function $\util$ is non-negative and its value at the prior mean $\util(\priorMean)$ is $0$. 

Now we only need to show this constructed $\util$ is bounded by $[0, 1]$. For any state $\state$, consider the subgradient $\sg$ at its symmetric point $\report'=2\priorMean-\state$. Since the tangent hyperplane of the utility function $\util$ at point $r'$ is below the function $\util$, we have the value of this tangent hyperplane at the prior mean $\util(\report')+\sg(\report')\cdot(\priorMean-\report')\leq 0$. Since $\util(\report')\geq 0$, we have $\util(\report')+\sg(\report')\cdot(\report-\report')\leq 0$. By the boundedness constraint $\util(\state)-\util(\report') - \sg(\report')\cdot(\state - \report') \leq 1$, we know $\util(\state)\leq 1$. Since the report space $\reportspace$ is the convex hull of the state space $\statespace$, the convex utility function is bounded by $[0, 1]$ on the report space.

By boundedness constraint, $\util(\state)-\util(\report) - \sg(\report)\cdot(\state - \report) \leq 1$. These implies $\util(\report) + \sg(\report)\cdot(\state - \report) \in [-\frac{1}{2}, \frac{3}{2}]$ for all $\state, \report$. By re-scaling it to satisfy $\util(\report) + \sg(\report)\cdot(\state - \report) \in [0, 1]$ for all $\state, \report$, we obtain a bounded scoring rule that is a $2$-approximation.
\end{proof}

\subsection{Proof of Lemma~\ref{lem:symm_nd_bounded}}\label{sec:appendix-proof-symm_nd_bounded}

\begin{proof}
  The following geometry of the utility function is easy verify.
  First, convexity of report space $\reportspace$ implies convexity of $\util$.
  Second, consider the $\numasm+1$ dimensional space $\reportspace
  \times [-\sfrac 1 2,\sfrac 1 2]$, where the $\numasm+1$st dimension represents the
  utility $\util$.  The utility function defines a truncated convex
  cone with vertex equal to $(\priorMean,0)$ and base at height
  $\sfrac 1 2$ with cross section $\reportspace$.  Consider the point reflection, henceforth, the
  reflected cone, of this convex cone around its vertex
  $(\priorMean,0)$.  By basic properties of cones and their point
  reflections, this reflected cone has the same supporting hyperplanes
  as the original cone.  By the symmetry assumption of $\reportspace$
  around $\priorMean$, the reflected cone is equal to the mirror
  reflection of the original cone with respect to the $\util = 0$
  plane.  Consequently, the base of the reflected cone at $\util =
  -\sfrac 1 2$ has cross section equal to $\reportspace$.

  We now argue that the utility function satisfies the boundeness
  constraint, restated for convenience (with report $\report
  \in \reportspace$ and state $\state \in \statespace$):
  $$
  \util(\state) - \util(\report) - \subgradient \util(\report) \cdot(\state-\report)  \leq 1.
  $$
  By definition of the V-shaped utility, we know that the first term
  is at most $\sfrac 1 2$.  The second and third terms, together, can
  be viewed as subtracting the evaluation, at state $\state$, of the
  supporting hyperplane of $\util$ at $\report$.  The highest point in
  the reflected cone for any $\state \in \reportspace$ is
  $-\util(\state)$ and this point lower bounds the value of $\state$
  in any of the reflected cones supporting hyperplanes (which are the
  same as the original cones supporting hyperplanes).  By definition,
  the reflected cone satisfies $-\util(\state) \geq -\sfrac 1 2$ for
  $\state \in \reportspace$.  We conclude, as desired, that the
  difference between the first term and the second and third terms is
  at most 1.
  \end{proof}
  
% In this section we show that the extended utility function $\extutil$ defined in \Cref{lem:obj ratio} induces bounded scoring rules.  
% First we state the following lemma showing that the set of subgradients that induce proper scoring rules is a closed set.

\subsection{Proof of Theorem~\ref{thm:opt_symm_nd}}\label{sec:appendix-proof-opt_symm_nd}

\begin{proof}
  Consider relaxing the optimization problem on the general space
  solve it independently on lines through the center.  Specifically,
  consider the conditional distribution of $\marginalReport$ on the
  line segment through the center $\priorMean$ and the boundary points
  $\report$ and $2\priorMean - \report$ on $\boundaryReport$.  Center
  symmetry implys symmetry on this line segment.  By
  \Cref{c:symmetric-single-dimensional}, the solution to this
  single-dimensional problem is symmetric V-shaped, i.e., with
  $\util(\report) = \util(2\priorMean-\report) = \sfrac 1 2$ and
  $\util(\priorMean) = 1/2$.

  The solutions on all lines through the center $\priorMean$ coincide
  at $\priorMean$ with $\util(\priorMean) = 0$.  They can be combined,
  and the resulting utility function $\util$ is a symmetric V-shaped
  function (\Cref{def:symmetric v-shaped scoring rule}).
  \Cref{lem:symm_nd_bounded} implies that $\util$ is convex and
  bounded and, thus feasible for the original program.  Since it
  optimizes a relaxation of the original program, it is also optimal
  for the original program.
\end{proof}

\subsection{Proof of Proposition~\ref{lem:v shape is max over separate}}
\label{apx:lem_vshape}
\begin{proof}
First, it is easy to verify that the single dimensional scoring rules $\sdscore_i$ are proper and bounded in $[0, 1]$. 
For each dimension $i$, the utility function for each single dimensional scoring rule $\sdscore_i$ is V-shaped with
\begin{align*}
\sdutil_i(\report_i) &= 
\begin{cases}
-\frac{1}{b_i-a_i} (\report_i-\Pmeani) & \report_i \leq \Pmeani\\
\frac{1}{b_i-a_i} (\report_i-\Pmeani) & \report_i \geq \Pmeani
\end{cases}, \qquad\text{and}&
\sdkappa_i(\state_i) &= \sfrac{1}{2}. 
\end{align*}
By~\Cref{def:max over separate}, the \mos scoring rule $\score$ is 
$\score(\report,\state) = \sdscore_i(\report_i,\state_i)$ where $i \in \argmax_j \sdutil_j(\report_j)$, 
and hence the utility function for \mos scoring rule~$\score$ 
can be computed as
$\util(\report) = \max_{i \in [n]} \sdutil_i(\report_i)$, 
which coincides with the symmetric V-shaped function~$\util$. 
\end{proof}

\subsection{Properties of Choose-and-Report Scoring Rules}\label{sec:appendix-proof-proper and bounded for mos}
\begin{lemma}\label{lem:proper and bounded for mos}
The \car scoring rule $\score$ defined by proper and bounded
single-dimensional scoring rules $(\sdscore_1, \dots, \sdscore_{\numasm})$
is itself proper and bounded.
\end{lemma}
\begin{proof}
Given posterior distribution $\posterior$, let $i$ be the dimension
that maximizes the agent's expected utility under separate scoring
rules $\sdscore_1$, \dots, $\sdscore_{\numasm}$, i.e., $i =
\argmax_{j} \expect{\state_j \sim
  \posterior_j}{\sdscore_j(\mean{\posterior_j}, \state_j)}$, and let
$\report_i = \mean{\posterior_i}$ be the mean of the posterior on
dimension~$i$.  For report~$\report = (i,\report_i)$ and any other
report~$\report' = (i', \report'_i)$, we have
\begin{align*}
\expect{\state \sim \posterior}{\score(\report, \state)} 
&=
\expect{\state_i \sim \posterior_i}{\sdscore_i(\report_i, \state_i)}
\geq \expect{\state_{i'} \sim
  \posterior_{i'}}{\sdscore_{i'}(\mean{\posterior_{i'}}, \state_{i'})}\\
&\geq \expect{\state_{i'} \sim
  \posterior_{i'}}{\sdscore_{i'}(\report'_{i'}, \state_{i'})} =
\expect{\state \sim \posterior}{\score(\report', \state)}.
\end{align*}
The first and last equality hold by the definition of \car proper
scoring rules, and the first inequality holds by the definition of
dimension $i$.  The second inequality holds since each single
dimensional scoring rule is proper.  Thus the \car scoring rule
$\score$ is proper.  Moreover, if each single dimensional proper
scoring rule $\sdscore_i$ is bounded, it is easy to verify that the
\car scoring rule $\score$ is also bounded.
\end{proof}

% \subsection{Proof of \cref{lem}}
% \begin{proof}
% First, it is easy to verify that the single dimensional scoring rules $\sdscore_i$ are proper and bounded in $[0, 1]$. 
% For each dimension $i$, the utility function for each single dimensional scoring rule $\sdscore_i$ is V-shaped with
% \begin{align*}
% \sdutil_i(\report_i) &= 
% \begin{cases}
% -\frac{1}{b_i-a_i} (\report_i-\Pmeani) & \report_i \leq \Pmeani\\
% \frac{1}{b_i-a_i} (\report_i-\Pmeani) & \report_i \geq \Pmeani
% \end{cases} \text{, and}&
% \sdkappa_i(\state_i) &= \sfrac{1}{2}. 
% \end{align*}
% By~\Cref{def:max over separate}, the \mos scoring rule $\score$ is 
% $\score(\report,\state) = \sdscore_i(\report_i,\state_i)$ where $i \in \argmax_j \sdutil_j(\report_j)$, 
% and hence the utility function for \mos scoring rule~$\score$ 
% can be computed as
% $\util(\report) = \max_{i \in [n]} \sdutil_i(\report_i)$, 
% which coincides with the symmetric V-shaped function~$\util$. 
% % $
% % \sdutil_i(\report_i) = -\frac{1}{b_i-a_1} (\report_i-\Pmeani)$ for $\report_i \leq \Pmeani$ 
% % and $\frac{1}{b_i-a_1} (\report_i-\Pmeani)$ for $\report_i \geq \Pmeani$. 
% % Then, we have $\sdscore_i(\report_i,\state_i) = \sdutil_i(\report_i) + \sdutil_i'(\report_i)(\state_i-\report_i) + \frac{1}{2}$. By~\Cref{def:max over separate}, the \mos scoring rule $\score$ is $\score(\report,\state) = \sdscore_i(\report_i,\state_i)$ where $i \in \argmax_j \sdutil_j(\report_j)$. Thus, the utility function of scoring rule $\score$ is $\util(\report) = \max_{i \in [n]} \sdutil_i(\report_i)$, which is optimal.
% \end{proof}

%\subsection{Proof of \Cref{}}

\subsection{Proof of Proposition~\ref{thm:score separately multiplicative}}\label{sec:appendix-proof-score separately multiplicative}

\begin{proof}
  We first argue the upper bound that scoring separately in rectangular report and state spaces guarantees an
  $O(\numasm)$ approximation. 
  By \Cref{thm:approx_max_over_separate}, 
  there exists proper and bounded single-dimensional proper scoring rules $(\score_1, \dots, \score_{\numasm})$
  such that the induced \mos $\score$ is an 8-approximation to the optimal scoring rule. 
  Let $\hat{\score}$ be the separate scoring rule induced by single-dimensional proper scoring rules 
  $(\frac{1}{\numasm}\score_1, \dots, \frac{1}{\numasm}\score_{\numasm})$. 
  It is easy to verify that scoring rule $\hat{\score}$ is bounded, with objective value at least $\frac{1}{\numasm}$ fraction of that for scoring rule $\score$. 
  Thus, separate scoring rule $\hat{\score}$ is an
  $O(\numasm)$ approximation to the optimal scoring rule. 
%   ADD PROOF HERE.

  We now give an example of a symmetric distribution over posteriors
  over the space $\reportspace = \statespace = [0,1]^{\numasm}$ such that the approximation is $\Omega(\numasm)$.
  Consider the i.i.d.\ distribution over posterior means $\marginalReport$ with
  marginal distribution $\marginalReport_i$ dimension $i$ defined by
  \begin{align*}
\report_i =  
\begin{cases}
  1          & \text{w.p.\ $\sfrac 1 {2\numasm}$},\\
  \sfrac 1 2 & \text{w.p.\ $1- \sfrac 1 \numasm$},\\
  0          & \text{w.p.\ $\sfrac 1 {2\numasm}$}.
\end{cases}
\end{align*}\sloppy
  The prior mean for each dimension is $\sfrac{1}{2}$ and by
  \Cref{c:symmetric-single-dimensional}, the optimal scoring rule for
  each dimension $i$ has V-shaped utility function $\sdutili$ with $\sdutili(0) = \sdutili(1) = \sfrac 1 2$
  and $\sdutili(\sfrac 1 2) = 0$.  Thus, the expected objective value for
  the optimal scoring rule of dimension $i$ is $\sfrac 1 2 \,
  \prob{\report_i \sim \marginalReport_i}{\report_i \in \{0,1\}} =
  \sfrac 1 {2\numasm}$.  Any average of optimal separate scoring rules, thus, has
  objective value $\sfrac 1 {2\numasm}$.

  Now consider the \mos scoring rule which has a (multi-dimensional)
  symmetric V-shaped utility function $\util$ and is optimal (see
  \Cref{def:symmetric v-shaped scoring rule} and
  \Cref{thm:opt_symm_nd}).  The objective value is $\expect{\report
    \sim \marginalReport}{\util(\report)}$.  Importantly
  $\util(\report) = 0$ if $\report = (\sfrac 1 2,\ldots,\sfrac 1 2)$
  and, otherwise, $\util(\report) = \sfrac 1 2$.  Thus, \begin{align*}
    \OPT(\marginalReport) &= \sfrac 1 2 \, \prob{\report \sim
      \marginalReport}{\report \neq (\sfrac 1 2,\ldots,\sfrac 1
      2)}\\ &= \sfrac 1 2 \, (1 - (1 - \sfrac 1 {\numasm})^{\numasm}) \geq \sfrac 1 2
    \,(1-\sfrac 1 e).
    \end{align*}
  Thus, the approximation ratio of optimal separate scoring to optimal
  scoring is at least $\sfrac{e\, \numasm}{e-1}$ (and this bound is
  tight in the limit of $\numasm$).
\end{proof}

\subsection{Max-ove-separate vs.\ Separate Scoring Rules}
\label{apx:thm:mos better}
\begin{proposition}\label{thm:mos better}
For any distribution $\marginalReport$ over posterior means
and for any separate scoring rule
with single dimensional scoring rules $(\score_1, \dots, \score_{\numasm})$ such that 
$\score(\report, \state)
= \sum_i\score_i(\report_i, \state_i)$, 
there exists a \mos scoring rule with objective value weakly higher than $\score$.
\end{proposition}

Intuitively, by taking the max instead of taking the weighted average of single-dimensional scoring rules, 
the principal obtains at least the maximum over the marginal objective values from all dimensions instead of their average, 
and hence the objective value weakly improves. 
\begin{proof}[Proof of \cref{thm:mos better}]
Given any separate scoring rule
with corresponding single dimensional scoring rules $(\score_1, \dots, \score_{\numasm})$, 
for any dimension $i$, let 
\begin{align*}
\underline{s}_i = \min_{\report_i,\state_i} \score_i(\report_i,\state_i)
\text{ and }
\bar{s}_i = \min_{\report_i,\state_i} \score_i(\report_i,\state_i).
\end{align*}
Let $\hat{\score}_i(\report_i,\state_i) \triangleq \frac{1}{\bar{s}_i - \underline{s}_i} (\score_i(\report_i,\state_i)- \underline{s}_i)$.
It is easy to verify that $\hat{\score}_i$ is bounded in $[0,1]$.
By the boundedness constraint of $\score$, we have 
$\sum_i \underline{s}_i\geq 0$ and $\sum_i \bar{s}_i\leq 1$,
which implies $\sum_i (\bar{s}_i-\underline{s}_i)\leq 1$.
Therefore, the separate scoring rule 
$\hat{\score}(\report,\state) = \sum_i(\bar{s}_i-\underline{s}_i)\cdot \hat{\score}_i(\report_i,\state_i)$ 
is also bounded in $[0,1]$ and has the same objective value as $\score$. 

For any dimensional $i$, let $\marginalReport_i$ be the marginal distribution over mean on dimension $i$, 
and let $\objfunc(\hat{\score}_i,\marginalReport_i)$ be the objective value of scoring rule $\hat{\score}_i$ when the marginal distribution is $\marginalReport_i$. 
Let $i^* = \argmax_i \objfunc(\hat{\score}_i,\marginalReport_i)$.
It is easy to verify that 
\begin{align*}
\objfunc(\hat{\score},\marginalReport) = \sum_i(\bar{s}_i-\underline{s}_i)\cdot\objfunc(\hat{\score}_i,\marginalReport_i)
\leq \max_i \objfunc(\hat{\score}_i,\marginalReport_i)
= \objfunc(\hat{\score}_{i^*},\marginalReport_{i^*}). 
\end{align*}
Consider the \mos scoring rule $\Tilde{\score}$ with single-dimensional scoring rules $(\Tilde{\score}_1, \dots, \Tilde{\score}_n)$
where $\Tilde{\score}_{i^*}= \hat{\score}_{i^*}$
and $\Tilde{\score}_i \equiv 0$ for any $i\neq i^*$.
The objective value of \mos scoring rule $\Tilde{\score}$ coincides with the objective value of $\Tilde{\score}_{i^*}$
since the agent always chooses dimension $i^*$ to be scored. 
Therefore, 
\begin{equation*}
\objfunc(\Tilde{\score},\marginalReport) = \objfunc(\Tilde{\score}_{i^*},\marginalReport_{i^*})
= \objfunc(\hat{\score}_{i^*},\marginalReport_{i^*})
\geq \objfunc(\hat{\score},\marginalReport)
= \objfunc(\score,\marginalReport)
\end{equation*}
and the \mos scoring rule has weakly higher objective value. \end{proof}

% \section{Missing Proofs in Section \ref{subsec:max over separate}}
% \label{apx:closed}

\subsection{Proofs of Lemma~\ref{lem:opt-separate}-Lemma~\ref{lem:symmetry-in-ext}}
\label{app:opt-separate}

% \begin{lemma}
%   %\label{lem:opt-separate}
%   Evaluated on any distribution over posterior means
%   $\marginalReport$, the optimal \mos scoring rule for the
%   distribution $\marginalReport$ and the state space $\statespace$ is
%   at least as good as the optimal scoring rule for the extended
%   distribution $\extMarginalReport$ and the extended state space
%   $\extstatespace$.
% \end{lemma}

\begin{proof}[Proof of \Cref{lem:opt-separate}]
This result follows because the extended distribution is symmetric
on the extended state space, thus, its optimal scoring rule is \mos
(\Cref{cor:opt_symm_nd}).  This scoring rule can be applied to the
original space where it is still max-over-separate. The optimal \mos scoring
rule for the original space is no worse.
\end{proof}

% \begin{lemma}
%   %\label{lem:extopt-to-orig}
%   The symmetric optimizer $\extutil$ for the symmetric extended
%   distribution $\extMarginalReport$ and extended state
%   space $\extstatespace$ attains the same objective
%   value on the original distribution $\marginalReport$, i.e.,
%   $\objfunc(\extutil,\marginalReport) =
%   \OPT(\extMarginalReport,\scorebound, \extstatespace)$.
% \end{lemma}

\begin{proof}[Proof of \Cref{lem:extopt-to-orig}]
  Let $\extutil$ be the optimal utility function corresponding to
  $\OPT(\extMarginalReport, \scorebound, \extstatespace)$. Since the
  distribution $\extMarginalReport$ is center symmetric,
  by \Cref{thm:opt_symm_nd}, the utility function $\extutil$ is
  symmetric V-shaped. Thus, we have
\begin{align*}
  \OPT(\extMarginalReport, \scorebound, \extstatespace)
  &= \int_{\extreportspace} \extutil(\report) \dd \extMarginalReport(\report)\\
  &= \frac{1}{2}\int_{\reportspace} \extutil(\report)\dd\marginalReport(\report) + \frac{1}{2}\int_{\reportspace} \extutil(2\priorMean - \report)\dd\marginalReport(\report) \\
  &= \int_{\reportspace} \extutil(\report)\dd\marginalReport(\report) = \objfunc(\extutil,\marginalReport). \qedhere
\end{align*}
  \end{proof}

% \begin{lemma}
%   %\label{lem:symmetry-in-ext}
%   On extended state space $\extstatespace$, the optimal value of
%   \Cref{eq:program} for the symmetric extended
%   distribution $\extMarginalReport$ is at least half that for the
%   original distribution $\marginalReport$, i.e.,
%   $\OPT(\extMarginalReport,\scorebound, \extstatespace) \geq
%   \frac{1}{2} \OPT(\marginalReport,\scorebound, \extstatespace)$.
% \end{lemma}

\begin{proof}[Proof of \Cref{lem:symmetry-in-ext}]
  Let $\extreputil$ be the optimal solution of
  \Cref{eq:program} with distribution $\marginalReport$ and
  state space $\extstatespace$, i.e., $\objfunc(\extreputil,
  \marginalReport) = \OPT(\marginalReport, \scorebound,
  \extstatespace)$.  On the other hand, utility function $\extreputil$
  may not be optimal for distribution $\extMarginalReport$, thus,
  $\OPT(\extMarginalReport, \scorebound,
  \extstatespace) \geq \objfunc(\extreputil,
  \extMarginalReport)$.  We have, 
  \begin{align*}
  \OPT(\extMarginalReport,\scorebound,\extstatespace)
  &\geq \objfunc(\extreputil,\extMarginalReport)
  =
\int_{\extreportspace} \extreputil(\report)\dd\extMarginalReport(\report)  \\
&= \frac{1}{2} \int_{\reportspace} \extutil(\report)\dd\marginalReport(\report) 
+ \frac{1}{2} \int_{\reportspace} \extutil(2\priorMean - \report)\dd\marginalReport(\report)\\
&\geq \frac{1}{2} \int_{\reportspace} \extutil(\report)\dd\marginalReport(\report) = \frac{1}{2} \OPT(\marginalReport, \scorebound, \extstatespace) 
  \end{align*}
  where the final inequality follows from convexity of $\extreputil$, $\int_{\reportspace} (2\priorMean - \report)\dd\marginalReport(\report) = \priorMean$, Jensen's Inequality, and $\extreputil(\priorMean) = 0$.
\end{proof}

\subsection{Proof of Lemma~\ref{lem:obj ratio}}

\label{sss:lem-obj-ratio}

The approach to proving \Cref{lem:obj ratio}, i.e.,
$\OPT(\marginalReport,\scorebound, \extstatespace) \geq \frac{1}{4}
\OPT(\marginalReport,\scorebound, \statespace)$, is as follows.  Let
$\util$ be the optimal utility corresponding to
$\OPT(\marginalReport,\scorebound, \statespace)$.  We construct
$\extutil$ that (a) exceeds $\util$ at all point $\report
\in \reportspace$ and (b) is feasible for
$\OPT(\marginalReport,4\scorebound, \extstatespace)$.  The utility
function $\extutil/4$, thus, has objective value at least $\frac{1}{4}
\OPT(\marginalReport,\scorebound, \statespace)$ and is feasible for
$\OPT(\marginalReport,\scorebound, \extstatespace)$.  The optimal
utility is only better.

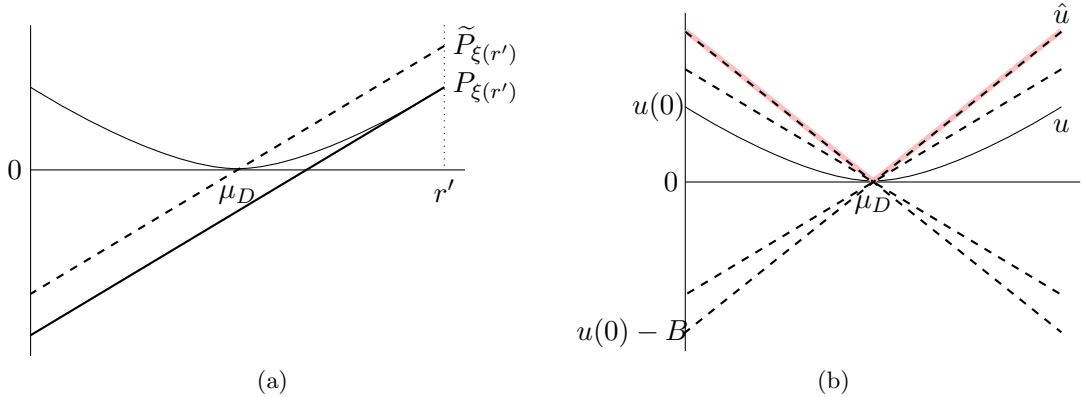
\begin{figure}[t]
\centering
\subfloat[]{
\begin{tikzpicture}[scale = 0.55]

% \draw [white] (0, 0) -- (15, 0);
\draw (0,0) -- (10.5, 0);
\draw (0, -4.5) -- (0, 3.5);

\draw plot [smooth, tension=0.7] coordinates {
(0, 2) (5,0.03) (10, 2)
};

\draw [thick] (10,2) -- (0,-4);
\draw [dashed,thick] (10,3) -- (0,-3);
%\draw [dotted] (0, 3) -- (10, -3);

\draw (-0.38, 0) node {$0$};
%\draw (10.3, -0.5) node {$1$};
 \draw (5, -0.6) node {$\priorMean$};
% \draw (7, 0) -- (7, 0.2);
% \draw (10, 0) -- (10, 0.2);

% \draw (0, 3) -- (0.2, 3);
% \draw (0, 2.2) -- (0.2, 2.2);

% \draw (-0.8, 3) node {$\util(0)$};
% \draw (-0.8, 2.2) node {$\util(1)$};
% \draw (-1.9, -3) node {$\util(0)+\util'(0)$};
% \draw (0, -3) -- (0.2, -3);

% \draw [dotted] (10, 2.2) -- (10, -3);
% \draw [dotted] (0, 2.2) -- (10, 2.2);
% \draw [dotted] (0, -3) -- (10, -3);

\draw [dotted] (10,3.5) -- (10,0);
\draw (10, -0.5) node {$\report'$};

\draw (11, 2) node {$\plane_{\sg(\report')}$};
\draw (11, 3) node {$\widetilde{\plane}_{\sg(\report')}$};
% \draw (5, 0) -- (5, 0.2);
% \draw [dotted] (5, 0) -- (5, 0.5);

% \draw (9.3, -0.5) node {$\state$};
% \draw (9, 0) -- (9, 0.2);
% \draw [dotted] (9, 0.7) -- (9, -0.9);

% \draw [dotted] (5, 0.5) -- (9, -0.9);

% \draw (-2.8, -0.9) node {$\util(\report) + \util'(\report)(\state - \report)$};
% \draw (0, -0.9) -- (0.2, -0.9);
% \draw [dotted] (0, -0.9) -- (9, -0.9);

% \draw (-0.8, 0.7) node {$\util(\theta)$};
% \draw (0, 0.7) -- (0.2, 0.7);
% \draw [dotted] (0, 0.7) -- (9, 0.7);

\end{tikzpicture}
\label{fig:util_construct}
}
\subfloat[]{
\begin{tikzpicture}[scale = 0.5]

% \draw [white] (0, 0) -- (15, 0);
\draw (0,0) -- (10.5, 0);
\draw (0, -4.5) -- (0, 4.5);

\draw plot [smooth, tension=0.7] coordinates {
(0, 2) (5,0.03) (10, 2)
};

\draw [line width = 2pt, color=pink] plot [smooth, tension=0.01] coordinates {
(0, 4.03) (5,0.03) (10, 4.03)
};

\draw (-0.9, 2) node {$\util(0)$};
\draw (-1.7, -4) node {$\util(0) - \scorebound$};

%\draw (10,2) -- (0,-4);
\draw [dashed,thick] (10,3) -- (0,-3);

%\draw (0,2) -- (10,-4);
\draw [dashed,thick] (0,3) -- (10,-3);

\draw [dashed,thick] (0,4) -- (10,-4);
\draw [dashed,thick] (10,4) -- (0,-4);

\draw (10,1.5) node {$\util$};
\draw (10,4.5) node {$\hat{\util}$};

%\draw [dotted] (0, 3) -- (10, -3);

\draw (-0.38, 0) node {$0$};
%\draw (10.3, -0.5) node {$1$};
 \draw (5, -0.6) node {$\priorMean$};
% \draw (7, 0) -- (7, 0.2);
% \draw (10, 0) -- (10, 0.2);

% \draw (0, 3) -- (0.2, 3);
% \draw (0, 2.2) -- (0.2, 2.2);

% \draw (-0.8, 3) node {$\util(0)$};
% \draw (-0.8, 2.2) node {$\util(1)$};
% \draw (-1.9, -3) node {$\util(0)+\util'(0)$};
% \draw (0, -3) -- (0.2, -3);

% \draw [dotted] (10, 2.2) -- (10, -3);
% \draw [dotted] (0, 2.2) -- (10, 2.2);
% \draw [dotted] (0, -3) -- (10, -3);

%\draw [dotted] (10,3.5) -- (10,0);
%\draw (10, -0.5) node {$\report'$};

%\draw (11, 2) node {$\plane_{\sg(\report')}$};
%\draw (11, 3) node {$\widetilde{\plane}_{\sg(\report')}$};
% \draw (5, 0) -- (5, 0.2);
% \draw [dotted] (5, 0) -- (5, 0.5);

% \draw (9.3, -0.5) node {$\state$};
% \draw (9, 0) -- (9, 0.2);
% \draw [dotted] (9, 0.7) -- (9, -0.9);

% \draw [dotted] (5, 0.5) -- (9, -0.9);

% \draw (-2.8, -0.9) node {$\util(\report) + \util'(\report)(\state - \report)$};
% \draw (0, -0.9) -- (0.2, -0.9);
% \draw [dotted] (0, -0.9) -- (9, -0.9);

% \draw (-0.8, 0.7) node {$\util(\theta)$};
% \draw (0, 0.7) -- (0.2, 0.7);
% \draw [dotted] (0, 0.7) -- (9, 0.7);

\end{tikzpicture}
\label{fig:util_new}
}
\caption{\label{f:util_new}
The figure on the left hand side illustrates a  hyperplane for report $\report'$ on the boundary of the report space, 
which is shifted from a tangent plane of $\util$ at the boundary $\report'$. 
The figure on the right hand side illustrates the extended utility function $\extutil$ 
that takes the supremum over all hyperplanes shifted from the feasible tangent planes to intersect with the
$(\priorMean,0)$ point.}
\end{figure}

The proof of the lemma introduces the following constructs.

\begin{itemize}
  \item The {\em extended utility function} $\extutil$ for program
    $\OPT(\marginalReport,4\scorebound, \extstatespace)$ given
    utility function $\util$ for the program
    $\OPT(\marginalReport,\scorebound, \statespace)$ is defined as
    follows.

    Feasibility of $\util$ for \Cref{eq:program} defines subgradients
    $\{\sg(\report) : \report \in \reportspace\}$ that satisfy the
    boundedness condition.  Let $\sgset_{\util}$ be the set of all
    subgradients of $\util$ that satisfy the boundedness constraint.
    Clearly the latter set contains the former set.  Define the
    extended utility function $\extutil$ as the convex function
    defined by the supremum of the supporting hyperplanes given by the
    subgradients $\sgset_{\util}$ shifted to intersect with the
    $(\priorMean,0)$ point.  See \Cref{f:util_new}.
    
    Convexity of $\util$ implies that its supporting hyperplane at
    $\report$ with subgradient $\sg(\report)$ is below
    $\util(\priorMean) = 0$ at $\priorMean$.  Thus, relative to the
    supporting hyperplanes of $\util$ these supporting hyperplanes of
    $\extutil$ are shifted upwards.

    The extended utility function $\extutil$ is {\em convex-conical} as it is
    defined by supporting hyperplanes that all contain point
    $(\priorMean,0)$.

  \item The {\em extended state spaces} are $\statespace \subset
    \extsymstatespace \subset \extconvstatespace
    \subset \extstatespace$.  State space $\extsymstatespace$ is the
    union of the original state space and its point reflection about
    $\priorMean$ as $\extsymstatespace = \statespace \cup
    \{2\priorMean - \state : \state \in \statespace\}$, state space
    $\extconvstatespace$ is the convex hull of $\extsymstatespace$,
    and state space $\extstatespace$ (as previously defined) is the
    extended rectangular state space containing $\extconvstatespace$.
\end{itemize}
\Cref{lem:obj ratio}, i.e., $\OPT(\marginalReport,4\scorebound, \extstatespace) \geq
  \OPT(\marginalReport,\scorebound, \statespace)$, follows by combining the following lemmas.

\begin{lemma}\label{lem:util function improve}
  \label{lem:first-of-boundedness}
For any feasible solution $\util$ for \Cref{eq:program}, the extended
utility function $\extutil$ is at least $\util$, i.e.,
$\extutil(\report) \geq \util(\report)$ for any report
$\report\in\reportspace$.
\end{lemma}

\begin{lemma}
  \label{lem:expansion_bounded_R}
  For any feasible solution $\util$ for \Cref{eq:program} with score
  bound $\scorebound$ and state space $\statespace$, the extended
  utility function $\extutil$ is a feasible solution of
  \Cref{eq:program} with score bound $2\scorebound$ and state space
  $\statespace$.
\end{lemma}

\begin{lemma}
  \label{lem:extsymstatespace}
Any convex-conical utility function $\extutil$ that is a feasible
solution of \Cref{eq:program} with score bound $2\scorebound$ and
state space $\statespace$ is a feasible solution to \Cref{eq:program}
with bound $2\scorebound$ and state space~$\extsymstatespace$.
\end{lemma}

\begin{lemma}
  \label{lem:convHullBounded}
  Any convex-conical utility function $\extutil$ that is a feasible solution
  of \Cref{eq:program} with score bound $2\scorebound$ and
  state space $\extsymstatespace$ is a feasible solution to \Cref{eq:program} with bound
  $2\scorebound$ and state space $\extconvstatespace = \conv(\extsymstatespace)$.
\end{lemma}

\begin{lemma}\label{lem:extension_bound_RH}
  \label{lem:last-of-boundedness}
  Any convex-conical utility function $\extutil$ that is a feasible solution
  of \Cref{eq:program} with score bound $2\scorebound$ and
  state space $\extconvstatespace$ is a feasible solution to \Cref{eq:program} with bound
  $4\scorebound$ and state space~$\extstatespace$.
\end{lemma}

\begin{proof}[Proof of \Cref{lem:util function improve}]
  %% Feasibility of $\util$ for \Cref{eq:program} defines subgradients
  %% $\{\sg(\report) : \report \in \reportspace\}$.  Let $\extutil$ be
  %% the convex function defined by the supporting hyperplanes given by
  %% the subgradients $\sg$ of $\util$ shifted to intersect with the
  %% $(\priorMean,0)$ point. Convexity of $\util$ implies that supporting
  %% hyper plane at $\report$ with subgradient $\sg(\report)$ is below
  %% $\util(\priorMean) = 0$ at $\priorMean$.  Thus, relative to the
  %% supporting hyperplanes of $\util$ these supporting hyperplanes of
  %% $\extutil$ are shifted upward.
  Since the supporting hyperplanes of $\extutil$ are shifted upwards
  relative to $\util$, we have $\extutil(\report) \geq \util(\report)$
  at all $\report \in \reportspace$.  Thus, $\extutil$ obtains at
  least the objective value of $\util$, i.e.,
  $\objfunc(\marginalReport,\extutil) \geq
  \objfunc(\marginalReport,\extutil)$.
\end{proof}

\begin{proof}[Proof of \Cref{lem:expansion_bounded_R}]
  First, the subgradients of $\extutil$ are a subset of the
  subgradients of $\util$ that satisfy the boundedness constraint.
  \Cref{lem:bounded subgradient is close}
  (stated and proved at the end of this subsection)
  shows that the set of subgradients
  $\sgset_{\util}$ of $\util$ that satisfy the boundedness constraint
  is closed.  As $\extutil$ is defined the supremum over these
  hyperplanes, closure of the set implies that the supremum at any
  report $\report \in \reportspace$ is attained on one of these
  hyperplanes.
  
  Now observe that in the construction of $\extutil$, the supporting
  hyperplanes of $\util$ are shifted up by at most $\scorebound$.  The
  boundedness constraint corresponding to state $\priorMean$ and the
  report $\report$ with subgradient $\sg(\report) \in
  \subgradient\util(\report)$ implies that the supporting hyperplane
  corresponding to $\sg(\report)$ at $\report$ has value at least
  $-\scorebound$ at $\priorMean$.  Thus, in the construction of the
  extended utility function $\extutil$, the hyperplane corresponding
  to $\sg(\report)$ is shifted up by at most $\scorebound$ and, at any
  state $\state \in \statespace$, $\extutil(\state) \leq \util(\state)
  + \scorebound$.

  Finlly, the boundedness constraint is the difference between the
  utility at a given state and the value of any supporting hyperplane of
  the utility evaluated at that state.  From $\util$ to $\extutil$ the
  former has increased by at most $\scorebound$ and the latter is no
  smaller; thus, $\extutil$ satisfies the boundedness constraint on
  state space $\statespace$ with bound $2\scorebound$.
\end{proof}

\begin{proof}[Proof of \Cref{lem:extsymstatespace}]
  The lemma follows by the geometries of the boundedness constraint
  and convex cones.  The boundedness constraint requires a bounded
  difference between the utility at any state (in the state space) and
  the value at that state on any supporting hyperplane of the utility
  function (corresponding to any report in the report space).  For convex-conical utility
  functions, the supporting hyperplanes are also supporting
  hyperplanes of the cone defined by the point reflection of the
  utility function around its vertex $(\priorMean,0)$, henceforth, the
  reflected cone.  Thus, the boundedness constraint for convex-conical
  utility function requires that the difference between the original
  cone and the reflected cone be bounded at all states in the state
  space.

  The original space $\statespace$ and the reflected state space
  $\{2\priorMean - \state : \state \in \statespace\}$ are symmetric
  with respect to the original cone and the reflected cone.  Thus, if
  states in the original state space are bounded, by comparing a state
  on the cone to the same state on the reflected cone; then states in
  the reflected state space are bounded by comparing its reflected
  state (in the original state space) on the reflected cone to its
  reflected state on the original cone.

  Thus, if a boundedness constraint holds on $\statespace$ it also
  holds on the reflected state space $\{2\priorMean - \state : \state
  \in \statespace\}$ and their union.
\end{proof}

\begin{proof}[Proof of \Cref{lem:convHullBounded}]
  Consider the cone and reflected cone defined in the proof of
  \Cref{lem:extsymstatespace} and the geometry of the boundedness
  constraint.  Notice that, by convexity of the cone defining the
  utility function $\extutil$ and concavity of the reflected cone, the
  convex combination of the bounds, i.e., the difference of values of
  states on these two cones, of any set of states is at least the bound
  of the convex combination of the states.  Hence, if the boundedness
  constraint holds on state space $\extsymstatespace$, then it holds on
  its convex hull $\extconvstatespace = \conv(\extsymstatespace)$.
\end{proof}

\begin{proof}[Proof of \Cref{lem:extension_bound_RH}]
  Consider any ray from $\priorMean$.  Since the utility $\extutil$ is
  a convex cone, the utility on this ray is a linear function of the
  distance from $\priorMean$.  The same holds for this ray evaluated
  on the point reflection of the utility at $\priorMean$.  The
  difference between these utilities is also linear.  Thus, by the
  geometry of the boundedness constraint for convex-conical utility
  functions, on any ray from $\priorMean$, the bound is linear.
  Considering the state space $\extconvstatespace$ and
  $\extstatespace$, if the former is scaled by a factor of two around
  $\priorMean$, then it contains the latter (by simple geometry, see
  \Cref{fig:convHull}).  Thus, if the convex-conical utility function
  $\extutil$ satisfies bound $2\scorebound$ on state space
  $\extconvstatespace$ it satisfies bound $4\scorebound$ on state
  space $\extstatespace$.
\end{proof}

\begin{lemma}\label{lem:bounded subgradient is close}
For any feasible solution $\util$ for \Cref{eq:program}, the set
$\sgset_\util$ of all subgradients of $\util$ satisfying the bounded
constraints is a closed set.
\end{lemma}
\begin{proof}
By \Cref{thm:interim utility on the boundary}, any feasible solution $\util$ for \Cref{eq:program} is convex, bounded and continuous with bounded subgradients. 
For any convex, bounded and continuous function $\util$, let $\{\sg^k(\report^k)\}_{k=1}^{\infty} \subseteq \sgset_{\util}$ 
be a convergent sequence of subgradients in set $\sgset_{\util}$, 
where $\report^k$ is the report corresponds to the $k^{th}$ subgradient. 
Let $\sg^* = \lim_{k \to \infty}\sg^k(\report^k)$ be the limit of the subgradients. 
Since %the boundary of 
the report space is a closed and bounded space, 
there exists a subsequence of reports
$\{\report^{k_j}\}_{j=1}^{\infty} \subseteq \{\report_k\}_{k=1}^{\infty}$
such that $\{\report^{k_j}\}_{j=1}^{\infty}$ converges. Letting report $\report = \lim_{j \to \infty} \report^{k_j}$,
we have report $\report$ is in the report space,
%on the boundary, 
i.e., $\report \in \reportspace$. 
Moreover, we have $\lim_{j \to \infty}\sg^{k_j}(\report^{k_j}) = \lim_{k \to \infty}\sg^k(\report^k) = \sg^*$.
Next we show that $\sg^*$ is a subgradient for some report $\report \in \reportspace$
such that the bounded constraints of the induced scoring rule are satisfied for any state $\state \in \statespace$,
i.e., $\sg^* \in \sgset_{\util, \report}$.
% $\util(\state)-\util(\report)-\sg^* \cdot(\state-\report) \leq \scorebound$.

First for any state $\state$, we have
\begin{align*}
&\util(\report)+\sg^* \cdot(\state-\report)
= \lim_{j\to\infty} [\util(\report^{k_j})+\sg^* \cdot(\state-\report^{k_j})]\\
&= \lim_{j\to\infty} [\util(\report^{k_j})+\sg^{k_j}(\report^{k_j}) \cdot(\state-\report^{k_j})]
\leq \util(\state),
\end{align*}
where the first equality holds because function $\util$ and function $\sg^* \cdot \report$ are continuous and bounded in reports. 
%Second equality holds because function $\sg^* \cdot \hat{\report}^k$ are continuous and bounded in subgradients. 
The inequality holds because $\sg^{k_j}(\report^{k_j})$ is a subgradient for report $\report^{k_j}$.
Thus $\sg^*$ is subgradient for report $\report$. 
Next we show that the scoring rule induced by subgradient $\sg^*$ is bounded for report $\report$. 
For any state $\state$, we have
\begin{align*}
\util(\state)-\util(\report)-\sg^* \cdot(\state-\report)
&= \util(\state) - \lim_{j\to\infty} [\util(\report^{k_j})+\sg^{k_j}(\report^{k_j}) \cdot(\state-\report^{k_j})]\\
&\leq \util(\state) - (\util(\state)-\scorebound) = \scorebound,
\end{align*}
where the inequality holds because the subgradient $\sg^{k_j}(\report^{k_j})$ satisfies the bounded constraint for report $\report^{k_j}$ at state $\state$,
i.e., $\sg^{k_j}(\report^{k_j}) \in \sgset_{\util, \report^{k_j}}$
and $\util(\report^{k_j})+\sg^{k_j}(\report^{k_j}) \cdot(\state-\report^{k_j} )\geq \util(\state)-\scorebound$. 
Therefore, $\sg^* \in \sgset_{\util, \report} \subset \sgset_{\util}$,
which implies the set $\sgset_{\util}$ is a closed set. 
\end{proof}

\section{Elicition of Full Distribution}
\label{apx:report full}

\begin{proposition}\label{thm:approx betting}
Letting $\objfunc(c)$ be the objective value of the betting mechanism with parameter~$c$. 
We have that $\objfunc(c)$ is concave in $c$ and 
$\objfunc(\frac{1}{2})\geq\frac{1}{2}\cdot\max_c\objfunc(c)$.
% The betting scoring rule with mean at $1/2$ is a $2$-approximation for eliciting the full distribution.
\end{proposition}
\begin{proof}
For any $c_1, c_2$ and $c=\alpha c_1 + (1-\alpha)c_2$ for $\alpha\in [0,1]$, 
let $u_1, u_2$ and $u$ be the utility function of the agent in betting mechanism with parameter $c_1,c_2$ and $c$ respectively. 
By \cref{lem:frac knapsack}, we have 
\begin{align*}
u(\prior)=c=\alpha c_1 + (1-\alpha)c_2
= \alpha u_1(\prior) + (1-\alpha)u_2(\prior).
\end{align*}
Moreover, for any posterior $\posterior$, one feasible choice is to mix the score from betting mechanism with parameter $c_1$ and $c_2$ with probability $\alpha$, 
and hence 
\begin{align*}
u(\posterior)\geq \alpha u_1(\posterior) + (1-\alpha)u_2(\posterior).
\end{align*}
Combining the inequalities, we have 
\begin{align*}
\objfunc(c)&=\int_{\reportspace}\left[ \util(G)-\util(\statedist)\right]
\distoverposterior(G) \,\de G\\
&\geq \alpha\cdot \int_{\reportspace}\left[ \util_1(G)-\util_1(\statedist)\right]
\distoverposterior(G) \,\de G
+ (1-\alpha)\cdot 
\int_{\reportspace}\left[ \util_2(G)-\util_2(\statedist)\right]
\distoverposterior(G) \,\de G \\
&= \alpha\cdot \objfunc(c_1) + (1-\alpha)\cdot \objfunc(c_2),
\end{align*}
and hence the objective function is concave.

Next we show that $\objfunc(\frac{1}{2})$ is approximately optimal using the concavity of the objective value function. 
Let $c^* = \argmax_c \objfunc(c)$. 
If $c^* \leq \frac{1}{2}$, there exists $\alpha\in [\frac{1}{2},1]$ such that 
$\frac{1}{2}=\alpha c^* + (1-\alpha)$, 
and by concavity, we have that 
\begin{align*}
\objfunc\left(\frac{1}{2}\right) \geq \alpha \cdot \objfunc(c^*) + (1-\alpha)\objfunc(1)\geq \alpha \cdot \objfunc(c^*)
\geq \frac{1}{2} \objfunc(c^*).
\end{align*}
Similarly, if $c^* \geq \frac{1}{2}$, there exists $\alpha\in [\frac{1}{2},1]$ such that 
$\frac{1}{2}=\alpha c^*$, 
and by concavity, we have that 
\begin{equation*}
\objfunc\left(\frac{1}{2}\right) \geq \alpha \cdot \objfunc(c^*) + (1-\alpha)\objfunc(0)\geq \alpha \cdot \objfunc(c^*)
\geq \frac{1}{2} \objfunc(c^*). \qedhere
\end{equation*}
\end{proof}

\begin{proof}[Proof of \cref{thm:betting_full}]
Given any feasible scoring rule $\score$, we show that there exists a betting mechanism that achieves objective value weakly higher than $\score$. %Then, we show that this betting scoring rule is also a feasible solution for Program. 

Let $\util$ be the utility function of scoring rule $S$. We consider the betting mechanism with parameter $c = \util(\statedist)$. 
For every posterior $\posterior \in \reportspace$, the agent can achieve at least the same expected utility $\util(\posterior)$ by choosing the score $s(\theta) = \score(\posterior,\theta)$ for each state $\theta$. 
Since the agent chooses the bets optimally, the utility function $\hat{\util}$ of the betting mechanism with parameter $c = \util(\statedist)$ satisfies $\hat{\util}(\posterior) \geq \util(\posterior)$ 
for any report $\posterior \in \reportspace$ and $\hat{\util}(\statedist) = \util(\statedist)$. 
Therefore, the objective value of betting mechanism is weakly higher, 
and the optimal scoring rule must be the revelation version
of the betting mechanism.
\end{proof}

\subsection{Proof of Theorem~\ref{thm:gap between full dist and mean}}\label{sec:appendix-proof-gap between full dist and mean}
\begin{proof}
Consider the following single dimensional problem with state space
$\statespace = \{0, \sfrac 1 2 - \epsilon, \sfrac 1 2 + \epsilon,
1\}$.  The distribution over posteriors is
\begin{enumerate}
\item pointmass distributions at state $0$ and $1$ with probability
  $\sfrac{\epsilon}{2}$ each.
\item pointmass distributions at state $\sfrac 1 2 -\epsilon$ and $\sfrac 1 2 + \epsilon$ with probability $\sfrac{(1-\epsilon)}{2}$ each.
\end{enumerate}
Thus, the prior mean is $\priorMean = \sfrac{1}{2}$ and by
\Cref{c:symmetric-single-dimensional} the optimal scoring rule for
reporting the mean is V-shaped with $\util(0) = \util(1) = \sfrac 1 2$
and $\util(\sfrac 1 2) = 0$.  Utility is linear above and below the
mean with magnitude of its slope equal to 1; thus, $\util(\sfrac 1 2
\pm \epsilon) = \epsilon$.  The expected utility under the above
distribution is
\begin{align*}
  \expect{\report \sim \marginalReport}{\util(\report)}
  &= \tfrac 1 2 \, \epsilon + \epsilon\, (1-\epsilon) \leq \epsilon,
\end{align*}
assuming $\epsilon \leq \sfrac 1 2$.

Consider the following mechanism for reporting the full distribution.
The designer combines the low states as $L = \{0,\sfrac 1 2 -
\epsilon\}$ and the high states as $H = \{\sfrac 1 2 + \epsilon,1\}$
and uses a scoring rule for the indicator variable that the state $\state$ is
high, i.e., the variable is 1 if $\state \in H$ and 0 if $\state \in L$.  Note
that for Bernoulli distributions, reporting the distribution is
equivalent to reporting the mean of the distribution.  The mean of the
posteriors of this indicator variable is $\priorMean = \sfrac 1 2$.
For the indicator on high states, the symmetric V-shaped utility
function of \Cref{c:symmetric-single-dimensional} is optimal. 
Its performance is
\begin{align*}
  \expect{\report \sim \marginalReport}{\util(\indicate{\report \in H})} &= \sfrac 1 2.
\end{align*}

Combining these two analyses, the approximation factor of the optimal
scoring rule for the mean is at least $\sfrac 2 \epsilon$.  As
$\epsilon$ approaches zero, the approximation ratio is unbounded.
\end{proof}

\newpage
\section*{Online Appendix for ``Optimization of Scoring Rules''}
\label{apx:additional}
\setcounter{page}{1}
\section{The Quadratic Scoring Rule and Prior-independent Approximation}
\label{sec:quadratic}

The previous section showed that the optimal single-dimensional
scoring rule depends on the distribution over posteriors and, more
specifically, on the mean of this distribution.  On the other hand,
standard scoring rules in theory and practice, like the quadradic
scoring rule, are prior-independent, i.e., they do not depend on the
principal's prior distribution (over posterior distributions of the
agent), cf.\ \citet*{DRY-15}.  This section focuses on the quadratic
scoring rule.  It gives the characterization in terms of utility of
the quadratic scoring rule for eliciting the mean of a
single-dimensional state.  It analizes the approximation factor of the
quadratic scoring rule with respect to the optimal scoring rule, and
shows that the performance of the former is quadratic in the
performance of the latter.  Specifically, the ratio of performances is
unbounded as the performance of the optimal scoring rule approaches
zero (and such a sequence of prior distributions exists).  Thus, we
conduct the prior-independent analysis on families of priors which
give the same performance of the optimal scoring rule (cf.\ the
``max/max ratio'' of \citealp*{BB-94}).  Within each such family, the
quadratic rule is approximately optimal among all prior-independent
scoring rules.

The following observations will be useful in our analysis of the
quadratic and other prior-independent scoring rules.  First, for
prior-independent analysis, the designer does not know the prior
mean~$\priorMean$ of the distribution.  Therefore, we consider
\Cref{eq:1d_program_u} equivalently with the agent's utility for
reporting the prior mean $\util(\priorMean)$ subtracted from the
objective and without the constraint $\util(\priorMean) = 0$. Second,
in the worst case it is sufficient to only consider posterior
distributions that are uniformly drawn as one of two deterministic
points.  This latter result is formalized in the following lemma.

\begin{lemma}\label{lem:2 point mass}
For any distribution $\marginalReport$ over posterior means, there
exists another distribution $\widetilde{\marginalReport}$ over
posterior means with 2 point masses that satisfies
$\OPT(\widetilde{\marginalReport}) = \OPT(\marginalReport)$ and for
any convex function $\util$, $\objfunc(\util,\widetilde{\marginalReport}) \leq \objfunc(\util,\marginalReport)$.
\end{lemma}
\begin{proof}
For any distribution $\marginalReport$ with prior mean $\priorMean$, 
let $\widetilde{\marginalReport}$ be the distribution that has
\begin{itemize}
\item a point mass at $\expect{\marginalReport}{\report' | \report' < \priorMean}$
with probability $\prob{\marginalReport}{\report' < \priorMean}$;
\item a point mass at $\expect{\marginalReport}{\report' | \report' \geq \priorMean}$
with probability $\prob{\marginalReport}{\report' \geq \priorMean}$.
\end{itemize}
By \Cref{thm:1d_opt}, it is easy to verify that the optimal does not
change, i.e., $\OPT(\marginalReport) =
\OPT(\widetilde{\marginalReport})$, and for any convex $\util$, by
Jensen's Inequality, we have $\objfunc(\util,\widetilde{\marginalReport}) \leq \objfunc(\util,\marginalReport)$.
\end{proof}

The quadratic scoring rule that is the focus of this section is
defined as follows.

\begin{definition}\label{def:quadratic}
  The $[0,1]$-bounded {\em quadratic scoring rule} for eliciting the
  mean with state and report spaces $\statespace = \reportspace =
  [0,1]$ is $\quadra(\report,\state) = 1- (\state-\report)^2$.  For
  functions $\utilquad(\report) = \report^2$ and $\kappa_q(\state) =
  1-\state^2$ the quadratic scoring rule is $\quadra(\report,\state) =
  \utilquad(\report) +\utilquad'(\report) \cdot
  (\state-\report)+\kappa_q(\state).$
\end{definition}
\Cref{lem:2 point mass} enables the identification of the worst-case
performance the quadratic scoring rule.  
Moreover, note that it is easy to verify that 
% Recall that, by \Cref{cor:upper bound opt}, 
the optimal objective value is at most
$\sfrac{1}{2}$, i.e., $\OPT(\marginalReport) \in (0,\sfrac{1}{2}]$.

\begin{lemma}\label{lem:maxmin quadratic}
Let $\distset$ be the set of distributions 
such that the objective value of the optimal scoring rule is $c\in (0,\sfrac 1 2]$, 
i.e., $\OPT(\marginalReport) = c$ for any $\marginalReport\in \distset$. 
We have that for utility function $\utilquad$ of quadratic scoring rule,  
\begin{equation*}
\min_{\marginalReport\in \distset} \objfunc(\utilquad,\marginalReport)
= c^2.
\end{equation*}
\end{lemma}
\begin{proof}
Suppose the distribution over report $\marginalReport(\report)$ has
two point masses, which is $a$ with probability $p$, and $b > a$ with
probability $1-p$.  Then, we have the mean of prior is $\priorMean =
pa+(1-p)b$ and $a<\priorMean < b$. Without loss of generality, we can
assume that $\priorMean \leq \frac{1}{2}$. By~\Cref{thm:1d_opt}, it
holds that
\begin{align}
  \label{eq:quadratic-opt}
c &= \OPT(\marginalReport) = \frac{1}{\max\{\priorMean,1-\priorMean\}} \cdot (1-p)(b-\priorMean)
= \frac{p(1-p)(b-a)}{\max\{\priorMean,1-\priorMean\}}.\\
\intertext{For quadratic scoring rule with utility function $\utilquad(\report) = \report^2$ (\Cref{def:quadratic}), we have} 
\label{eq:quadratic-obj}
\objfunc(\utilquad,\marginalReport) 
&= \expect{\report\sim\marginalReport}{\utilquad(\report)} - \utilquad(\priorMean)
= p(a^2-\priorMean^2)+(1-p)(b^2-\priorMean^2)
= p(1-p)(b-a)^2.\\
\intertext{Combining equations~\eqref{eq:quadratic-opt} and~\eqref{eq:quadratic-obj}, we have}
\notag
\objfunc(\utilquad,\marginalReport) &= (\max\{\priorMean,1-\priorMean\})^2
\cdot \frac{c^2}{p(1-p)}.
\end{align}
The worst case ratio is achieved when $\objfunc(\utilquad,\marginalReport)$
is minimized, i.e., $\priorMean = \frac{1}{2}$
and $p = \frac{1}{2}$, 
which gives $\min_{\marginalReport \in \distset}\objfunc(\utilquad,\marginalReport) = c^2.$
\end{proof}

As is evident from the proof of \Cref{lem:maxmin quadratic}, for any $c
\in (0,\sfrac 1 2]$ there is a non-trivial family of distributions
  $\distset$ for which $\OPT(\marginalReport) = \objval$.  
Since the worst-case
  performance of the quadratic scoring rule on $\distset$ is $\min_{\marginalReport
    \in \distset} \objfunc(\utilquad,\marginalReport) = \objval^2$, the
  prior-independent approximation factor of the quadratic scoring rule
  is unbounded.
  In fact, as we show next, this result is not a limitation of the
  quadratic scoring rule.  For the family of distributions $\distset$,
  any prior-independent scoring rule can at most guarantee a
  worst-case objective value of $O(\objval^2)$.  Thus, the quadratic
  rule is within a constant factor of the prior-independent optimal
  rule. We defer the proof of \Cref{thm:lower bound for pi} to \Cref{sec:appendix-proof-lower bound for pi}.

\begin{lemma}\label{thm:lower bound for pi}
Let $\distset$ be the set of distributions over posterior means such
that the objective value of the optimal scoring rule is $\objval\in
(0,\sfrac{1}{2}]$, i.e., $\OPT(\marginalReport) = \objval$ for any
  $\marginalReport\in \distset$.  For any convex and bounded utility
  function $\util$, we have
\begin{equation*}
  \min_{\marginalReport\in \distset} \objfunc(\util,\marginalReport)
  \leq \min(\tfrac 1 2,\tfrac{8\objval^2}{(1-4\objval)^2}) \leq 32\objval^2.
\end{equation*}
\end{lemma}

Combining \Cref{lem:maxmin quadratic} with \Cref{thm:lower bound for pi},
the quadratic scoring rule approximates any prior-independent scoring rule in terms of worst case payoff. 

% \yl{Move definition of OPT and P to model section}
\begin{theorem}\label{thm:pi}
For any constant $\objval \in (0, \sfrac{1}{2}]$, let $\distset$ be the set of distributions 
such that the objective value of the optimal scoring rule is $c$, i.e., $\OPT(\marginalReport) = c$ for any $\marginalReport\in \distset$. 
% Given utility function $\utilquad$ for quadratic scoring rule 
% and any convex function $\util$ satisfying the bounded constraint, 
Let $\mathcal{U}$ be the set of convex and bounded utility functions $\util$. 
For quadratic utility function $\utilquad$,
we have
% in the closure of proper regular scoring rules, we have 
\begin{equation*}
\min_{\marginalReport\in \distset} \objfunc(\utilquad,\marginalReport)
\geq \frac{1}{32} \max_{\util \in {\mathcal{U}}}\ \min_{\marginalReport\in \distset} \objfunc(\util,\marginalReport).
\end{equation*}
\end{theorem}

Note that in \Cref{thm:pi}, the 
quadratic scoring rule does not exploit the extra information that 
$\OPT(\marginalReport) = c$
and still achieves a constant approximation to the optimal max-min scoring rule 
in worst case. 

Although the quadratic scoring rule is approximately max-min optimal, 
the approximation ratio between the quadratic scoring rule and the optimal scoring rule can still grow unboundedly as the optimal objective value $\OPT(\marginalReport)$ vanishes to zero.
In the following theorem, 
we will show that for any fixed distribution over posterior mean with variance $\sigma^2$, the performance of the quadratic scoring rule is an approximation of the optimal solution within a factor of the standard deviation~$\sigma$. 
That is, the quadratic scoring rule is approximately optimal when the distribution over posterior mean is sufficiently disperse.

\begin{proposition}\label{thm:quadratic-var-disperse}
For any $\sigma \in [0,1]$, any distribution over posterior mean $\marginalReport$ with variance $\sigma^2$, 
we~have 
\begin{align*}
\objfunc(\utilquad,\marginalReport) \geq \sigma \cdot \OPT(\marginalReport).
\end{align*}
\end{proposition}
\begin{proof}
By Theorem~\ref{thm:1d_opt}, there is an optimal utility function that is V-shaped at $\mu_f$ with parameters $\abs{a},\abs{b} \leq 1$.  Thus, we have
$$
\OPT(\marginalReport) = \int_{0}^{\priorMean} a(\report - \priorMean) 
\marginalReport(\report) \text{ d}\report
+ \int_{\priorMean}^{1} b(\report - \priorMean) 
\marginalReport(\report) \text{ d}\report
 \leq \expect{r \sim \marginalReport}{\abs{r-\priorMean}}.
$$
By Definition~\ref{def:quadratic}, the objective value of the quadratic scoring rule is 
$$
\objfunc(\utilquad,\marginalReport) = \expect{r\sim \marginalReport}{\utilquad(r) - \utilquad(\priorMean)} = \expect{r\sim \marginalReport}{(r-\priorMean)^2}.
$$
By Jensen's inequality, we have
$$
\expect{r \sim f}{\abs{r-\mu_f}} = \expect{r \sim f}{\sqrt{(r-\mu_f)^2}} \leq \sqrt{\expect{r \sim f}{(r-\mu_f)^2}} = \frac{\expect{r \sim f}{(r-\mu_f)^2}}{\sigma},
$$
where the last equality is due to $\expect{r \sim f}{(r-\mu_f)^2} = \sigma^2$.
\end{proof}

\subsection{Proof of Lemma~\ref{thm:lower bound for pi}}\label{sec:appendix-proof-lower bound for pi}
To simplify the proof of \Cref{thm:lower bound for pi}, we define the
benchmark $\OPTUB$ as an approximate upper-bound on $\OPT$:
\begin{align*}
  \OPTUB(\marginalReport) = 2\OPTdenom\,\OPT(\marginalReport) = 2 \expect{\report \sim
    \marginalReport}{\maxzero{\report - \priorMean}}.
\end{align*}
Notice that $\OPTdenom \in [\sfrac 1 2,1]$; thus, $\OPT(\priorMean)
\leq \OPTUB(\priorMean) \leq 2\,\OPT(\priorMean)$.  Thus,
approximation of benchmark $\OPTUB$ is equivalent to approximation of
$\OPT$ up to a factor of two.  \Cref{thm:lower bound for pi} is
obtained from \Cref{lem:lower bound for pi} and the bound of $\objval
\leq \ubc \leq 2\,\objval$.

\begin{lemma}\label{lem:lower bound for pi}
Let $\distsetUB$ be the set of distributions over posterior means such
that benchmark $\OPTUB$ is $\ubc\in(0,\sfrac{1}{2}]$.  For any convex and
  bounded utility function $\util$, we have
\begin{equation*}
\min_{\marginalReport\in \distset} \objfunc(\util,\marginalReport)
\leq \min(\tfrac 1 2,\tfrac{2\ubc^2}{(1-2\ubc)^2}) \leq 8\ubc^2.
\end{equation*}
\end{lemma}

\newcommand{\numpart}{d}
\newcommand{\widthpart}{\delta}

\begin{proof}
  A convex and bounded utility function $\util$ has monotone
  derivative $\util'$ and, by \Cref{lem:bounded for single}, the
  amount this derivative increases on its $[0,1]$ domain is $\util'(1)
  - \util'(0)$ bounded by $2$. Consider any positive integer
  $\numpart$ and partition the $[0,1]$ domain of $\util$ 
  into $\numpart$ intervals of width $\sfrac 1 \numpart$.  By the
  pigeon hole principle, one part must contain at most the average
  increase of $\util'$, i.e., there exists interval $[a,b = a + \sfrac 1
    \numpart]$ with $\util'(b) - \util'(a) \leq \sfrac 2 \numpart$.

  \newcommand{\mrd}[1][\numpart]{\marginalReport_{#1}}
  \newcommand{\pmd}[1][\numpart]{\mean{\numpart}}

  Consider distribution $\mrd$ defined as the uniform
  distribution over deterministic points $a$ and $b$ with mean $\pmd = a + \sfrac 1 {2\numpart}$.  By the definition of benchmark $\OPTUB$:
  \begin{align*}
    \OPTUB(\mrd)
    &= 
        2\expect{\report \sim \mrd}{\maxzero{\report - \pmd}}
        = \tfrac 1 {2\numpart}.
  \end{align*}
Calculating the objective value of utility function $\util$, we have 
\begin{align*}
\objfunc(\util,\mrd) 
&= \frac{\util(a) + \util(b)}{2} - \util(\pmd)
\leq \frac{\util'(b)-\util'(a)}{2}\cdot \frac{b-a}{2}
 = \frac 1 {2\numpart^2},
\end{align*}
where the inequality follows from identifying an optimal utility $\util$
satisfying $\util'(b) - \util'(a) \leq \sfrac 2 \numpart$.  It is $\util'(\report) = \sfrac{-1} \numpart$ for $\report
\in [a,\pmd)$ and $\util'(\report) = \sfrac 1 \numpart$ for $\report \in
    (\pmd,b]$.
Combining the two bounds with $\OPTUB(\mrd) = \ubc$ we see that
$\objfunc(\util,\mrd) \leq 2\,\ubc^2$ for $\ubc \in \{ \sfrac 1
{2\numpart} : \numpart \in \{1,\ldots\}\}$.

To extend this bound to all $\ubc \in [0,\sfrac 1 2]$, observe that
the bound on $\objfunc(\util,\mrd)$ easily extends to
$\objfunc(\util,\mrd[\numpart'])$ for non-integral $\numpart' \geq
\numpart$, while the value of $\OPTUB(\mrd[\numpart'])$ holds as
calculated for non-integral $\numpart'$.  Thus, we can obtain bounds
for non-integral $\numpart'$ by combining bounds on
$\OPTUB(\mrd[\numpart + 1])$ and $\objfunc(\util,\mrd[\numpart])$.
Solving for the bound on $\objfunc(\util,\mrd[\numpart])$ in terms of
$\ubc = \OPTUB(\mrd[\numpart + 1])$: for any $\ubc \in (0,\sfrac 1 2]$
  there exists $\marginalReport \in \distsetUB$ with
  $\objfunc(\util,\marginalReport) \leq \min(\frac 1 2,\frac{2\ubc^2}{(1 - 2
    \ubc)^2}) \leq 8\, \ubc^2$.  The first inequality holds by substituting $d=\sfrac{1}{2\ubc}-1$ into the formula of $\objfunc(\util,\mrd)$,
  the second inequality uses $\objfunc(\util,\marginalReport) \leq
  \sfrac 1 2$ and notes that the bound of the first inequality is
  trivial until $\ubc \leq \sfrac 1 4$, and thereafter the denominator
  is lower bounded by $\sfrac 1 4$.
\end{proof}

\section{Robustness to Distributional Knowledge}
\label{sec:robust}
By \cref{thm:approx_max_over_separate}, 
the optimal \mos scoring rule is approximately optimal for multi-dimensional problems 
and is exactly optimal in the degenerated single-dimensional problems, 
and to implement such a scoring rule, 
it is sufficient to know the prior mean of the distribution. 
In this section, we show that 
we can even relax the assumption of exact knowledge of the prior mean, 
and show that the designer can approximately attain the performance of the optimal \mos scoring rule
by having an estimate of the prior mean.
To simplify the presentation, 
we will focus on the state space $\statespace=\bigtimes_{i=1}^{\numasm}[0,1]$
and score bound $\scorebound = 1$. 
The results can be directly extended to general rectangular state spaces
and any score bound $\scorebound > 0$. 
\begin{theorem}\label{thm:infty norm error}
For any $\epsilon>0$,
any distribution $\marginalReport$ with prior mean $\priorMeanState$ 
in state space $\statespace=\bigtimes_{i=1}^{\numasm}[0,1]$, 
% let $\util_{\priorMeanState}$ be the utility function for the optimal V-shaped scoring rule for the extended space~$\extstatespace$ according to mean $\priorMeanState$.
for any $\mu$ such that $\inftynorm{\mu-\priorMeanState}\leq \epsilon$, 
the incentive for effort of the V-shaped scoring rule for $\mu$ is at least that of the V-shaped scoring rule for $\mu_D$ less $3\epsilon$.
% we have that 
% \begin{align*}
% \objfunc(\util_{\mu},\marginalReport) 
% \geq \objfunc(\util_{\priorMeanState},\marginalReport) - 3\epsilon.
% \end{align*}
\end{theorem}

\begin{proof}
Note that by definition, 
it is easy to verify that the utility function $\util_{\priorMeanState}$ satisfies
\begin{align*}
\util_{\priorMeanState}(\report) = \max_{i} \frac{1}{2\max\{\Pmeani, 1-\Pmeani\}} \abs{\report_i-\Pmeani}
% \begin{cases}
% -\frac{1}{b_i-a_i} (\report_i-\Pmeani) & \report_i \leq \Pmeani\\
% \frac{1}{b_i-a_i} (\report_i-\Pmeani) & \report_i \geq \Pmeani
% \end{cases} \text{, and}&
% \sdkappa_i(\state_i) &= \sfrac{1}{2}. 
\end{align*}
and hence 
\begin{align*}
\objfunc(\util_{\priorMeanState},\marginalReport)
= \expect{\report\sim\marginalReport}{\max_{i} \frac{1}{2\max\{\Pmeani, 1-\Pmeani\}} \abs{\report_i-\Pmeani}}.
\end{align*}
Moreover, we have 
\begin{align*}
&\objfunc(\util_{\mu},\marginalReport) - \objfunc(\util_{\priorMeanState},\marginalReport)\\
&= \expect{\report\sim\marginalReport}{\max_{i} \frac{\abs{\report_i-\mu_i}}{2\max\{\mu_i, 1-\mu_i\}} 
- \max_{i} \frac{\abs{\report_i-\Pmeani}}{2\max\{\Pmeani, 1-\Pmeani\}} } - \util_{\mu}(\priorMeanState)\\
&\geq -3\epsilon,
\end{align*}
which implies that the incentive for effort of the V-shaped scoring rule for $\mu$ is 
at least that of the V-shaped scoring rule for $\mu_D$ less $3\epsilon$,
and the theorem holds.
Note that the last inequality holds because
\begin{align*}
\util_{\mu}(\priorMeanState) = \max_{i} \frac{\abs{\Pmeani-\mu_i}}{2\max\{\mu_i, 1-\mu_i\}} 
\leq \max_{i} \abs{\Pmeani-\mu_i} \leq \epsilon
\end{align*}
and for any dimension $i\in[\numasm]$,
\begin{align*}
&\frac{1}{2\max\{\Pmeani, 1-\Pmeani\}} \abs{\report_i-\Pmeani} 
\leq \frac{1}{2\max\{\Pmeani, 1-\Pmeani\}} (\abs{\report_i-\mu_i} +\epsilon) \\
\leq\,& \frac{1}{2\max\{\Pmeani, 1-\Pmeani\}} \abs{\report_i-\mu_i} +\epsilon
\leq \frac{1}{2\max\{\mu_i, 1-\mu_i\}} \abs{\report_i-\mu_i} +2\epsilon.\qedhere
\end{align*}
\end{proof}

%The proof of \Cref{thm:infty norm error} is deferred to \Cref{sec:appendix-proof-infty norm error}.
% Combining \Cref{thm:infty norm error,thm:approx_max_over_separate}, we have the following corollary.
% \begin{corollary}
% For any $\epsilon>0$,
% any distribution $\marginalReport$ with prior mean $\priorMeanState$ 
% in state space $\statespace=\bigtimes_{i=1}^{\numasm}[0,1]$, 
% for any $\mu$ such that $\inftynorm{\mu-\priorMeanState}\leq \epsilon$
% we have that 
% \begin{align*}
% \objfunc(\util_{\mu},\marginalReport) 
% \geq \frac{1}{8}\OPT(\marginalReport, \statespace) - 3\epsilon.
% \end{align*}
% where $\util_{\mu}$ is the utility function for the optimal V-shaped scoring rule for the extended space~$\extstatespace$ according to mean $\mu$.
% \end{corollary}
Note that in the following theorem we show that the prior mean can be estimated efficiently using samples. 
\begin{theorem}\label{thm:sample}
For any $\epsilon>0, \delta>0$,
any distribution $\marginalReport$ with prior mean $\priorMeanState$ 
in state space $\statespace=\bigtimes_{i=1}^{\numasm}[0,1]$, 
letting $\mu$ be the empirical mean with $\frac{1}{\epsilon^2}\cdot \log \frac{n}{\delta}$ samples, 
with probability at least $1-\delta$, 
we have $\inftynorm{\mu-\priorMeanState} \leq \epsilon$.
\end{theorem}
\begin{proof}
By Chernoff-Hoeffding inequality, we have that 
for any sequence of $k$ independent random variables $\{r_i\}_{i=1}^k$ bounded in $[0,1]$ with the same mean $m$, we have
\begin{align*}
\Pr\left[\left|\frac{1}{k}\sum_{i=1}^k r_i - m\right|\geq \epsilon \right] \leq 2\exp(-2n\epsilon^2).
\end{align*}
Thus, with $\frac{1}{\epsilon^2}\cdot \log \frac{n}{\delta}$ samples, by union bound, we have that with probability at least $1-\delta$,
$\inftynorm{\mu-\priorMeanState} \leq \epsilon$.
\end{proof}
\paragraph{Remark:} In the proof of \cref{thm:sample}, we do not require the samples are drawn from i.i.d.\ distributions. 
Instead we only impose the constraint of independence with the same mean. 
This is particularly helpful 
if our estimate of the prior mean is from historical reports from different agents
as the distribution of reports may vary from agent to agent
as their abilities for acquiring information vary. 
However, all these distributions have the same mean by Bayesian plausibility.

Note that in the case the estimated mean is far from the prior mean, which occurs with probability at most $\delta$,
the loss in incentive for effort is at most $1$.
Combining \Cref{thm:infty norm error,thm:sample}, 
by setting $\delta=\epsilon$, we have the following corollary.
\begin{corollary}\label{cor:sample}
For any $\epsilon>0$,
any distribution $\marginalReport$ with prior mean $\priorMeanState$ 
in state space $\statespace=\bigtimes_{i=1}^{\numasm}[0,1]$, 
letting $\mu$ be the empirical mean with $\frac{1}{\epsilon^2}\cdot \log \frac{n}{\epsilon}$ samples, 
% with probability at least $1-\delta$, 
the expected incentive for effort of the V-shaped scoring rule for $\mu$ is at least that of the V-shaped scoring rule for $\mu_D$ less $4\epsilon$.
\end{corollary}

\section{Relating Eliciting the Full Distribution and Eliciting the Mean}
\label{sec:full dist}
%The previous discussions in this section focused on scoring rules for
%eliciting the mean of the posterior distribution.  Note that elicitation
%of the mean is a restriction on scoring rules and in general, the
%principal could solicit the full distribution and reward the agent
%accordingly.  
In this section, we will show that, with respect to
optimization and approximation, the problems of eliciting the full
posterior distribution over a finite state space can be reduced to problems
of eliciting the mean of a multi-dimensional state space.

We use the characterization of the proper scoring rule for eliciting the full distribution shown in \Cref{thm:characterization of report full dist}. Similar to \Cref{thm:interim utility on the boundary}, 
if the scoring rule is bounded, then the utility function $\util$ in \Cref{thm:characterization of report full dist} is bounded and continuous. 
The proof of continuity is the same as \Cref{thm:interim utility on the boundary} and hence omitted here.

Note that there is no function $\kappa(\state)$ in the characterization of \Cref{thm:characterization of report full dist}. 
The reason is that here for any finite state space $\statespace$, 
any scoring rule 
$\score(\reportFull, \state)
= \util(\reportFull) +\sg(\reportFull) \cdot (\state-\reportFull)
+ \kappa(\state),$
there exists another convex function $\hat{\util}$
such that 
$\score(\reportFull, \state)
= \hat{\util}(\reportFull) +\sg(\reportFull) \cdot (\state-\reportFull),$
where $\sg(\reportFull) \in \subgradient \hat{\util}(\reportFull)$ is a subgradient of $\hat{\util}$.
The objective value for reporting the full distribution 
with distribution $\marginalReport$ and scoring rule $\score$ is
\begin{align*}
\objfunc(\util,\marginalReport)
=\expect{\posterior \sim \distoverposterior, \state \sim \posterior}{\score(\reportFull, \state) - \score(\statedist, \state)}
=\int_{\reportspace}\left[ \util(\reportFull)-\util(\statedist)\right]
\distoverposterior(\posterior) \ \textrm{d}\posterior.
%=\int_{\reportspace} \left[ \util(\posterior)-\util(\statedist)\right]\marginalReport(\posterior) \ \textrm{d}\posterior.
\end{align*}

Thus the form of the objective function for reporting the full
distribution coincides with the objective function for reporting the
mean.  Moreover, it is easy to verify that the bounded constraint
coincides as well.  This result follows because distributions with
finite state space $\statespace$ can be viewed as
$|\statespace|$-dimensional perfectly negatively correlated
distributions with Bernoulli marginals.  One important
property of Bernoulli distributions is that reporting the full
distribution is equivalent to reporting the mean of the distribution.
Since reporting the full distribution and reporting the mean have the
same characterization in this case, by viewing the distribution as
$|\statespace|$-dimensional correlated distribution, we have the
following theorem.

% \begin{definition}\label{def:induced scoring rule for bernoulli}
% For any finite state space $\statespace = \{\state_i\}_{i=1}^{|\statespace|}$ 
% and any scoring rule 
% $\score: [0,1]^{|\statespace|}\times \statespace \to \reals$
% for eliciting the full distribution,
% the \emph{induced mean scoring rule} 
% $\hat{\score}: [0,1]^{|\statespace|}\times \{0, 1\}^{|\statespace|} \to \reals$ for eliciting the mean for Bernoulli distributions 
% satisfies $\hat{\score}(\report, \hat{\state}) = \score(\report, \state)$ 
% for any report $\report \in [0,1]^{|\statespace|}$, 
% state $\state \in \{0, 1\}^{|\statespace|}$
% and $\hat{\state}_i = 1$ if and only if $\state = \state_i$.
% \end{definition}
\begin{proposition}\label{thm:reduce full to mean}
For any finite state space $\statespace$, report space $\reportspace = \conv(\statespace)$,
and any distribution $\marginalReport \in \reportspace$ over posteriors, 
scoring rule $\score$ is optimal for eliciting the full distribution 
if and only if 
it is optimal for eliciting the mean.
\end{proposition}

\section{Computing the Optimal Scoring Rule}
\label{sec:computation efficiency for mean}

We adopt an approach from \citet*{BCKW-15} and show that when the state
space and the support of the posterior means are finite, there exists
a polynomial time algorithm that solves the optimal scoring rule for
eliciting the marginal means of a posterior.
\begin{theorem}\label{thm:computing optimal score for mean}
Given any $n$-dimensional state space $\statespace$ with
$|\statespace|=d$ states and any distribution~$\marginalReport$ with support
size $m$ over posterior means, there exists an algorithm that computes
the optimal proper bounded scoring rule for eliciting the mean in time
polynomial in $n$, $m$, and $d$.
\end{theorem}

To prove this theorem, we introduce a proposition stating the
equivalence of Bayesian auction design and the design of proper
scoring rules.  With this equivalence result, we can solve \Cref{eq:program} with finite reports using a linear program with
$(n+1)(m+d+1)$ variables and a quadratic number of constraints. 

\begin{proposition}\label{lem:alloc and pay}
A function $\utility$ is the utility function 
of a $\mu$-differentiable $\scorebound$-bounded proper scoring rule for eliciting the mean 
on report space $\reportspace=\conv(\statespace)$ 
and $n$-dimensional state space $\statespace$ 
if and only if there exists allocation and payment functions $\alloc(\cdot)$ and $\pay(\cdot)$ satisfying
\begin{enumerate}
\item Bayesian incentive compatible: 
$\alloc(\report)\cdot\report-\pay(\report)
\geq \alloc(\report')\cdot\report-\pay(\report')$,
for any report $\report, \report'\in \reportspace$;
\item bounded utility difference: 
$\alloc(\state)\cdot\state-\pay(\state)
\leq \scorebound+\alloc(\report)\cdot\state-\pay(\report)$,
for any report $\report \in \reportspace$
and state $\state\in\statespace$;
\item induced utility is 
$\utility(\report)=\alloc(\report)\cdot\report-\pay(\report)$
for any $\report \in \reportspace$.
\end{enumerate}
Note that the bounded utility difference property means the utility loss for misreporting $\report$ with true state $\theta$ is at most $\scorebound$.
% The induced scoring rule is $\score(\report, \state)=\state\alloc(\report)-\pay(\report)+\kappa(\state)$, with  bounded function $\kappa: \statespace\to \reals$.
\end{proposition}
\begin{proof}
  For the ``if'' direction:
if the allocation $\alloc$ and payment $\pay$ satisfies the above conditions, 
by \citet*{roc-85} and the Bayesian incentive compatibility, 
the utility function $\util$ is continuous and convex,
and $\xi(\report) = \alloc(\report)$ is a feasible subgradient of the utility function.
By the bounded utility difference, 
we have that 
\begin{align*}
\util(\state)-\util(\report) - \sg(\report)\cdot(\state - \report) 
&= \alloc(\state)\cdot\state-\pay(\state) 
- \alloc(\report)\cdot\report+\pay(\report)
- \alloc(\report)\cdot(\state-\report)\\
&= \alloc(\state)\cdot\state-\pay(\state) 
- \alloc(\report)\cdot\state+\pay(\report)
\leq \scorebound,
\end{align*}
which implies utility function $\util$ corresponds to a 
$\mu$-differentiable $\scorebound$-bounded proper scoring rule.

For the ``only if'' direction: given a utility function $\utility$
of a $\mu$-differentiable bounded proper scoring rule for eliciting the mean, 
by \Cref{thm:bounded_simple}, there exists a set of subgradients $\xi(\report) \in \partial \util(\report)$
such that
$$
\util(\state)-\util(\report) - \sg(\report)\cdot(\state - \report)\leq \scorebound 
% \qquad \forall \report\in \reportspace, \state\in \statespace.
$$
for any report $\report\in \reportspace$ and state $\state\in \statespace$. 
Setting the allocation as $\alloc(\report)=\sg(\report)$, 
and the payment as 
$\pay(\report)=
\report\cdot \sg(\report) - \utility(\report)$, 
it is easy to verify that this allocation and payment satisfy all three conditions above. 
% induced scoring rule is $\score(\report, \state)=\state\alloc(\report)-\pay(\report)+\kappa(\state)$. 
\end{proof}

\begin{proof}[Proof of \Cref{thm:computing optimal score for mean}]
  Denote the finite set of state space as
  $\statespace=\{\theta_j\}^{d}_{j=1}$, Let the support of
  distribution~$\marginalReport$ over posterior means be
  $\{\report_i\}^m_{i=1}$.  Denote the probability that posterior mean
  $\report_i$ happens as $\marginalReport_i$.  For simplicity, denote
  $\report_0=\priorMean$ as the mean of the prior and
  $\report_{m+j}=\state_j$ as the report for pointmass distribution on
  states for any $j\in[d]$.  \Cref{eq:program} is equivalent
  to the following program.
\begin{align}\label[scoringruleprogram]{eq:LPmean}
    \max_{\{\alloc_i, \pay_i\}_{i\in \{0,\dots, m+d\}}}\, &\sum_{i\in[m]}  (\alloc_i\cdot\report_i-\pay_i)\,\marginalReport_i\\
    \text{s.t.}\quad&
    \alloc_0\cdot\report_0-\pay_0 = 0, \nonumber\\
    &\alloc_i\cdot\report_i-\pay_i\geq \alloc_{i'}\cdot\report_i-\pay_{i'}, \qquad\quad\;\;\forall i, i'\in \{0,\dots, m+d\},\nonumber\\
    & (\alloc_i\cdot\report_i-\pay_i)-( \alloc_{i'}\cdot\report_i-\pay_{i'})\leq \scorebound,\, \forall i \in\{m+1,\dots, m+d\}, \nonumber\\ &\qquad\qquad\qquad\qquad\qquad\qquad\qquad\quad\forall i'\in \{0,\dots, m+d\}.\nonumber
\end{align}
Note that \Cref{eq:LPmean} is a linear program with number of variables and constraints polynomial in $n$, $m$, and~$d$; 
and hence there exists a polynomial time algorithm that optimally solves it. 
Next we will formally prove the equivalence of \Cref{eq:program} and \Cref{eq:LPmean}.

For one direction: For any utility function $\utility$ that is a
feasible solution to \Cref{eq:program}, by \Cref{lem:alloc and pay},
there exists corresponding allocation and payment functions $\alloc$
and $\pay$.  Let the variables in \Cref{eq:LPmean} be
$\alloc_i=\alloc(\report_i)$, $\pay_i=\pay(\report_i)$, for any $i\in
\{0,\dots, m+d\}$.  It is easy to verify that this is a feasible
solution to \Cref{eq:LPmean} with the same objective value.

For the other direction: For any feasible solution $\{\alloc_i, \pay_i\}_{i\in \{0,\dots, m+d\}}$ to \Cref{eq:LPmean}, define the utility function
\begin{equation*}
    \utility(\report)=\max_{i\in \{0,\dots, m+d\}} \alloc_i\cdot\report-\pay_i
\end{equation*}
for any report $\report \in \reportspace$. We show that this utility function $\util$ satisfies \Cref{eq:program} and has the same objective value.
Obviously, the utility function $\utility$ is continuous and convex.
For any $i\in \{0,\dots, m+d\}$,
% For any report $\report_i$ in the set $\{\report_i\}_{i=0}^{m+n}$, 
the utility function
$\utility(\report_i)=\alloc_i\cdot\report_i-\pay_i$
by the definition of Bayesian incentive compatibility, 
and hence the objective value of \Cref{eq:program} given by this utility $\utility$ equals the  objective value of \Cref{eq:LPmean}. 
% Then we are going to show that it satisfies the score bounded constraint. 
Moreover, for any report $\report \in \reportspace$, 
letting $i'=\argmax_{i \in \{0,\cdots,m+d\}}\alloc_{i}\report-\pay_{i}$,
the allocation $\alloc_{i'}$ is a subgradient of the utility function $\utility(\report)$ at report~$\report$. Thus, we have for any state $\statej\in\statespace$
\begin{align*}
\utility(\statej)-\utility(\report)-\xi(\report)\cdot(\statej-\report)
&=(\alloc_{m+j}\cdot\statej-\pay_{m+j})-(\alloc_{i'}\cdot\report-\pay_{i'})-\alloc_{i'}\cdot(\statej-\report) \\
&=(\alloc_{m+j}\cdot\statej-\pay_{m+j})-(\alloc_{i'}\cdot \statej - \pay_{i'}) \leq \scorebound,
\end{align*}
where the last inequality holds by the bounded utility difference
property.  Therefore, utility function $\util$ is a feasible solution
to \Cref{eq:program}, which establishes the equivalence of two
programs.
\end{proof}

In \Cref{sec:full dist}, we gave a reduction from the problem of
optimal scoring rules for eliciting the full distribution over a
finite state space to the problem of optimal scoring rules for
eliciting the marginal means over a multi-dimensional state space.  This
reduction is based on representing the state space by an indicator
vector.  In this section, we observe that the optimal scoring
rule can be found in polynomial time when the distribution of
posteriors is given explicitly.  This result is a simple corollary of
\Cref{thm:reduce full to mean} and \Cref{thm:computing optimal score
  for mean}.  
  
% Second, we show that even for single dimensional state
% space with finite size, the gap in performance between the optimal
% scoring rule for eliciting the mean and the optimal scoring rule for
% eliciting the full distribution is unbounded.

\begin{corollary}\label{thm:optimal computation}
Given any finite state space $\statespace$ with $|\statespace|=d$ 
and any distribution~$\marginalReport$ with support size $m$ over posteriors, 
there exists an algorithm that computes the optimal proper bounded scoring rule for eliciting the full distribution in time polynomial in $m$ and $d$.
\end{corollary}

\begin{proof}
  This result follows from combining \Cref{thm:reduce full to mean}
  (the reduction from full distribution reporting to reporting the
  mean) and \Cref{thm:computing optimal score for mean} (polynomial
  time computation of the optimal scoring rule for the mean).
\end{proof}

\section{Eliciting the Mean with an Expected Score Bound}
\label{apx:bound expect score}
In this section, we provide the optimal scoring rule for eliciting the single-dimensional mean under a boundedness constraint on the expected score. We consider the following optimization program: 

\begin{align}
\max_\score \qquad 
&\expect{\posterior \sim \distoverposterior, \state \sim \posterior}{\score(\mean{\posterior}, \state) - \score(\priorMean, \state)}\label[scoringruleprogram]{eq:program-expectedbound}\\
\text{s.t.}\qquad &\score \text{ is a proper scoring rule for eliciting the mean},\nonumber\\
&\score \text{ is non-negative
in space $\reportspace \times \statespace$},\nonumber\\
& \expect{\posterior \sim \distoverposterior, \state \sim \posterior}{\score(\mean{\posterior}, \state)} \text{is upper bounded by $B$.} \nonumber 
\end{align}

We consider this optimization program with restriction to canonical scoring rules. By Definition~\ref{def:canonical}, we write the single dimensional version of this optimization program as follows: 
\begin{align}
\max_\util \qquad &\int_{\reportspace} \util(\report)\marginalReport(\report)
\ \textrm{d}\report \label[scoringruleprogram]{eq:expect-budget-1d}\\
\text{s.t.} \qquad 
&\util\text{ is a continuous and convex function,}\nonumber\\ 
&\util(\priorMean) = 0,\nonumber\\
&\util(\report)
+ \util'(\report) \cdot (\state-\report) + \kappa(\state) \geq 0, 
\quad \forall \report \in [0,1], \state\in [0,1],\nonumber\\
&\expect{\posterior\sim\distoverposterior, \state\sim\posterior}{\util(r)+\kappa(\state)}\leq 1. \nonumber
\end{align}

\iffalse

\begin{definition}
A convex function $\util$ is cone-shaped at point $C$ if all the points are linearly interpolated between the center and the boundary. i.e. $\util(\alpha \cdot \report +(1-\alpha)\cdot C)=\alpha\util(\report)+(1-\alpha)\cdot\util(C)$, for any $\report\in\boundaryReport$.
\end{definition}

\fi

\begin{theorem}\label{thm:1d-expected-score-bound}
The optimal solution for \Cref{eq:program-expectedbound} is V-shaped.
\end{theorem}

To prove \Cref{thm:1d-expected-score-bound}, we show there is a feasible V-shaped utility function that gives the same objective.

\begin{proof}
Consider any feasible solution $\util$ of \Cref{eq:program-expectedbound}. We construct a V-shaped utility function $\tilde{\util}$ as follows:
\begin{align*}
    \tilde{\util}(\report) = 
    \begin{cases}
    -\int_0^{\priorMean} \util(x)\marginalReport(x)\mathrm{d}x / \int_0^{\priorMean} x\marginalReport(x)\mathrm{d}x \cdot (\report-\priorMean)& \text{for $\report \leq \priorMean$,}\\
    \int_{\priorMean}^1 \util(x)\marginalReport(x)\mathrm{d}x / \int_{\priorMean}^1 x\marginalReport(x)\mathrm{d}x  \cdot (\report-\priorMean)& \text{for $\report \geq \priorMean$.}
    \end{cases}
\end{align*}
This V-shaped utility function $\tilde{\util}$ has the same objective value as the utility function $\util$. We then show that this V-shaped utility function $\tilde{\util}$ is a feasible solution of \Cref{eq:program-expectedbound}. 

It is easy to see that $\tilde{\util}$ is a continuous and convex function and $\tilde{\util}(\priorMean) = 0$. 
We now show that there exists a function $\kappa$ such that the scoring rule defined by $\tilde{\util}$ and $\kappa$ is bounded in expectation and non-negative in space $\reportspace \times \statespace$.  
Since $\util$ is a feasible solution, there exists a function $\kappa$ such that function $\util$ and $\kappa$ satisfies constraints in \Cref{eq:expect-budget-1d}. Thus, we have for any $\theta \in [0,1]$
$$
\kappa(\theta) \geq \max_{r \in [0,1]} \{-u(r)-u'(r)\cdot(\theta-r)\}. 
$$
Since function $\util$ is convex, we have for any $\theta \in [0,1]$
$$
\max_{r \in [0,1]} \{-u(r)-u'(r)\cdot(\theta-r)\} = \max \{-u(1)-u'(1)\cdot(\theta-1), -u(0)-u'(0)\cdot\theta\}.
$$

We then show that $\util(0) \geq \tilde{\util}(0)$ and $\util'(0) \leq \tilde{\util}'(0)$. Note that the V-shaped utility function satisfies $\int_0^{\priorMean} \util(r)f(r)\mathrm{d}r = \int_0^{\priorMean} \tilde{\util}(r)f(r)\mathrm{d}r$.
If $\util(0) < \tilde{\util}(0)$, then by the convexity of function $u$, we have for any $\report \in [0,\priorMean]$ 
$$
u(r) \leq (1-\report/\priorMean)\util(0) <  (1-\report/\priorMean)\tilde{\util}(0) = \tilde{\util}(r),
$$
which contradicts with $\int_0^{\priorMean} \util(r)f(r)\mathrm{d}r = \int_0^{\priorMean} \tilde{\util}(r)f(r)\mathrm{d}r$. If $\util'(0) > \tilde{\util}'(0)$, then by the convexity of function $u$, we have $\util(r) > \tilde{\util}(r)$ for any $\report \in [0,\priorMean]$, which also contradicts with $\int_0^{\priorMean} \util(r)f(r)\mathrm{d}r = \int_0^{\priorMean} \tilde{\util}(r)f(r)\mathrm{d}r$.

Similarly, we have $\util(1) \geq \tilde{\util}(1)$ and $\util'(1) \geq \tilde{\util}'(1)$. Thus, we have for $\theta \in [0,\priorMean]$
$$
-\tilde{\util}(1) - \tilde{\util}'(1)\cdot (\theta - 1) = -\tilde{\util}'(1) \cdot(\theta-\priorMean) \leq -u'(1)\cdot(\theta-\priorMean) \leq -u(1)-u'(1)\cdot(\theta-1),
$$
where the last inequality is due to the convexity of function $u$. 
Similarly, we have for $\theta \in [\priorMean,1]$
$$
-\tilde{\util}(0) - \tilde{\util}'(0)\cdot \theta = -\tilde{\util}'(0) \cdot(\theta-\priorMean) \leq -u'(0)\cdot(\theta-\priorMean) \leq -u(0)-u'(0)\cdot\theta.
$$
Combining these two inequalities, we derive that
$$
\max \{-\tilde{\util}(1)-\tilde{\util}'(1)\cdot(\theta-1), -\tilde{\util}(0)-\tilde{\util}'(0)\cdot\theta\} \leq \max \{-u(1)-u'(1)\cdot(\theta-1), -u(0)-u'(0)\cdot\theta\},
$$
which means the same function $\kappa$ also satisfies constraints of \Cref{eq:expect-budget-1d} for the V-shaped utility function $\tilde{\util}$. Therefore, this V-shaped function $\tilde{\util}$ is also a feasible solution, which completes the proof.
\end{proof}

% \newpage
% \input{Notes/notes}

% \printbibliography
\end{document}